\documentclass[prx,twocolumn,english,superscriptaddress,floatfix,longbibliography]{revtex4-2}

\usepackage[utf8]{inputenc}
\usepackage[T1]{fontenc}
\usepackage{lmodern}

\usepackage{adjustbox}
\usepackage[ruled,vlined,linesnumbered]{algorithm2e}
\usepackage{amssymb,amsmath,amsthm,amsfonts}
\usepackage{bbm}
\usepackage{braket}
\usepackage{comment}
\usepackage{enumitem}
\usepackage{float}
\usepackage{footnote}
\usepackage[margin=1in]{geometry}
\usepackage{graphicx}
\usepackage[colorlinks]{hyperref}
\usepackage{longtable}
\usepackage{mathtools}
\usepackage{mathrsfs}
\usepackage{outlines}
\usepackage{physics}
\usepackage{stmaryrd}
\usepackage{textgreek}
\usepackage{tikz}
\usetikzlibrary{quantikz}
\usepackage{xcolor}
\usepackage[noend]{algpseudocode}

\newtheorem{theorem}{Theorem}
\newtheorem{definition}{Definition}

\newtheorem{proposition}{Proposition}
\newtheorem{example}{Example}

\newcommand{\eqn}[1]{\hyperref[eqn:#1]{EQ~\ref*{eqn:#1}}}
\newcommand{\thm}[1]{\hyperref[thm:#1]{Theorem~\ref*{thm:#1}}}
\newcommand{\cor}[1]{\hyperref[cor:#1]{Corollary~\ref*{cor:#1}}}
\newcommand{\defn}[1]{\hyperref[defn:#1]{Definition~\ref*{defn:#1}}}
\newcommand{\lem}[1]{\hyperref[lem:#1]{Lemma~\ref*{lem:#1}}}
\newcommand{\prop}[1]{\hyperref[prop:#1]{Proposition~\ref*{prop:#1}}}
\newcommand{\fig}[1]{\hyperref[fig:#1]{FIG~\ref*{fig:#1}}}
\newcommand{\tab}[1]{\hyperref[tab:#1]{Table~\ref*{tab:#1}}}
\newcommand{\algo}[1]{\hyperref[algo:#1]{Algorithm~\ref*{algo:#1}}}
\renewcommand{\sec}[1]{\hyperref[sec:#1]{Section~\ref*{sec:#1}}}
\newcommand{\append}[1]{\hyperref[append:#1]{Appendix~\ref*{append:#1}}}
\newcommand{\fac}[1]{\hyperref[fac:#1]{Fact~\ref*{fac:#1}}}
\newcommand{\lin}[1]{\hyperref[lin:#1]{Line~\ref*{lin:#1}}}
\newcommand{\fnote}[1]{\hyperref[fnote:#1]{Footnote~\ref*{fnote:#1}}}

\newcommand{\rmk}[1]{\hyperref[rmk:#1]{Remark~\ref*{rmk:#1}}}

\def\>{\rangle}
\def\<{\langle}

\newcommand{\x}{\ensuremath{\mathbf{x}}}

\newcommand{\Q}{\mathbb{Q}}

\newcommand{\M}{\mathcal{M}}
\renewcommand{\H}{\mathcal{H}}
\renewcommand{\O}{\mathcal{O}}

\DeclareMathOperator{\diag}{diag}

\def\:{\hbox{\bf:}}

\def \eps {\epsilon}

\let\oldnl\nl
\newcommand{\nonl}{\renewcommand{\nl}{\let\nl\oldnl}}

\hypersetup{
    colorlinks=true,
}

\let\oldnl\nl

\begin{document}
\title{Implementation and Learning of Quantum Hidden Markov Models}

\author{Vanio Markov}
\affiliation{Wells Fargo}
\author{Vladimir Rastunkov}
\affiliation{IBM Quantum, IBM Research}
\author{Amol Deshmukh}
\affiliation{IBM Quantum, IBM Research}
\author{Daniel Fry}
\affiliation{IBM Quantum, IBM Research}
\author{Charlee Stefanski}
\affiliation{Wells Fargo}

\normalsize

\date{\today}

\begin{abstract}
In this article, we use the theory of quantum channels and open quantum systems to provide an efficient unitary characterization of a class of stochastic generators known as quantum hidden Markov models (QHMMs). By utilizing the unitary characterization, we demonstrate that any QHMM can be implemented as a quantum circuit with mid-circuit measurement. We prove that QHMMs are more compact and more expressive definitions of stochastic process languages compared to the equivalent classical hidden Markov models (HMMs). 
Starting with the formulation of QHMMs as quantum channels, we employ Stinespring's construction to represent these models as unitary quantum circuits with mid-circuit measurement.
By utilizing the unitary parameterization of QHMMs, we define a formal QHMM learning model. The model formalizes the empirical distributions of target stochastic process languages, defines hypothesis space of quantum circuits, and introduces an empirical stochastic divergence measure - hypothesis fitness - as a success criterion for learning. We demonstrate that the learning model has a smooth search landscape due to the continuity of Stinespring's dilation. The smooth mapping between the hypothesis and fitness spaces enables the development of efficient heuristic and gradient descent learning algorithms.

We propose two practical learning algorithms for QHMMs. The first algorithm is a hyperparameter-adaptive evolutionary search. The second algorithm learns the QHMM as a quantum ansatz circuit using a multi-parameter non-linear optimization technique.

\end{abstract}

\maketitle
\hypersetup{linkcolor=black}
\tableofcontents
\newpage
\thispagestyle{empty}


\section{Introduction}\label{sec:intro}

In this article we study quantum models of discrete stochastic processes.
We assume that the processes are generated by \textit{finitary} mechanisms with restricted memory and time.
Furthermore, we assume that only subset of the process variables are observable, 
and the distributions of these observations are dependent on an underlying unobservable process. A common modeling approach for such situations is to use two joint stochastic processes: a latent state process, and a dependent output emission process. The state process is considered a stochastic Markovian process \cite{karlin2012first} with a finite number of "hidden" states. The observable process is represented as a probabilistic map from the hidden states to the symbols of a finite alphabet. Such a modeling framework is referred to as a \textit{hidden Markov model} (HMM) \cite{mamon2007hidden}.

Hidden Markov models are applicable in various real-world situations.
The concept of latent variables is used to estimate circumstances and relations which are not directly observable. The assumption of underlying Markovian dynamics allows for the development of tractable learning and inference procedures. These models were introduced by Baum and Petri \cite{baum1966statistical} in 1966, and since then have been
successfully applied in areas such as linguistics \cite{rabiner1989tutorial,anandika2021review,Lefevre2003,KJK2018}, DNA analysis
\cite{porter2021profile,YPR2018,KBS1994}, engineering \cite{Jong2004,UAK2018}, and finance \cite{mamon2007hidden,UAK2019}.

The \textit{learning problem} for HMMs is to estimate the transition and emission operators of an unknown HMM given a finite sample of its emissions distributions. The conventional approach to solving the learning problem is based on maximization of observations' likelihood
\cite{baum1966statistical,baum1970maximization,baum1972inequality,bickel1998asymptotic,KNW2013}. Despite the fact that the expectation-maximization cannot be performed in polynomial time \cite{terwijn2002learnability}, this approach has been widely successful in practice \cite{ghojogh2019hidden}. Another group of learning algorithms, known as \textit{spectral algorithms} \cite{anderson1999realization,hsu2012spectral, balle2014spectral} have polynomial time complexity. These algorithms represent the joint frequencies of the observed sequences as a matrix known as a Hankel matrix \cite{anderson1999realization}.
The rank of the matrix is used as an estimate of the order of the model. The model parameters are mapped to the factors of matrix's singular value decomposition. These learning algorithms belong to a specialized area known as \textit{realization problem} for HMMs \cite{anderson1999realization,hsu2012spectral}.

The objective of this research is to propose a methodology for the physical implementation of quantum HMMs, as well as to develop practical learning algorithms for these models. Our approach is to \textit{quantize} a classical HMM by replacing its stochastic vector space by the space of quantum density operators. The classical transition and emission operators are represented as \textit{observable operators} \cite{jaeger2000observable} and replaced by \textit{quantum operations}.
 A quantum operation, also known as \textit{quantum channel}, is a linear, completely positive trace-preserving map (CPTP map) defined on the space of density operators \cite{nielsen_chuang_2010}. When a quantum operation is used to quantize a classic HMM, it has two roles. Firstly, it defines the stochastic changes of model's state, representing the underlying Markovian process. Secondly, when interpreted as a \textit{general measurement} (POVM) operator, the quantum operation describes the observable process. 

The definition of a quantum hidden Markov model (QHMM) as a quantum operation was initially introduced by Monras et al. \cite{monras2011hidden}. Since then, various aspects of QHMMs were studied by Srinivasan et al. \cite{Srinivasan2017learningHQMM,Srinivasan2018LearningInference,Adhikary2019LearningQuantumGraphicalModels,Adhikary2020Expressiveness,Srinivasan2021Towards,Srinivasan2020QuantumTensorNetworks,srinivasan2022quantum}, O'Neill \cite{ONeill2012hqmm_one_qubit}, and Clark \cite{clark2015hidden}. Javidian \cite{javidian2021learning} discussed a different type of model called the circular QHMM. Cholewa et al. \cite{cholewa2017quantum} studied HMMs based on quantum walks and transition operation matrices. Elliot \cite{Elliott2021Compress} uses the concept of \textit{memory compression} to demonstrate that the quantum HMMs achieve strict advantage compared to the corresponding classical HMMs. The memory compression is related to the size of the models' state-space and the amount of information it stores.
Furthermore, other quantum approaches to sequence modeling are being proposed. For instance, Blank et al. \cite{blank2021quantum} investigated characteristic function of discrete stochastic processes and proposed quantum algorithm for learning it, leveraging techniques such as quantum amplitude estimation and quantum Monte Carlo methods.

To implement QHMMs on quantum computers and develop feasible learning procedures, we employ the open quantum systems theory \cite{nielsen_chuang_2010} and utilize Stinespring's dilation theorem \cite{Stinespring1955} to derive unitary representation of the models. According to this theorem, every quantum operation can be represented as unitary evolution in an extended Hilbert space, followed by disregarding some degrees of freedom. We use this construction to parameterize any QHMM, defined by a quantum operation, as a unitary transformation in a larger Hilbert space. This parametrization enables the physical implementation of QHMMs as unitary quantum circuits with mid-circuit measurements. Furthermore, we exploit the continuity of Stinespring's procedure \cite{KRETSCHMANN20081889}, to define a hypothesis space and corresponding learning objective, enabling efficient learning of QHMMs.

The main contributions of our study are
\begin{itemize}
    \item We demonstrate that QHMMs are more \textit{compact generators} of stochastic process languages compared to classical hidden Markov models. We show that if a QHMM generates a stochastic language of finite rank $r$, it requires only a Hilbert space of dimension at most $N = \sqrt{r} $. In contrast, the minimal classical HMM for such a language, if the model exists, requires at least $n=r$-dimensional classic vector space.

     \item We demonstrate that QHMMs are more \textit{expressive generators} of stochastic process languages compared to classical hidden Markov models. For every stochastic language with finite rank $r$, there exists QHMMs $\mathbf{Q}$ in Hilbert space with dimension $N = \lfloor \sqrt{r} \rceil$. It is well known \cite{Dharm63, Fox68}, that there exist finite-rank stochastic languages which cannot be realized by classical QHMMs.
   
    \item For every QHMM defined as a quantum operation we provide an equivalent QHMM implementation as a unitary
    circuit with mid-circuit measurement.
   
    \item We provide a formal QHMM learning model formalizing the training data sample, hypothesis space of unitary
    quantum circuits, and hypothesis quality criteria as stochastic divergence between two distributions.
    
    \item We propose a tractable empirical distance measure between two QHMMs in terms of the divergence of their observable distributions. We prove that the empirical distance is dominated by the distance measure induced on the space of QHMMs by the diamond norm.  
    
    \item We demonstrate that the landscape of the proposed QHMM learning model is smooth, which follows from the convergence properties of Stinespring's dilation theorem \cite{KRETSCHMANN20081889}. The smooth landscape implies efficient heuristic and gradient-based learning algorithms.
 
    \item We specify a QHMM learning algorithm as adaptive evolutionary search in the space of quantum circuits
    with mid-circuit measurement.
 
    \item We specify a QHMM learning algorithm as multi-parameter nonlinear optimization in the space of pre-defined
    quantum ansatz.
\end{itemize}
 
This article is organized as follows:
\begin{itemize}
    \item Section \ref{section:quantum_hidden_markov_models} presents some preliminary information regarding stochastic
    process languages, classic hidden Markov models and observable operators models. It introduces critical concepts including quantum states,
    quantum operations, general measurement and quantum mid-circuit computation.
   
    \item Section \ref{section:QHMM} presents the basic definition of a QHMM as quantum operation and provides unitary definition of QHMMs and their implementation as
    quantum circuits with mid-circuit measurement.
 
    \item Section \ref{section:learning_quantum_hidden_markov_models} discusses a formal QHMM learning model
    defining data samples, hypothesis space and learning criterion. It specifies adaptive evolutionary and ansatz-based QHMM learning
    algorithms.
 
    \item Section \ref{section:examples_of_qhmms} presents two examples: a market model and a stochastic volatility
    model.
 
    \item Section \ref{section:conclusions_outlook} discusses our conclusion and outlook.
 
\end{itemize}

\section{\label{section:quantum_hidden_markov_models}Preliminaries}

\subsection{\label{subsubsection:stochastic_and_hmms}Stochastic Languages and Hidden Markov Models}

We study a class of discrete time stationary stochastic processes,
\begin{equation*}
    \{ Y_t : t \in \mathbb{N}, Y_t \in\Sigma \},
\end{equation*}
where $\Sigma=\{a_1, \ldots, a_m\}$ is a finite set of observable symbols, called an \textit{alphabet}.
The set of all finite sequences over the alphabet $\Sigma$, including the empty sequence $\epsilon$, is denoted by ${\Sigma}^{*}$.
Finite sequences of symbols are denoted by the bold lowercase letters $\textbf{a},\, \textbf{p},\, \textbf{s}$.
The set of all sequences with length exactly $t$ is denoted by ${\Sigma}^{t}$.
For any $\textbf{a} \in {\Sigma}^{*} $ let $|\textbf{a}|$ be the number of symbols in the sequence.
If \textbf{p} and \textbf{s} are sequences, then \textbf{ps} denotes their \textit{concatenation}.
The sequence \textbf{p} is called the \textit{prefix} and \textbf{s} is called the \textit{suffix} of \textbf{ps}.
Any subset of ${\Sigma}^{*}$ is a \textit{language} ${L}$ over the alphabet.
The sequences belonging to a language are referred to as \textit{words}.
The set of sequences originated by observations or measurement of the evolution of a discrete-time process is called
the \textit{process language}.
It is easy to verify that if a word results from the observation of a process, then every one of its subwords has also been observed. Therefore, the process languages are \textit{subword-closed}.

A \textit{stochastic} process language $L$ is a process language together with a set 
\begin{equation}
\label{eqn:stochstic language distribution_0}
 D^L=\bigl \{D_t^L : t \geq 0 \bigr \} 
\end{equation}
\noindent
of finite dimensional probability 
distributions $D_t^L$, each of which is defined on the sequences with length exactly $t$:

\begin{equation}
\label{eqn:stochstic language distribution}
    D_t^{L} = \bigl \{ P[\textbf{a}] : \textbf{a} \in {\Sigma}^{t}, \sum_{\textbf{a} \in {\Sigma}^{t} } P[\textbf{a}] = 1  \bigr \}.
\end{equation}

The following properties of a stochastic process language are straightforward to prove:

\begin{itemize}
    \item If the probability of a sequence $\textbf{p}$ is $P[\textbf{p}]$, then the conditional probability to
    observe a symbol $a$ after observing $\textbf{p}$ is
    \begin{equation*}
    P[a\mid\textbf{p}]=\frac{P[\textbf{p}a]}{P[\textbf{p}]}.
    \end {equation*}
    \item Every stochastic language $L$ and sequence of symbols $\textbf{p}\in {\Sigma}^{*}$ define collections of
    distributions $P_{t}[ \textbf{s} \mid \textbf{p}]$  over the sequences with equal length $t > 0$:
    \begin{equation*}
        \sum_{\textbf{s} \in {\Sigma}^{t} } P[\textbf{ps}] = P[\textbf{p}], \forall\textbf{p}\in {\Sigma}^{*}.
    \end{equation*}
\end{itemize}
    
Two sequences $\mathbf{p}_1$ and $\mathbf{p}_2$ are called \textit{t-equivalent} with respect to a stochastic language $L$ if they define the same distribution
    over ${\Sigma}^{t}$:
    \begin{equation*}
     P_{t}[\textbf{s} \mid \textbf{p}_1] = P_{t}[\textbf{s} \mid \textbf{p}_2], \forall \textbf{s} \in {\Sigma}^{t}.
    \end{equation*}

We study observable processes $\{ Y_t: Y_t \in {\Sigma} \}$ where each observation $Y_t$ depends on the state of a non-observable or \textit{hidden} finite-state process $\{ X_t : X_t \in S \}$, where $S =\{s_1, \ldots, s_n\}$.
The processes $\{ Y_t \}$ and $\{ X_t \}$ are called \textit{emission process} and \textit{state process} respectively.
If the joint process $\{ Y_t, X_t : t \in \mathbb{N},\,Y_t \in {\Sigma},\, X_t \in S \}$ is assumed to be stationary,
then we can describe it by a linear model known as (classical) hidden Markov model (HMM).
The model assumes \textit{Markovian} evolution of the hidden state process $\{ X_t \}$.
The observed process $\{ Y_t\}$ is emitted by a linear stochastic operator mapping the current hidden state to an observable
symbol.

\begin{definition}[Classical Hidden Markov Model]
    \label{def_chmm}
    A classical hidden Markov model is defined as a 5-tuple:
    \begin{equation*}
        \textbf{M}= \bigl\{ \Sigma,S,A,B, \x_0 \bigr\}
    \end{equation*}
    where $\Sigma=\{a_1, \ldots, a_m\}$ is a finite set of observable symbols, $S=\{s_1, \ldots, s_n\}$ is a finite set
    of unobservable or hidden states and the number of states $n$ is called the \textit{order of the model}, $A$ is
    a row-stochastic state transition matrix, $B$ is a column-stochastic observations emission matrix, and $\x_0$ is a
    stochastic vector defining the initial superposition of the process' states.
\end{definition}
At any moment in time $t$, the model is in a superposition (stochastic mixture) of its hidden states, described by a stochastic vector $x_t \in \mathbb{R}^n$. The component $x_t^i$ represents the probability of being in state $s_i$. The evolution of the hidden states follows a Markovian process defined by the transition matrix:
\begin{equation*}
    x_{t+1} = Ax_t.
\end{equation*}

At any moment in time $t$, the model defines symbol emission probabilities using a stochastic vector $y_t \in \mathbb{R}^m$, which depends on the current state $x_t$ through the emission matrix $B$:
\begin{equation*}
    y_{t} = Bx_t.
\end{equation*}
The component $y_t^i$ represents the probability of emitting symbol $a_i$. However, it's important to note that the evolution of the observation process, in general, is not a Markovian process.
 
The \textit{steady state} of a stationary HMM is a stochastic vector $\textbf{x}^{*}$ defined as
\begin{equation*}
    \textbf{x}^{*} = A\textbf{x}^{*}.
\end{equation*}
The steady state $\textbf{x}^{*}$ is the eigenvector of the state transition operator $A$ with eigenvalue 1,
normalized to represent a probability distribution.

For every HMM $\mathbf{M}$ we can define a set of \textit{observable operators} \textbf{T} as
\begin{equation}
\label{classic observable operators}
    \textbf{T} = \{ T_a : T_a=AB_a, a \in\Sigma \},
\end{equation}
where $B_a = \diag(B[a,i]), i \in 1,\dots,n$ are diagonal matrices defining symbol observation probabilities for each
state $s_i$ ~\cite{jaeger, carlyle_paz}.

Every element of an observable operator $T_a$ defines the conditional probability of the model being in state $s_j$ if, at the previous moment in time, it was in state $s_i$ and the symbol $a$ was emitted.
\begin{equation*}
    T_a^{i,j} = P[s_j \mid s_i, a].
\end{equation*}
Let $\textbf{a}={a_1} \ldots {a_t}$ be any sequence.
We define the observable operator corresponding to $\textbf{a}$ as follows: 
\begin{equation*}
    T_\textbf{a} = T_{a_t} \ldots T_{a_1}.
\end{equation*}
For any sequences $\textbf{a}, \textbf{b} \in {\Sigma}^{*} $ the following composition of the observable operators is easy to verify:
\begin{equation*}
    T_\textbf{ab} = T_\textbf{b}T_\textbf{a}.
\end{equation*}
The state transition operator $A$ is related to the observable operators $\{T_a\}$ as follows:
\begin{equation*}
    A = \sum_{a \in {\Sigma}} T_a.
\end{equation*}

A Hidden Markov Model (HMM) $\mathbf{M}$ defines the probability of every finite observable sequence $\mathbf{a} = a_1 \ldots a_t$ as:
\begin{equation}
 \label{classic_observable_prob}  
    P[\textbf{a} \vert \mathbf{M}]=\textbf{1}T_{a_t} \ldots T_{a_1}\textbf{x}_0=\textbf{1}T_{\textbf{a}}\textbf{x}_0,
\end{equation}
where \textbf{1} is the all-ones row vector with dimension $n$.
With every model $\mathbf{M}$ we associate a \textit{sequence function} $f^{\mathbf{M}} : {\Sigma}^{*} \rightarrow \left[ 0, 1 \right] $ defined as:
\begin{equation}
    \label{eqn:hmm_sequence_function}
    {f^{\mathbf{M}}}(\textbf{a})=P[\textbf{a} \vert \mathbf{M}],\forall \textbf{a} \in {\Sigma}^{*}.
\end{equation}

Through its sequence function, every HMM $\mathbf{M}$ defines a stochastic process language $L^M$ consisting of a set of distributions for each $t \in \mathbb{N}$:
\begin{equation}
\label{eqn: D_M^T classic}
    D^{\mathbf{M}}_t = \{ {f^\mathbf{M}}(\textbf{a}) : \textbf{a} \in {\Sigma}^{t} \}
\end{equation}

We compare HMMs based on the languages they define or generate. Two HMMs are considered \textit{equivalent} if they define the same stochastic process language.

Let us assume that we have observed a process up to a certain moment in time $t$:
\begin{equation*}
    \textbf{a}={a_1}\ldots{a_t}.
\end{equation*}
The probability of the next symbol being $a$ is given by:
\begin{equation*}
    \begin{split}
    P[a\vert\textbf{a}] & =\frac{P[\textbf{a}a]}{P[\textbf{a}]} \\
    & = \textbf{1}T_a \frac{T_{\textbf{a}}\textbf{x}_0}{\textbf{1}T_{\textbf{a}}\textbf{x}_0} \\
    & = \textbf{1}T_a\textbf{x}_t,
    \end{split}
\end{equation*}
where
\begin{equation*}
    \textbf{x}_t = \frac{T_{\textbf{a}}\textbf{x}_0}{\textbf{1}T_{\textbf{a}}\textbf{x}_0}
\end{equation*}
is the stochastic mixture of states after the emission of the sequence ${\mathbf{a}}$.

For every HMM $\mathbf{M}$, let's consider a \textit{bi-infinite} matrix, commonly referred to as a generalized \textit{Hankel matrix}, defined by the sequence function (\ref{eqn:hmm_sequence_function}) of the model \cite{anderson1999realization}:
\begin{equation}
\label {eqn:hankel matrix}
    H^\mathbf{M} \in {R}^{{\Sigma}^{*} \times {\Sigma}^{*}}, {H^\mathbf{M}} \left [ \textbf{p},\textbf{s} \right ] = f^\mathbf{M} (\textbf{p}\textbf{s}), \forall\textbf{p},\textbf{s}\in {\Sigma}^{*},
\end{equation}
where $\textbf{p}\textbf{s}$ denotes a concatenation of sequences $\textbf{p}$ and $\textbf{s}$.
We can say that $H^\mathbf{M}$ is indexed by the prefix $\mathbf{p}$ and suffix $\mathbf{s}$ of the sequence $\mathbf{p}\mathbf{s}$. The rows of $H^\mathbf{M}$ are indexed by the prefixes $\mathbf{p}$ in the \textit{first lexicographical order}, and the columns are indexed by suffixes $\mathbf{s}$ in the \textit{last lexicographical order} of the sequences~\cite{anderson1999realization}.
A sequence function $f^L$ and the corresponding generalized Hankel matrix $H^L$ are defined for every stationary stochastic language $L$ by the distributions $D^{L}_t$ (\ref {eqn:stochstic language distribution}):
\begin{equation}
\label{eqn:stochastic language seq function}
    {f^{\mathbf{L}}}(\textbf{a})=P[\textbf{a} \vert D^{L}_t],\forall \textbf{a} \in {\Sigma}^{t}, \forall t \in \mathbb{N}.
\end{equation}
\noindent
The \textit{rank} of a stochastic language is the rank of its Hankel matrix. This matrix is important, because its finite-rank property is a necessary condition for the stochastic language to have a HMM realization. 

\begin{theorem}[Anderson~\cite{anderson1999realization}]
 \label {theorem: Anderson}
    If $H$ is the infinite generalized Hankel matrix associated with a HMM $\mathbf{M}$ of order $n$, then $rank(H) \leq n$.
\end{theorem}

It has been proven \cite{huang}, that the equality in \textbf{Theorem}  \ref{theorem: Anderson} can be reached, i.e. there exist minimal order $n$ HMMs with $rank(H) = n$. If a HMM $M$ has order $n$ its finite dimensional distributions $D^{\mathbf{M}}_t$ are completely defined by the values of the sequence function $f^\mathbf{M}(\mathbf{a})$ for all sequences $\mathbf{a}$ of length at most $2n-1$~\cite{carlyle_paz, jaeger}.

 In general, not every finite-rank stochastic language has a HMM~\cite{vidyasagar}. Determining whether a stochastic language $L$ has finite rank is an undecidable problem \cite{sontag}. 

    \begin{table*}[ht]
        \begin{tabular}{c|l}
            \hline
            \hspace*{1.5mm}{$s$}\hspace*{1.5mm} &  \hspace*{1.5mm}Description  \\
            \hline
            0 &   \hspace*{1.5mm}Bear - Tendency down \\
            1 &   \hspace*{1.5mm}Bull - Tendency up \\
            2 &   \hspace*{1.5mm}Transition to Bear \\
            3 &   \hspace*{1.5mm}Transition to Bull \\
            \hline
        \end{tabular}
        \hfill
        \begin{tabular}{|c|c|c|c|c|}
            \hline
            \hspace*{1.5mm}{$s$}\hspace*{1.5mm} & \hspace*{3.5mm}{0}\hspace*{3.5mm} & \hspace*{3.5mm}{1}\hspace*{3.5mm}
            & \hspace*{3.5mm}{2}\hspace*{3.5mm} & \hspace*{3.5mm}{3}\hspace*{3.5mm}  \\
            \hline
            0 &  0.5&	0.1&	0.15&	0.25 \\
            1 &  0.1&	0.5&	0.25&	0.15 \\
            2 &  0.25&	0.15&	0.5&	0.1 \\
            3 &  0.15&	0.25&	0.1&	0.5 \\
            \hline
        \end{tabular}
        \hfill
        \begin{tabular}{|c|c|c|}
            \hline
            \hspace*{1.5mm}{$s$}\hspace*{1.5mm} & {Symbol \textbf{0}}& {Symbol \textbf{1}}  \\
            \hline
            0 &  0.8&	0.2 \\
            1 &  0.2&	0.8 \\
            2 &  0.4&	0.6 \\
            3 &  0.6&	0.4 \\
            \hline
        \end{tabular}
        \caption{Left: Market hidden states descriptions. Center: State transition probabilities. Right: Observed symbol probabilities.}
        \label{tab:hmm_example}
    \end{table*}

\begin{figure*}[ht]
        \centering
        \begin{minipage}{.4\textwidth}
            \centering
            \includegraphics[width=0.7\linewidth]{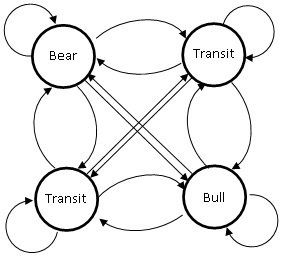}
        \end{minipage}
        \begin{minipage}{.59\textwidth}
            \centering
            \includegraphics[width=1.0\linewidth]{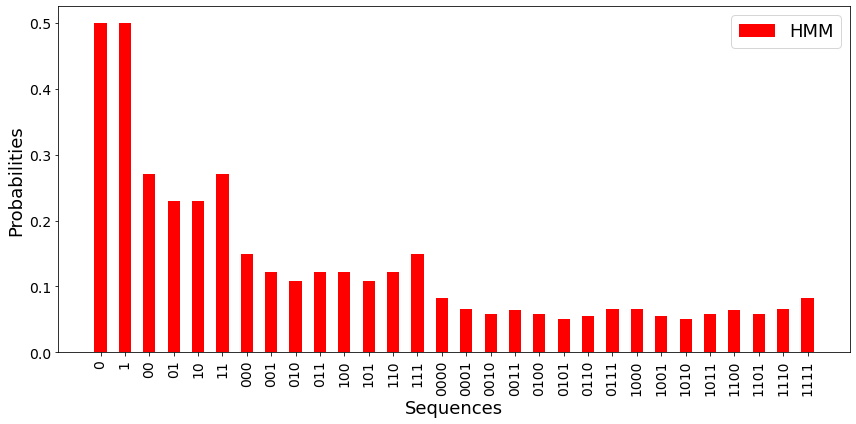}
        \end{minipage}
 
        \caption{Left: State transition graph. Right: Distribution of observed symbol sequences.}
        \label{fig:hmm_example}
    \end{figure*}

\begin{example}
    \label{example1}
    \normalfont
The evolution of asset prices on today's exchanges is driven by the interaction of ask and bid orders placed
on an electronic two-sided queue known as a \textit{Limit Order Book} (LOB) \cite{bouchaud_2018}.
One important statistic associated with the LOB is the \textit{mid-price}. The mid-price is the mean of lowest ask order price and highest bid order price at any given moment. In this example we present very simple HMM of the discrete-time mid-price direction of change.

    At any given moment, the symbol $\textbf{1}$ is observed if the price has increased, while the symbol $\textbf{0}$ is observed if the price has decreased or remained unchanged. It is assumed that the distribution of observed directions of change depends on the underlying state of the market, which is not directly observable. Let us consider a very simple situation, where the market can be in one of four states: bear, bull, transition to bear, or transition to bull, as listed in Table~\ref{tab:hmm_example}.
    The distribution of observation sequences is shown in \fig{hmm_example}.
    The generalized Hankel matrix associated with the model has a rank of 4 for sequences containing up to 20 symbols. Therefore, we conclude that the order of this model is 4.

\end{example}

\subsection{\label{subsubsection:qstates_qoperations}Quantum States and Quantum Operations}

We study quantum models of discrete-time joint stochastic processes. These models assume two processes: an underlying, hidden Markov process that evolves in discrete time according to linear stochastic dynamics, and a second, observable process that produces an observation at each time step by applying a stochastic map to the hidden state. To quantize this model, it is necessary to define quantum concepts  that correspond to its states, state evolution, observables, and the dynamics of the measurement process. 
The states of the process will be represented by the states of an $n$-qubit quantum system, where the associated complex Hilbert space $\H$ ~\cite{nielsen_chuang_2010} is a linear vector space with an inner product and dimension $N=2^n$. The vectors of $\H$ are denoted by $\ket{v}$. A complete description of the quantum system at any given moment in time is called a \textit{quantum state}. When a quantum state is precisely determined, it is a \textit{pure} state. The pure states are defined by \textit{rays} in the space $\H$. The ray is an equivalence class of vectors that differ by multiplication by a nonzero complex scalar. Vectors within the same ray represent the same pure quantum state. Every ray is represented by a vector $\ket{v}$ with unit norm $\braket{v} = 1$. This normalization ensures that the representative vector captures the essential geometric properties and direction of the quantum state.
To give a probabilistic interpretation of the pure states let's assume that $\left\{\ket{x}\right\}$ is an orthonormal basis in $\H$ and any pure state $\ket{v}$ can be represented as

\begin{equation*}
  \ket{v} = \sum_x a_x \ket{x}.
\end{equation*}

The coefficients $a_x$ are called amplitudes and the sum of the squares of their modules is normalized:
\begin{equation*}
   \braket{v} = \sum_x {\lvert a_x\rvert }^2 = 1
\end{equation*}
The square of the magnitude of the amplitude $a_x$ corresponds to the probability of observing the system in the classical state $\ket{x}$:
\begin{equation}
    P \bigl[\ket{x}\big | \ket{v}\bigr] = \vert\braket{x}{v}\vert^{2}=\bra{v}P_x\ket{v},
\end{equation}
where $P_x$ is the projector onto the subspace $\ket{x}$.

When the state of a quantum system is uncertain, it is described as a statistical mixture of pure states, each occurring with a specified probability. This ensemble is referred to as a \textit{mixed state}, which can be mathematically represented as:
\begin{equation}
\label {mixed state}
    \sum_v p_v \ket{v}, \quad \sum_v p_v = 1.
\end{equation}
Here, $p_v$ represents the probability of the system to be in the pure state $\ket{v}$, and the sum of all these probabilities is equal to 1.

The space of pure and mixed states forms a \textit{convex} subspace within the Hilbert space $\mathcal{H}$. The extreme points of this subspace are the pure states, while the mixed states are considered as points within this convex subspace.

If the system is in the mixed state (\ref{mixed state}), then the expected value of an observable $M$ from the ensemble is defined as:
\begin{equation}
     \langle M \rangle = \sum_v p_v \bra{v}M\ket{v}   =  \operatorname{tr} \bigl(M\rho \bigr),
\end{equation}
\noindent
where $\rho$ represents the \textit{density matrix} or the \textit{density operator} of the system:
\begin{equation}
    \label{eqn:pure_states_ensemble}
    \rho = \sum_v p_v \rho_v,
\end{equation}
\noindent
and $\rho_v$ denotes the density operator of a pure state $\ket{v}$:
\begin{equation*}
\rho_v = \ket{v} \bra{v}.
\end{equation*}
\noindent
The density operator encodes all the essential statistics of the ensemble  even when the number of its pure states is greater than $N$. Formally, a density operator $\rho$ is a bounded linear operator $\rho:\H \rightarrow \H$ with the following properties:
\begin{enumerate}
    \item $\rho$ is \textit{Hermitian}, i.e. $\rho = \rho^{\dagger}$.
    \item $\rho$ is a non-negative definite: $\bra{v}\rho\ket{v} \geq 0$, $\forall \ket{v} \in \H$.
    \item $\rho$ has unit trace: $\operatorname{tr}(\rho) = 1$.
    \item The eigenvalues $  \lambda_1,\, \cdots,\,\lambda_N  $ of $\rho$ form a probability distribution. 
\end{enumerate}
\noindent

The density operators of pure states satisfy the equation $\rho = \rho^2$ and they form a complex projective space. In other words, they are projectors onto the one-dimensional subspaces of $\mathcal{H}$.
We will denote the space of bounded linear operators on the Hilbert space $\H$ by $B(\H)$, the space of Hermitian operators  by $H(\H)$ and the space of corresponding density operators by $D(\H)$.
The space $H(\H)$ is a linear vector space, but $D(\H)$ is not.
A general linear combination of density operators is Hermitian, but its trace is not necessarily equal to one. The space of density operators $D(\H)$ forms an $N^2-1$ - dimensional convex space which embeds in $H(\H)$~\cite{Geometry_2017}.

The classical Hidden Markov Models use linear operators to describe the stochastic evolution of states. Similarly, in the context of discrete-time quantum systems, quantum state transformations are defined by linear operators known as quantum operations. A quantum operation, denoted by $\mathcal{O}$, is a linear map:
\begin{equation*}
    \O: B(\H) \rightarrow B(\H)
\end{equation*}
\noindent
with the following properties:

\begin{enumerate}
    \item $\O$ preserves convex combinations of density operators. If the map is applied to a stochastic ensemble of density operators, then the outcome is the same ensemble of the resulting densities.
    \item  $\O$ is trace preserving:
    \begin{equation*}
    \operatorname{tr}(\O\beta) = \operatorname{tr}(\beta), \forall \beta \in B(\H) 
    \end{equation*}
    This property ensures the conservation of probability in physical processes.
    \item $\O$ is a \textit{completely positive} map. It is positive: $\rho \geq 0 \rightarrow \O\rho \geq 0$, and it remains positive 
    under any tensor product extensions. In other words, $\left( \O  \otimes I_K \right)$ is positive for $\forall K > 0$.
\end{enumerate}
\noindent
These properties ensure that a quantum operation is a valid and consistent evolution of quantum states, preserving the probabilistic nature and positivity of density operators.

According to these properties a quantum operation is a \textit{Completely Positive Trace-Preserving} (CPTP) linear map. It is commonly assumed that the evolution of a discrete-time quantum system can be described by a quantum operation. In the context of quantum communication, a quantum operation is often referred to as a \textit{quantum channel}.

An example of quantum operation is the density operator evolution under unitary
transformation $U$:
\begin{equation}
    \label{eqn:UO}
    \O\rho = U \rho U^{\dagger}.
\end{equation}
The general process of extracting classical information from a quantum state or \textit{generalized measurement}, also can be described by quantum operation. The operation is a \textit{complete} set of Hermitian non-negative operators $\mathcal{M} = \{ E_e \}$ with
decomposition $E_e = M_e^{\dagger} M_e$: 

\begin{equation}
\label{eqn:mes_oper_general}
    \mathcal{M} \rho = \sum_e M_e \rho M_e^{\dagger}.
\end{equation}
\noindent

The operators satisfy the \textit{completeness} property $\sum_{e} E_e = I$~\cite{nielsen_chuang_2010}. Each operator  $E_e$ is associated with a classical measurement outcome $e$. If the system is in an ensemble of pure states $\rho$ each  measurement outcome realizes with probability
\begin{equation}
\label{eqn:mes_prob}
    P \left[e\vert\rho\right] = \operatorname{tr}(E_e \rho ).
\end{equation}
\noindent
\noindent
The normalized state after observing outcome $e$ is defined by
\begin{equation}
\label{eqn:mes_stat}
    \rho_e = \frac{ E_e \rho}{P \left[e\vert\rho \right] }.
\end{equation}
A special case of generalized measurement is the \textit{projective measurement}, where the Hermitian measurement operators $\{ E_e \}$ in addition to satisfying the completeness property are orthogonal, i.e. $ E_{e_1} E_{e_2} =  \delta_{e_1, e_2}E_{e_1} $. In this case the operator $\mathcal{M}$ has eigenstates $\{ \ket{e} \}$ and eigenvalues $\{ e \}$ which form an orthogonal basis in $\H$:

\begin{equation*}
    \mathcal{M} = \sum_e e E_e,
\end{equation*}
\noindent
where $E_e = \ket{e} \bra{e}$ are orthogonal projections on the corresponding eigenspaces. The outcome of the projective measurement corresponds to one of the eigenvalues $\{e\}$ associated with the measurement operators, and the post-measurement state of the system collapses to the corresponding eigenstate. In contrast, for generalized measurement operations, the post-measurement states are mixed states, and the measurement outcomes are distributions of outcomes. Another important distinction between projective and generalized measurements is that the successive orthogonal measurements will produce the same outcome, while this is not the case for the generalized measurement operations.

The evolution of the hidden state is modeled by the dynamics of a quantum system with Hilbert space $\H_S$, dimension $N$ and orthonormal basis $\{ \ket{s} \}$. 
This quantum system represents the underlying dynamics of the stochastic process and is referred to as the "state" or $S$-system. 

The emission process, which depends on the hidden state, is modeled by a second quantum system - \textit{emission} system or $E$-system, with Hilbert space $\H_E$, dimension $M$, and orthonormal basis $\{ \ket{e} \}$. The dimension $M$ of the emission system is equal or greater than the number of the observable symbols of the stochastic process. To represent a quantum system composed of two subsystems, we use a quantum operation that combines the components through a tensor product, creating a \textit{bipartite} quantum system:

\begin{equation}
     \rho_S \mapsto \rho_S \otimes \rho_E = \rho_{SE},
\end{equation}
\noindent
where the composed Hilbert space is
\begin{equation}
    \H_{SE} = \H_S \otimes \H_E.
\end{equation}
\noindent
Here the state $\rho_{SE}$ is a \textit{tensor product state}, i.e. initially the subsystems are not entangled. A further assumption is that the initial state of the emission component is a fixed pure state $\rho_E = \ket{e_0} \bra{e_0}$, where $\ket{e_0} \in \{ \ket{e} \}$.

In order to describe one of the components of a bipartite system while disregarding the other we use a quantum operation known as \textit{partial trace} ~\cite{nielsen_chuang_2010}. The partial trace operation is a linear transformation that maps the density operator of the bipartite system to the \textit{reduced density} operator of one of the components:

\begin{equation}
\label{eqn:partial_trace}
    \begin{split}
        \O\rho_{SE} & = \operatorname{tr}_E \left( \rho_{SE} \right) \\
        & = \operatorname{tr}_E \left( \rho_S \otimes \rho_E \right) \\
        & = \sum_e \left(I_N \otimes \bra{e} \right)\left( \rho_S \otimes \rho_E \right) \left( I_N \otimes \ket{e} \right)\\
        & = \rho_S 
    \end{split}
\end{equation}

The partial trace operation in (\ref{eqn:partial_trace}) is applied to a non-entangled bipartite system and therefore the partial density of the state system is independent from the emission system. 

To model joint stochastic processes a unitary operation (\ref{eqn:UO}) is applied to entangle the subsystems and transfer information from the hidden state to the observable:
\begin{equation*}
    \O_U \rho_{SE} = U \rho_{SE} U^{\dagger}.
\end{equation*}
The evolution of the hidden state is described by tracing out the emission component (\ref{eqn:partial_trace}):

\begin{widetext}
\begin{align}
    \label{eqn:op_sum}
    \begin{split}
        \operatorname{tr}_E \left( \O_U \rho_{SE} \right) & = \operatorname{tr}_E \left(U \rho_{SE} U^{\dagger} \right)  = \operatorname{tr}_E \left(U \left( \rho_S \otimes \ket{e_0} \bra{e_0} \right) U^{\dagger} \right) \\
        & = \sum_e I_N \otimes \bra{e} \left[ U \left( (I_N \otimes \ket{e_0}) \rho_S (I_N \otimes \bra{e_0}) \right)
        U^{\dagger} \right] I_N \otimes \ket{e}.
    \end{split}
\end{align}
\end{widetext}

The equation (\ref{eqn:op_sum}) defines  \textit{operator-sum representation} of the quantum operation $\mathcal{T}$ as follows:
\begin{equation}
    \label{eqn:t_u}
    \mathcal{T}\rho_S = \sum_e K_e \rho_S K_e^{\dagger}.
\end{equation}
where the operators $\{K_e\}$, known as \textit{Kraus operators}, act on the state system $\rho_S$ and are defined as follows:
\begin{equation}
    \label{eqn:kraus_op}
    K_e = \left( I^N \otimes \bra{e} \right) U \left( I^N \otimes \ket{e_0} \right).
\end{equation}
\noindent
The Kraus operators depend on the unitary $U$, and the arbitrary selected orthonormal basis  $\{ \ket{e} \}$ of the emission system. Since any unitary transformation of a basis is also an orthonormal basis, the operator-sum representation is defined up to a unitary transformation. A quantum operation can be represented by different number of Kraus operators but there is a minimum. The minimal integer $r$ for which the operation $\mathcal{T}$ has decomposition with $r$ operators is called \textit{Kraus rank} of the operation $\mathcal{T}$ denoted by $\operatorname{rank}(\mathcal{T})$. The Kraus rank is equal to the rank of a positive linear operator known as \textit{Choi matrix} \cite{Cho75} of the operation $\mathcal{T}$:   
\begin{equation}
    \label{eqn:operator_choi}
    J_{\mathcal{T}} =  \sum_{i,j=1}^{N} \ket{i}\bra{j}\otimes\mathcal{T}(\ket{i}\bra{j})  .
\end{equation}
\noindent
The Kraus rank of a quantum operation ranges from 
$1$ for unitary transformations to $N^2$ for completely depolarizing operations. Each Kraus operator has
$N^2$ complex parameters, and all operators must satisfy the completeness constraint. Since the Kraus representation is unique up to a unitary transformation, the total number of free real parameters for a quantum operation is $N^4-N^2$. The set of quantum operations is a convex set with extreme points defined by linearly independent Kraus operators $\{K_e\}$ with maximal rank $N$~\cite{Cho75}. Every quantum operation can be represented as convex combination of  $N^2$ extreme operations. 

For every Kraus operator $K_e$ we define a quantum operation $T_e$ acting on the state $\rho_S$ as follows:
\begin{equation}
    \label{eqn:operator_v}
    T_e \rho_S = K_e \rho_S K_e^{\dagger}.
\end{equation}
\noindent
The operations $T_e = K_e^{\dagger}K_e $ are Hermitian and positive. They satisfy the completeness property $\Sigma_e T_e = I^N$, since the quantum operation $\mathcal{T}$ defined in (\ref{eqn:t_u}) is trace-preserving for all states $\rho_S$:
\begin{equation}
    \operatorname{tr}(\mathcal{T}\rho_S) = \operatorname{tr}\left(\sum_e T_e \rho_S \right)=1.
\end{equation}
\noindent
The operators $\left\{T_e\right\}$ define a generalized measurement  or a \textit{Positive Operator-Valued Measure} (POVM) operation on the state subsystem $S$. The  probability distribution of the measurement outcomes is reflected in the post-measurement density operator $\mathcal{T}\rho_S$ (\ref{eqn:t_u}). Each Kraus operator corresponds to a measurement outcome, and the set of Kraus operators forms a complete set of operators that span the space of possible measurement outcomes. Therefore, if the Kraus operators are linearly independent, then the number of measurement outcomes is equal to the Kraus rank.

We have shown that when an entangled bipartite system is in a pure state, the states of its individual components exhibit characteristics similar to mixed states when observed independently or described using the partial trace operation. This property of the composite system is important as it allows the use of density operators as models of classical stochastic states, despite the underlying pure quantum nature of the system.

A mixed state of one of the components in an entangled bipartite pure state can also be defined by performing a projective measurement on the other component. The measurement outcome on one component provides information about the state of the other component due to the entanglement-related correlation between the components. 
Let $\left\{\ket{e}\right\}$ be an orthonormal basis of the emission component. The corresponding set of orthogonal projectors is:
\begin{equation*}
    \{ M_e = \ket{e} \bra{e}, e = 0, \ldots, M-1 \}.
\end{equation*}
\noindent
We can extend this set to operators $\{ P_e \}$ acting on the bipartite system: 
\begin{equation*}
    P_e = I_N \otimes M_e.
\end{equation*}
\noindent
The projective measurement operation on the full system is defined as (\ref{eqn:mes_oper_general}):

\begin{equation*}
    \mathcal{M} (\rho_{SE}) = \sum_e P_e \rho_{SE} P_e^{\dagger}.
\end{equation*}
\noindent
The measurement outcomes have probabilities defined in  (\ref{eqn:mes_prob}): 
\begin{equation*}
    P \left[ e \mid U \rho_{SE} U^{\dagger} \right] = \operatorname{tr} \left( P_e U (\rho_S \otimes \rho_E) U^{\dagger}
    P_e^{\dagger} \right).
\end{equation*}
\noindent
These probabilities depend on the Kraus operators defined in (\ref{eqn:kraus_op}), and can be expressed in
terms of the measurement basis $\{ \ket{e} \}$:
\begin{widetext}
\begin{align}
\label{eqn:tr_Vrho}
    \begin{split}
        P \left[ e \mid U \rho_{SE} U^{\dagger} \right] & = \operatorname{tr} \left( P_e U (\rho_S \otimes \rho_E) U^{\dagger}
        P_e^{\dagger} \right) \\
        & = \operatorname{tr} \left( I_N \otimes \bra{e} \left[ U \left( (I_N \otimes \ket{e_0}) \rho_S (I_N \otimes \bra{e_0}) \right)
        U^{\dagger} \right] I_N \otimes \ket{e} \right) \\
        & = \operatorname{tr} \left( K_e \rho_S K_e^{\dagger} \right) \\
        & = \operatorname{tr} \left( T_e \rho_S \right)
    \end{split}
\end{align}
\end{widetext}
\noindent

We can summarize that in a bipartite system $\rho_{SE} = \rho_S \otimes \rho_E$ the operation of entanglement of the components, followed by tracing out of one component, e.g. $\rho_E$, has the same operator-sum representation as a projective measurement of $\rho_E$ and "discarding" the outcome.  Both cases are equivalent to a POVM operation on $\rho_S$ as defined in (\ref{eqn:operator_v}):
\begin{equation}
    \begin{split}
    \label{eqn:povm_equality}
    \mathcal{T} \rho_S & =  \sum_e \left( P_e U (\rho_S \otimes \rho_E)           U^{\dagger}P_e^{\dagger} \right) \\ 
        & =\sum_e K_e \rho_s K_e^{\dagger} = \sum_e T_e \rho_S \\ 
        & =\mathcal{M}\rho_S.
    \end{split}
\end{equation}

The equality in (\ref{eqn:povm_equality}) provides two equivalent approaches for modeling the joint process of a state evolution and dependent symbol emission. It can be represented either as a POVM operation on the state system or as a unitary transformation  $U$ entangling the full system followed by a projective measurement on the emission component. In both cases, the state of the state system $S$ becomes a mixed state, representing a stochastic ensemble that averages all possible post-measurement states. This behavior reflects the characteristics of classical stochastic generators and allows the quantum model to reflect the probabilistic nature of the joint stochastic process.

\section{\label{section:QHMM}Quantum Hidden Markov Models}

The POVM operations (\ref{eqn:operator_v}) as described in Section~\ref{subsubsection:qstates_qoperations} define a joint stochastic process
of quantum state transitions and symbols emission.
The quantum operation acting on a quantum state resembles the classical observable operators (\ref{classic observable operators}) transforming a classical stochastic state as defined in Section~\ref{subsubsection:stochastic_and_hmms}.
These similarities provide the foundation for the following quantum hidden Markov model (QHMM) definition.

\begin{definition}[Quantum Hidden Markov Model~\cite{monras2011hidden}]
    \label{def_qhmm}
    A \textit{quantum} HMM (QHMM) $\mathbf{Q}$ over an  $N$-dimensional Hilbert space $\H$ is a 4-tuple:
    \begin{equation}
    \label{qhmm}
        \mathbf{Q} = \{ \Sigma,\, \H,\, \mathcal{T}=\{ T_a \}_{a \in \Sigma},\, \rho_0 \},
    \end{equation}
    where
    \begin{itemize}
    \item $\Sigma$ is a finite alphabet of observable symbols. 
    \item $\H$ is an $N$-dimensional Hilbert space with space of density operators $D(\H)$.
    \item $\mathcal{T}$ is a CPTP map (quantum channel) $\mathcal{T}: D(\H) \rightarrow D(\H)$. 
    \item $\{ T_a \}_{a \in \Sigma}, \sum_{a \in \Sigma}T_a^{\dagger}T_a=I_N$ is an operator-sum representation of $\mathcal{T}$ in terms of a complete set of Kraus operators. 
    \item $\rho_0$ is an initial state, $\rho_0 \in D(\H) $.
    \end{itemize}
\end{definition}

When the operation $\mathcal{T}=\{ T_a \}_{a \in \Sigma}$ is applied to the state $\rho \in D(\H)$, any symbol $a \in \Sigma$ will be observed/emitted with probability (\ref{eqn:mes_prob}):
\begin{equation*}
    P \bigl[ a  \vert \rho \bigr] = \operatorname{tr}(T_a \rho)
\end{equation*}
\noindent
and the system's state will become (\ref{eqn:mes_stat}):
\begin{equation*}
    \rho_a = \frac{T_a \rho}{P\bigl[ a \vert \rho \bigr]}.
\end{equation*}

If the operation $\mathcal{T}$ is applied $t$ times starting at the initial state the model will emit a sequence of symbols $\mathbf{a} = a_1 \ldots a_t$ with
probability
\begin{equation}
    \label{eqn:seq_probability}
    P\bigl[ \mathbf{a} \vert \rho_0 \bigr] = \operatorname{tr}(T_{\mathbf{a}} \rho_0).
\end{equation}
\noindent
where
\begin{equation}
\label{string T}
T_{\mathbf{a}}=T_{a_t} \ldots T_{a_1}
\end{equation}
\noindent

For each QHMM $\mathbf{Q}$ the equation (\ref{eqn:seq_probability}) defines a \textit{sequence function} $f^{\mathbf{Q}}$ and corresponding Hankel matrix (\ref{eqn:hankel matrix}) $H^{\mathbf{Q}}$ as follows:
\begin{equation}
    \label{eqn:qhmm_sequence_function}
    f^{\mathbf{Q}} (\mathbf{a}) = P\bigl[ \mathbf{a} \vert \rho_0 \bigr], \forall \mathbf{a} \in
    \Sigma^{*}, 
\end{equation}
\begin{equation}
    \label{eqn:qhmm_hankel_matrix}
    {H^\mathbf{Q}} \left [ \textbf{p},\textbf{s} \right ] = f^\mathbf{Q} (\textbf{p}\textbf{s}), \forall\textbf{p},\textbf{s}\in {\Sigma}^{*}.
\end{equation}
The sequence function (\ref{eqn:qhmm_sequence_function}) defines a distribution $D^{\mathbf{Q}}_t$ over the sequences of length $t$ for $\forall t>0$:
\begin{equation}
\label{eqn: D_M^T quantum}
    D^{\mathbf{Q}}_t =  \bigl\{ {f^\mathbf{Q}}(\textbf{a}) : \textbf{a} \in {\Sigma}^{t}  \bigr\}
\end{equation}
\noindent
Therefore, every QHMM $\mathbf{Q}$ defines a \textit{stochastic process language} $\mathbf{L^\mathbf{Q}}$ (\ref{eqn:stochstic language distribution_0}) over the set of finite sequences $\Sigma^{*}$.

The eigenvector $\rho^{*}$ of the quantum operation $\mathcal{T}$ with eigenvalue 1, normalized to represent a
distribution, is the \textit{steady state} of the model $\textbf{Q}$:
\begin{equation}
\label{steady_state}
    \rho^{*} = \mathcal{T} \rho^*.
\end{equation}
The dimension of the Hilbert space $dim(\H) = N$ is called \textit{order} of the quantum model $\mathbf{Q}$.

Any iteration of the quantum operation $\mathcal{T}$:
\begin{equation}
\label{t_times}
    \rho_{t} = \mathcal{T}^{t} \rho_0, t \geq 0,
\end{equation}
\noindent
defines a mixed state $\rho_{t}$, which corresponds to all the stochastic paths for generation of a symbol sequence with length $t$.

The following theorem, presented in~\cite{monras2011hidden}, shows that QHMMs can generate the class of languages generated by
the classic HMMs.

\begin{theorem}[Monras~\cite{monras2011hidden}]
    For every classic HMM of order $N$ there exists a QHMM of the same order which generates the same
    stochastic process language.
\end{theorem}

The proof is by construction of a QHMM $\mathbf{Q}$ equivalent to any given classic HMM. Let's consider a classic HMM $\mathbf{C}$, defined by alphabet $\Sigma$ of $m$ observable symbols, space of $n$ unobservable (hidden) process states, initial state distribution $\mathbf{x_0}$ and a set of observable operators $\mathbf{O} = \{ O_a | a \in \Sigma\}$ as defined in (\ref{classic observable operators}).
We can define an equivalent (generating the same stochastic language) QHMM $Q$ with Hilbert space $\H$ with dimension $N=n$ as follows:
    \begin{equation*}
        \mathbf{Q} =  \bigl\{ \Sigma, \H, \mathcal{T} = \{ T_a \}_{a \in \Sigma}, \rho_0  \bigr\},
    \end{equation*}
    where 
    \begin{enumerate}
        \item The initial density operator is $\rho_0 = \ket{0}\bra{0}$.  
        \item Model's states are the the diagonal density operators:
            \begin{equation*}
                D(\H): \bigl\{ \rho_i = \ket{i}\bra{i} | i \in [0, N-1] \bigr\}.
            \end{equation*}
        \item The components of the quantum operation $\mathcal{T}$ are defined as
            \begin{equation}
            \label{theorem_2_kraus}
                T_a = \sum_{ij} K^{ij}_a K^{ij\dagger}_a
            \end{equation}
        where 
            \begin{equation*}   
                K^{ij}_a = \sqrt{O_a[i,j]} \ket{i} \bra{j}.
            \end{equation*}
    \end{enumerate}
    We can verify using (\ref{theorem_2_kraus}),(\ref{classic_observable_prob}), and (\ref{eqn:seq_probability}), that for every sequence $\mathbf{a} \in \Sigma^{*}$ the classical and quantum HMMs define equal probability:
    \begin{equation*}
        P \bigl[\mathbf{a} \vert \mathbf{C}, s_0 \bigr] = \textbf{1}O_{\mathbf{a}} \mathbf{x_0} = \operatorname{tr}(T_{a_l} \ldots T_{a_1} \rho_0) = P
        \bigl[\mathbf{a} \vert \mathbf{Q}, \rho_0 \bigr]
    \end{equation*}
    and therefore they define the same stochastic language.

The QHMM construction above simulates a classical HMM of order $n$ in a Hilbert space with dimension $N=n$. It is important to find out whether the quantum framework allows more efficient representation of the stochastic process languages. 

Let's assume that a classical HMM $\mathbf{M}$ and a QHMM $\mathbf{Q}$ define the same stochastic process language $L$.
If $H^L$ is the generalized Hankel matrix (\ref{eqn:hankel matrix}) associated with $L$, then it is proven \cite{anderson1999realization} that $\operatorname{rank}(H^L)$ provides lower bound of the order of the classic models defining $L$. The bound can be reached for the classic HMM, i.e. if the stochastic language is defined by a classic HMM, then there exists minimal classic HMM defining $L$ with exactly $\text{rank}(H^L)$ states. 
In order to prove similar statements for the QHMM, since the the states are represented by density operators, it is necessary to estimate the size of a basis consisting of orthonormal density operators. The space of density operators $D(\H)$ is a closed and bounded convex space with dimension $N^2-1$. The extreme point of $D(\H)$ are the density operators of the pure states. Extensive study of the geometry of $D(\H)$ is presented for example in \cite{Geometry_2017}. Here we will show that the space of Hermitian matrices $H(\H)$ is spanned by a set of $N^2$ density operators. We will prove the following proposition.

\begin{proposition}
\label{proposition1}
For any Hilbert space $\H$ with dimension $N$ there is a basis of density operators $B_D \subset D(\H)$ with size $N^2$ which spans the space of Hermitian operators $H(\H)$.
\end{proposition}
\begin{proof}
The computational basis for a Hilbert space $\H$ with dimension $N$
consists of $N$ orthonormal vectors:
\begin{equation*}
\bigl\{ \ket{i}: \bra{i}\ket{j} = \delta_{ij}, i \in [0,N-1], j \in [0,N-1]\bigr\}
\end{equation*}
We will define the basis set $B_H$ of Hermitian operators as follows:
\begin{equation*}
\begin{split}
B_H = & \bigl\{ \ket{i}\bra{i}:i \in [0,N-1]\bigr\} \cup \\
      &  \bigl\{\ket{i}\bra{j}+\ket{j}\bra{i}: 0 \le i < j \le N-1 \bigr\} \cup \\  
      &  \bigl\{ \operatorname{i}(\ket{i}\bra{j}-\ket{j}\bra{i}) : 0 \le i < j \le N-1\bigr\},
\end{split}
\end{equation*}
where $\operatorname{i}$ is the imaginary unit. 

\noindent
It is easy to verify that the set $B_H$ contains $N^2=N + {N \choose 2} + {N \choose 2}$ linearly independent operators which span the space $H(\H)$. 
The set contains three groups of Hermitian operators. The first group contains diagonal operators, projectors onto each computational basis state 
$\ket{i}$.The second group contains symmetric operators representing the off-diagonal components of Hermitian matrices. The third group contains antisymmetric operators representing the imaginary off-diagonal  parts of Hermitian matrices.

We will define a basis of density operators $B_D$ which span $H(\H)$ by  transforming the elements of $B_H$ (which are not density operators) to density operators using operation of element-wise summation and real scalar multiplication. These operations will change  the Hermitian basis $B_H$ to a density basis  $B_D$ as follows: 
\begin{itemize}
\item The set of diagonal operators  $\ket{i}\bra{i}$ which are density operators will not be changed.
$$\bigl\{\ket{i}\bra{i}:i \in [0,N-1]\bigr\}$$
\item Each non-diagonal symmetric operator   $\ket{i}\bra{j}+\ket{j}\bra{i}$ is traceless. To ensure that the resulting operator has a unit trace, we sum it with the corresponding diagonal operators and normalize the trace multiplying by the  real scalar $\frac{1}{2}$:
$$\frac{1}{2}(\ket{i}\bra{j}+\ket{j}\bra{i} + \ket{i}\bra{i}+\ket{j}\bra{j})$$

These operators are symmetric with unit trace, and they are positive semi-definite because they can be directly expressed as scaled outer products  $\frac{1}{2}\ket{v}\bra{v}$, where $\ket{v}$ is a linear combination of computational basis vectors:
$$\ket{v}=\ket{i}+\ket{j}$$
Therefore they are density operators.
\item Each non-diagonal anti-symmetric operator   $\operatorname{i}(\ket{i}\bra{j}-\ket{j}\bra{i})$ is traceless and  we also sum it with the corresponding diagonal operators and normalize the trace multiplying by the  real scalar $\frac{1}{2}$:
$$\frac{1}{2}(\operatorname{i}(\ket{i}\bra{j}-\ket{j}\bra{i}) + \ket{i}\bra{i}+\ket{j}\bra{j})$$

These operators are anti-symmetric with unit trace, and they are positive semi-definite because they can be directly expressed as scaled outer products  $\frac{1}{2}\ket{v}\bra{v}$, where $\ket{v}$ is  linear combination of computational basis vectors:
$$\ket{v}=\operatorname{i}\ket{i}-\ket{j}$$
These operators are density operators. 
\end{itemize}

\noindent
The derived set of density operators then becomes:
\begin{equation}
\label{eqn: B(D) basis}
\begin{split}
B_D  = &  \bigl\{ \ket{i}\bra{i}:0 \le i \le N-1\bigr\} \cup \\
    &  \bigl\{\frac{1}{2}(\ket{i}\bra{j}+\ket{j}\bra{i} + \ket{i}\bra{i}+\ket{j}\bra{j})\bigr\} \cup \\ 
    &  \bigl\{\frac{1}{2}(\operatorname{i}(\ket{i}\bra{j}-\ket{j}\bra{i}) + \ket{i}\bra{i}+\ket{j}\bra{j})\bigr\},\\
    & 0 \le i < j \le N-1
\end{split}
\end{equation}

The operators in $B_D$ are linearly independent because:
\begin{enumerate}
    \item The density operators are linearly independence within each group (diagonal, symmetric, and anti-symmetric) because they act on distinct pairs of basis vectors $\ket{i}, \ket{j}$.
    \item The density  operators belonging to different groups  are linearly independent because they act on distinct subspaces of the Hilbert space:
    \begin{itemize}
        \item The diagonal operators $\ket{i}\bra{i}$ act only on the diagonal elements of the matrix.
        \item The symmetric and anti-symmetric operators act on both the diagonal and off-diagonal elements.
        \item The symmetric operators are real, while the anti-symmetric operators always have imaginary component. A real operator cannot be written as a linear combination of operators with imaginary components  using real scalars.
    \end{itemize}
\end{enumerate}
Since the set $B_D$ contains $N^2$ linearly independent density operators it spans the space of Hermitian operators with dimension $N^2$ .
\end{proof}

The existence of a set of linearly independent density operators with size $N^2$ that span the space $H(\H)$ allows to estimate an upper bound on the rank of the stochastic process languages described by QHMMs. 
The following theorem provides a necessary condition for a finite-rank stochastic language to be generated by a QHMM.

\begin{theorem}
\label{theorem:Main1}
Let \(H\) be the infinite Hankel matrix associated with a quantum hidden Markov model (QHMM) \(\mathbf{Q}\) of order \(N\). Then:
\[
\operatorname{rank}(H) \leq N^2.
\]
\end{theorem}

\begin{proof}
Consider a QHMM \(\mathbf{Q}\) as defined in~\eqref{qhmm}, and let \( B_D = \{ b_m \}_{m=1}^{N^2} \) be a basis of the Hermitian operator space \( \mathcal{H}(\mathcal{H}) \), consisting of density operators \( b_m \in D(\mathcal{H}) \) as established in \textbf{Proposition}~\ref{proposition1}.

For each operator \( b_m \in B_D \), define the sequence function \( f_m: \Sigma^* \to \mathbb{R} \) as:
\begin{equation}
\label{eqn: base sequence function}
f_m(\mathbf{a}) = \operatorname{tr}(T_{a_t} \cdots T_{a_1} b_m), \quad \forall \mathbf{a} \in \Sigma^t, \ t \geq 0.
\end{equation}
We will show that the functions \( \{ f_m \} \) generate the column space of the Hankel matrix \( H \).
For any prefix and suffix sequences:
\[
\mathbf{p} = p_1 \cdots p_r, \quad \mathbf{s} = s_1 \cdots s_t,
\]
the entries of the Hankel matrix \( H \) are defined by~\eqref{eqn:seq_probability}:
\begin{equation*}
\begin{split}
H[\mathbf{p}, \mathbf{s}] &= \operatorname{tr}\bigl(T_{s_t} \cdots T_{s_1} T_{p_r} \cdots T_{p_1} \rho_0\bigr) \\
&= \operatorname{tr}\bigl(T_{s_t} \cdots T_{s_1} \rho_{\mathbf{p}}\bigr),
\end{split}
\end{equation*}
where we define the density operator:
\[
\rho_{\mathbf{p}} = T_{p_r} \cdots T_{p_1} \rho_0.
\]
Since \( \rho_{\mathbf{p}} \in \mathcal{H}(\mathcal{H}) \), it can be expressed in the basis \( B_D \):
\[
\rho_{\mathbf{p}} = \sum_{m=1}^{N^2} c_{\mathbf{p},m} b_m,
\]
for some coefficients \( c_{\mathbf{p},m} \in \mathbb{R} \).

Substituting this decomposition into the expression for \( H[\mathbf{p}, \mathbf{s}] \):
\begin{equation}
\label{eqn:HQ[p,s]}
\begin{split}
H[\mathbf{p}, \mathbf{s}] &= \operatorname{tr}\left(T_{s_t} \cdots T_{s_1} \sum_{m=1}^{N^2} c_{\mathbf{p},m} b_m\right) \\
&= \sum_{m=1}^{N^2} c_{\mathbf{p},m} \operatorname{tr}(T_{s_t} \cdots T_{s_1} b_m) \\
&= \sum_{m=1}^{N^2} c_{\mathbf{p},m} f_m(\mathbf{s}).
\end{split}
\end{equation}

Now, consider the row vector of \( H \) corresponding to prefix \(\mathbf{p}\):
\[
R_{\mathbf{p}} = H[\mathbf{p}, \cdot].
\]
From~\eqref{eqn:HQ[p,s]}, we have:
\begin{equation*}
R_{\mathbf{p}}(\mathbf{s}) = \sum_{m=1}^{N^2} c_{\mathbf{p},m} f_m(\mathbf{s}).
\end{equation*}

Thus, each row of the Hankel matrix is a linear combination of the sequence functions \( f_m \), demonstrating that the set \( \{ f_m \}_{m=1}^{N^2} \) spans the row space of \( H \).

This representation allows us to factorize the Hankel matrix \( H \) as:
\begin{equation}
\label{eqn: H - factorization}
H = C F,
\end{equation}
where:
$$ C[\textbf{p},m] =c_{\textbf{p},m}, \textbf{p} \in\Sigma^{*}, m\in[1,\ldots,N^2],    $$
$$ F[m,\textbf{s} ] = f_m(s),  \textbf{s} \in\Sigma^{*}, m\in[1,\ldots,N^2], $$
\noindent
Substituting into~\eqref{eqn: H - factorization}:
\[
H[\mathbf{p}, \mathbf{s}] = \sum_{m=1}^{N^2} c_{\mathbf{p}, m} f_m(\mathbf{s}).
\]

Since \( H \) is expressed as the product of a \((\infty \times N^2)\) matrix \( C \) and an \((N^2 \times \infty)\) matrix \( F \), the  rank of \( H \)  is bounded by:
\[
\operatorname{rank}(H) \leq N^2.
\]

\end{proof}

A comparison of \textbf{Theorem} \ref{theorem: Anderson} and \textbf{Theorem} \ref{theorem:Main1}  shows  that quantum models provide compact representation of the stochastic languages. They require a quadratically smaller state space to generate a finite-rank language compared to their classical counterparts.

Next, we provide a sufficient condition, guaranteeing that for any stochastic language of finite rank, there exists a QHMM that generates it. We also establish the explicit Hilbert space dimension required to construct such a QHMM.

\begin{theorem}
\label{theorem:Main2}
For every stochastic language of finite rank \(r\) over an alphabet of \(m\) symbols, there exists a quantum hidden Markov model (QHMM) \(\mathbf{Q}\), defined in a Hilbert space of dimension 
\[
N = \left\lceil \left( 0.5 + \sqrt{0.25 + m r^2} \right)^{\frac{1}{2}} \right\rceil,
\]
that generates the given stochastic language.
\end{theorem}
\begin{proof}
Let \( L \) be a finite-rank stochastic process language with an associated Hankel matrix \( H \) of finite rank \( r \). To construct a finite representation of \( H \), select \( r \) linearly independent rows, indexed by their corresponding prefixes in the Hankel matrix. Denote this set of prefixes as:
\[
P = \{\mathbf{p}_1, \ldots, \mathbf{p}_r\}.
\]

Since the rows in \( P \) form a basis for the row space of \( H \), the row corresponding to any prefix \(\mathbf{p}\) can be expressed as a linear combination of the basis rows. That is, there exist coefficients \(\alpha_1, \dots, \alpha_r \in \mathbb{R}\) such that:
\begin{equation}
\label{eqn:t33}
H[\mathbf{p}, \cdot] = \sum_{i=1}^{r} \alpha_i H[\mathbf{p}_i, \cdot].
\end{equation}

Assume that the process is in state \( q_i \), represented by the basis prefix \(\mathbf{p}_i\). Upon emitting a symbol \( a \), the process transitions to a new state defined by the concatenated prefix \(\mathbf{p}_i a\). The conditional observation probabilities are defined by the Hankel matrix:

\begin{equation}
\label{classical conditional prob}
\begin{split}
P\bigl[a|q_i\bigr] & = \frac{P[\mathbf{p}_i a]} {P[\mathbf{p}_i]} \\
                   & = \frac{H[\mathbf{p}_i a,\epsilon]}{H[\mathbf{p}_i,\epsilon]}
\end{split}
\end{equation}

By the linear structure of the language established in~\eqref{eqn:t33}, the row of the Hankel matrix indexed by \(\mathbf{p}_i a\) can be expressed as:
\begin{equation}
\label{representstion of rows in classical base}
H[\mathbf{p}_i a, \cdot] = \sum_{j=1}^r \alpha_{i,j}^{(a)} H[\mathbf{p}_j, \cdot],    
\end{equation}
where \(\alpha_{i,j}^{(a)} \in \mathbb{R}\) are the linear coefficients that describe how the new row is represented as a linear combination of the basis rows in \( P \).

This linear structure is directly reflected in the relation between the corresponding sequence probabilities. In particular:
\begin{equation}
\label{representstion of probabilities in classical base}
P[\mathbf{p}_i a] = \sum_{j=1}^r \alpha_{i,j}^{(a)} P[\mathbf{p}_j].   
\end{equation}
\noindent
Here, \(P[\mathbf{p}_i a]\) denotes the probability of observing the sequence defined by the prefix \(\mathbf{p}_i\) followed by the symbol \(a\), while \(P[\mathbf{p}_j]\) correspond to the probabilities associated with the basis prefixes. 
According to equation~\eqref{representstion of probabilities in classical base}, the probability of any observed sequence can be expressed as a linear combination of the probabilities of the basis prefixes. 

Substituting~\eqref{representstion of probabilities in classical base} into the classical conditional probability formula~\eqref{classical conditional prob} yields:
\begin{equation}
    \label{classical conditional 2}
    P\bigl[a \mid \mathbf{p}_i \bigr] = \frac{\sum_{j=1}^r \alpha_{i,j}^{(a)} P[\mathbf{p}_j]}{P[\mathbf{p}_i]}.
\end{equation}
\noindent

The non-normalized conditional transition probabilities allocated to the states $q_j$ are :
\begin{equation}
\label{transition probabilities}
P\bigl[q_j|q_i,a \bigr] = \alpha_{i,j}^{(a)} P\bigl[\mathbf{p}_j\bigr], \quad j=1,\ldots,r   
\end{equation}
These probabilities represent the contributions from states $q_j$ when transitioning from $q_i$ on symbol $a$.

From (\ref{transition probabilities}) we define the state transition matrix $\tau^{(a)}$ for each observable symbol $a \in \Sigma$:
\begin{equation}
    \label{transition matrix}
    \tau^{(a)}_{i,j} = P\bigl[q_j|q_i,a \bigr], \quad j=1,\ldots,r  
\end{equation}

The initial state distribution \(\omega \in \mathbb{R}^r\) is a stochastic vector and depends on $r-1$ free parameters. These parameters can be determined from the linear system:
\begin{equation}
\label{initial_distribution}
\mathbf{1}^\top \tau^{\mathbf{p}_i} \omega = H[\mathbf{p}_i, \epsilon], \quad i = 1, \dots, r-1,
\end{equation}
where \(\mathbf{1} = (1, \dots, 1)^\top \in \mathbb{R}^r\) is the all-ones vector, and:
\[
\tau^{\mathbf{p}_i} = \prod_{j=k}^1 \tau^{a_j}, \quad \mathbf{p}_i = a_1 \dots a_k.
\]

Equation~\eqref{initial_distribution} ensures that the initial state distribution is consistent with the given Hankel matrix probabilities for the selected basis prefixes.

Consider a QHMM \(\mathbf{Q}\) defined in a Hilbert space \(\mathcal{H}\) of dimension \( N \) as per Definition~\eqref{def_qhmm}. By \textbf{Proposition~\ref{proposition1}}, there exists a basis of density operators spanning the space \(\mathcal{H}(\mathcal{H})\) with size \( N^2 \):
\begin{equation}
\label{basis_states}
B = \{ b_1, \dots, b_{N^2} \}.
\end{equation}

Any state \(\rho \in D(\mathcal{H})\) of the QHMM can be expressed as a linear combination of the basis elements:
\[
\rho = \sum_{i=1}^{N^2} d_i^{(\rho)} b_i,
\]
where \( d_i^{(\rho)} \in \mathbb{R} \) are real expansion coefficients that depend on \(\rho\).

Let \( K_a \) be the Kraus operator corresponding to the emission of symbol \( a \in \Sigma \). The updated state after observing \( a \) is:
\begin{equation}
\begin{split}
\rho^{(a)} &= K_a \rho K_a^{\dagger} \\
           &= K_a \left( \sum_{i=1}^{N^2} d_i^{(\rho)} b_i \right) K_a^{\dagger} \\
           &= \sum_{i=1}^{N^2} d_i^{(\rho)} \left( K_a b_i K_a^{\dagger} \right).
\end{split}
\end{equation}

This shows that the action of \( K_a \) on a superposition of basis states is the corresponding superposition of the transformed basis elements:
\begin{equation}
\label{Kraus_action_on_basis}
b_i^{(a)} := K_a b_i K_a^{\dagger}.
\end{equation}

The probability of observing symbol \( a \) given that the system is in basis state \( b_i \) is:
\begin{equation}
\label{quantum_observation_probability}
P[a \mid b_i] = \operatorname{tr}(b_i^{(a)}).
\end{equation}

Since each transformed basis element \( b_i^{(a)} \) can be expanded back into the original basis \( B \):
\begin{equation}
\label{Kraus_action_expansion}
b_i^{(a)} = \sum_{j=1}^{N^2} \gamma_{i,j}^{(a)} b_j,
\end{equation}
where \(\gamma_{i,j}^{(a)} \in \mathbb{R}\) are expansion coefficients depending on both \( K_a \) and the structure of the basis \( B \).
Substituting~\eqref{Kraus_action_expansion} into~\eqref{quantum_observation_probability} yields:
\begin{equation}
\label{expansion_of_observation_probability}
\begin{split}
P[a \mid b_i] &= \operatorname{tr}\left( \sum_{j=1}^{N^2} \gamma_{i,j}^{(a)} b_j \right) \\
              &= \sum_{j=1}^{N^2} \gamma_{i,j}^{(a)} \operatorname{tr}(b_j) \\
              &= \sum_{j=1}^{N^2} P[b_j \mid b_i, a],
\end{split}
\end{equation}
where $\operatorname{tr}(b_j)=1$ and
\begin{equation}
\label{quantum conditional probability}
 P[b_j \mid b_i, a] = \gamma_{i,j}^{(a)}     
\end{equation}
represents the unnormalized local probability mass allocated to state \( b_j \) when transitioning from \( b_i \) upon emitting symbol \( a \).
The coefficients \(\gamma_{i,j}^{(a)}\) quantify how the probability mass is distributed across basis states after applying the Kraus operator.  Since the basis \( B \) is generally non-orthogonal, the expansion coefficients \(\gamma_{i,j}^{(a)}\) may be negative even though \( b_i^{(a)} \) is positive semidefinite. The coefficients \(\gamma_{i,j}^{(a)}\) act as local expansion weights and do not individually guarantee positivity. The overall positivity is ensured at the operator level rather than at the coefficient level.

This phenomenon parallels the classical model’s expansion coefficients \(\alpha_{i,j}^{(a)}\) in~\eqref{representstion of probabilities  in classical base}, which can also take negative values despite the underlying process being probabilistic.

To account for the non-orthogonality of the basis, define the Gram matrix \( G \in \mathbb{R}^{N^2 \times N^2} \) with entries:
\[
G_{jk} = \operatorname{tr}(b_j b_k).
\]

Then, the expansion coefficients can be explicitly expressed as:
\begin{equation}
\label{explicit_gamma}
\gamma_{i,j}^{(a)} = \sum_{k=1}^{N^2} (G^{-1})_{jk} \operatorname{tr}(b_k K_a b_i K_a^{\dagger}).
\end{equation}

To ensure that the QHMM \(\mathbf{Q}\) accurately simulates the classical stochastic process defined by (\ref{transition probabilities}), it is essential to establish a relationship between the transition probabilities in both models. The equality of non-normalized conditional transition probabilities allocated to the basis classical states $q_j$ (\ref{transition probabilities}) and quantum states $b_i$ (\ref{quantum conditional probability}) establish relation between the coefficients $\alpha_{i,j}^{(a)}$ and $\gamma_{i,j}^{(a)}$:

\begin{equation}
        \alpha_{i,j}^{(a)} = \frac{\gamma_{i,j}^{(a)}}{P\bigl[\mathbf{p}_j\bigr]}
\end{equation}

This equality enables the construction of Kraus operators that reproduce the probabilistic transitions of the classical stochastic language. In particular, the expansion coefficients \(\gamma_{i,j}^{(a)}\)~\eqref{explicit_gamma} are quadratic functions of the parameters of the Kraus operators \( K_a \).

For a finite-rank \( r \) stochastic language with \( m \) observable symbols \(\{ a_1, \ldots, a_m \}\), the classical process defines \( m r^2 \) parameters \(\alpha_{i,j}^{(a)}\), corresponding to \( m \) transition matrices (\ref{transition matrix}). The probabilities of the basis prefixes $\mathbf{p}_j$ are defined by the Hankel matrix.

In the quantum case, each Kraus operator \( K_a \) acts on an \( N \)-dimensional Hilbert space. A quantum system of dimension \( N \) can accommodate up to \( N^4 - N^2 \) real parameters through its Kraus operators. To ensure the existence of a solution for the system of quadratic equations, we require:
\begin{equation}
\label{N_estimate}
N \geq \left\lceil \left( 0.5 + \sqrt{0.25 + m r^2} \right)^{\frac{1}{2}} \right\rceil.
\end{equation}

Under this condition, the number of parameters available from the Kraus operators exceeds the number of equations defining the classical transition structure.The relationship between the Kraus parameters and the classical transition coefficients is given by (\ref{explicit_gamma}):
\begin{equation}
\label{Kraus_alpha_equation}
\sum_{k=1}^{N^2} (G^{-1})_{jk} \operatorname{tr}\bigl(b_k K_a b_i K_a^{\dagger}\bigr) = \alpha_{i,j}^{(a)} P\bigl[\mathbf{p}_j\bigr] .
\end{equation}
Since the system is underdetermined (given the dimension condition~\eqref{N_estimate}), solutions always exist in the complex field \(\mathbb{C}\) due to the algebraic closure property.

The final step is to define the initial state  \(\rho_0\) of the quantum model. It depends on \( N^2 - 1 \) real parameters. These parameters are determined by the system of equations derived from the Hankel matrix:
\begin{equation}
\label{initial_conditions}
\operatorname{tr}\bigl(T_{\mathbf{p}} \rho_0\bigr) = H[\mathbf{p}, \epsilon], \quad \mathbf{p} \in P,
\end{equation}
where \( P \) is the basis set of \( r \) linearly independent prefixes. There are $r$ equations and \( N^2 - 1 \) variables.
Using the estimate (\ref{N_estimate}) we have  
\begin{equation}
\label{rh0_estimate}
N^2-1 \geq \left\lceil \left( -0.5 + \sqrt{0.25 + m r^2} \right)\right\rceil.
\end{equation}
For every $m>1$ the inequality (\ref{rh0_estimate}) implies
\begin{equation}
\label{rh0_estimate final}
N^2-1 \geq r.
\end{equation}
Therefore, under the condition (\ref{N_estimate}) there exists a QHMM that reproduces the classical stochastic process characterized by the finite-rank Hankel matrix $H$. 
\end{proof}

The theorem demonstrates a fundamental distinction between quantum stochastic generative models and classical hidden Markov models (HMMs).
While a finite Hankel rank is a necessary condition for an HMM representation, it is not sufficient\cite{Fox68,Dharm63,carlyle_paz,vidyasagar2011complete}. In contrast, Theorem \ref{theorem:Main2} asserts that every finite-rank stochastic process language can be generated by a quantum stochastic generator as a quantum hidden Markov model (QHMM).  This implies that while every classical HMM can be embedded within a QHMM, there exist stochastic processes that admit QHMM representations but lack classical HMM representations. Consequently, the class of languages generated by QHMMs is a proper superset of those generated by classical HMMs. 

\subsection{\label{subsection:unitary_qhmm}Unitary Definition of QHMM}

The definition of a QHMM in the previous section is based on the concept of a POVM operation acting on a quantum state. The emission of the observable symbols is encoded in the operational elements (Kraus operators) of the quantum operation. 
This framework provides a convenient way to view QHMMs as channels for quantum information processing. It allows for the analysis of their informational complexity, expressive capacity, and establishes connections to stochastic process languages and the corresponding automata. 
However, this approach cannot be directly used for the implementation of the QHMMs on quantum computing hardware. 
\begin{figure*}[!ht]
    \centering
    \includegraphics[scale=0.5]{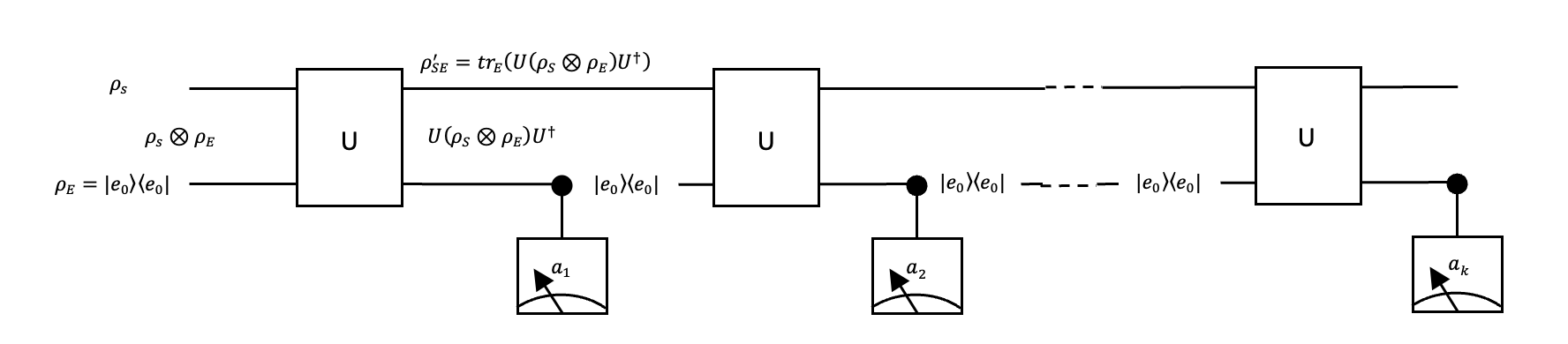}
    \caption{Unitary Simulation of a Quantum Operation.}
    \label{fig:unitary_composite_circuit}
\end{figure*}

In Section~\ref{subsubsection:qstates_qoperations} we demonstrated that every POVM operation $\mathcal{T}$ defined in Hilbert space $\H_S$ can be represented as composition of basic quantum operations:

\begin{enumerate}
    \item Dilation of the ``state'' system with an ``emission'' quantum system $\H_E$: 
    $$\H_{S} \mapsto\H_S \otimes \H_E.$$
    \item Unitary transformation $U$, entangling the components of the bipartite system $\rho_{SE}$:  
    $$U(\rho_S \otimes \rho_E) U^{\dagger}.$$
    \item Partial trace over the emission component, which provides the reduced density of the state system as the output of the operation:
    $$\rho_{S} = \operatorname{tr}_E \left( U (\rho_S \otimes
    \rho_E) U^{\dagger} \right).$$ 
\end{enumerate}
\noindent
The minimal dimension of the emission system $\H_E$ corresponds to the Kraus rank of the operation $\mathcal{T}$. This rank quantifies the operation's complexity by indicating the minimal auxiliary resources required for its unitary implementation. The unitary implementation of a QHMM relies on Stinespring's dilation theorem ~\cite{Stinespring1955},which states that any quantum operation can be represented as a unitary evolution on a larger (dilated) system.
\begin{theorem}
    \label{theorem_bijective_map}
    Any QHMM $\mathbf{Q} =  \bigl\{ \Sigma, \, \H_S,\, \mathcal{T}=\{ T_a \}_{a \in \Sigma},\, \rho_0  \bigr\}$ has parameterization in terms of a unitary $U$ defined in a dilated Hilbert space $\H$, where $\operatorname{dim}\H \geq \operatorname{dim}\H_S + |\Sigma|$. 
\end{theorem}

\begin{proof}
    Let's introduce an auxiliary (``emission'') quantum system with Hilbert space $\H_E$, $\operatorname{dim}(\H_E) = M, M = |\Sigma|$,  and orthonormal basis $\{ \ket{e_i} \}^{M-1}_{i = 0}$.  We assume that $\operatorname{ord}(a)$ is the position of the symbol $a$ in the set $\Sigma = \{a_1, \ldots, a_{M-1}\}$.
    Let's select a specific state $\ket{e_0}$ to be the initial state. 
    \noindent
    A linear operator $V:H_{S} \rightarrow \H_S \otimes \H_E$ is defined as follows:
    \begin{equation}
    \label{eqn:isometry_v}
        V: \ket{v} \mapsto \sum_{a \in \Sigma} T_a \ket{v} \otimes \ket{e_{\operatorname{ord}(a)}}, \ket{v} \in \H_S.
    \end{equation}
    \noindent 
    Since the quantum operation $\mathcal{T}$ is complete: $\sum_{a \in \Sigma} T_a^{\dagger} T_a = I_N$, the operator $V$ is an isometry: $V^{\dagger}V=I_N$, and it can be extended to a unitary operator $U:\H_S\otimes \H_E \rightarrow \H_S \otimes \H_E$ by using the Gram-Schmidt process \cite{nielsen_chuang_2010}.

    If we apply the unitary operator $U$ to the dilated system $\H_S \otimes \H_E$ and trace over the emission subsystem
    \begin{equation*}
    \operatorname{tr}_E \left( U(\ket{v}\bra{v} \otimes \ket{e_0}\bra{e_0})U^{\dagger} \right) = \sum_{a \in \Sigma} T_a \ket{v} \bra{v} T_a^{\dagger} , 
    \end{equation*}
\noindent
where $\ket{v} \in \H_S$,  we derive the operator-sum representation of $\mathbf{Q}$. Hence, the unitary $U$ parameterizes the QHMM.   
\end{proof}
\noindent
The above parameterization is unique up to a unitary transformation on the selected basis of the emission system. 

Every unitary operator $U$ acting on a composite system $\H_{SE} = \H_S \otimes \H_E$ defines a quantum operation by the Kraus operators (\ref{eqn:kraus_op}).
Every quantum operation defines a unitary operator on a dilated system through the isometry $V$ (\ref{eqn:isometry_v}). This equivalence allows us to provide a formal definition of unitary QHMMs.

\begin{definition}[Unitary Quantum Hidden Markov Model]
    \label{def_unitary_qhmm}
    A \textit{Unitary Quantum} HMM $\mathbf{Q}$ over an alphabet of observable symbols $\Sigma$ and finite
    $N$-dimensional Hilbert space is a 6-tuple:
    \begin{equation*}
        \mathbf{Q} =  \bigl\{ \Sigma, \H_S, \H_E, U, \mathcal{M}, R_0  \bigr\}
    \end{equation*}
    where
    \begin{itemize}
    \item $\Sigma$ is a finite set of $m$ observable symbols. 
    \item $\H_S$ is the Hilbert space of the hidden state system of dimension $N$. \item $\H_E$ is the Hilbert space of an auxiliary emission system with dimension $m \leq M \leq N^2$ and orthonormal basis $E=\{ \ket{e_i} \}^{M-1}_{i = 0}$.
    \item $U$ is a unitary operator defined on the bipartite Hilbert space $\H_S \otimes \H_E$.
    \item $\mathcal{M}$ is a bijective map $\mathcal{P}_m^{E} \rightarrow \Sigma$, where $\mathcal{P}_m^{E}$ is an $m$-element partition of $E$.
    \item $ R_0 = \rho_0 \otimes \ket{e_0}\bra{e_0}$ is an initial state.
\end{itemize}
\end{definition}

The Algorithm \ref{alg:sim_qhmm} simulates a unitary QHMM {$\mathbf{Q}$} for $T$ steps and generates a sequence of the process language $L^\mathbf{Q}$. The same generation process is presented graphically on \fig{unitary_composite_circuit}. 
\newline

\SetKwComment{Comment}{/* }{ */}

\begin{algorithm}[H]
    \SetKwInOut{Input}{Input}\SetKwInOut{Output}{Output}
    \DontPrintSemicolon
    \caption{QHMM Unitary Simulation}\label{alg:sim_qhmm}
    \Input{\textbf{QHMM} $\mathbf{Q} = \{ \Sigma, \H_S, \H_E, U, \mathcal{M}, R_0= \rho_0 \otimes \ket{e_0}\bra{e_0} \}$ }
    \Output{Generated Sequence $\mathbf{a}$}
    \textbf{Initialize: 
      \begin{itemize}
        \item State system $\H_S, \operatorname{dim}(\H_S)=N, \rho_S=\rho_0 $\; 
        \item Emission system $\H_E, \operatorname{dim}(\H_E)=M, m \leq M \leq N^2, \rho_E=\ket{e_0}\bra{e_0} $\;
        \item Prepare product state $ \rho_{SE} = \rho_S \otimes \rho_E $\;
        \item Sequence length: $T$ \;
        \item Current Step: $t \gets 0$\;
        \item Output Sequence: $\mathbf{a} \gets \epsilon$\;
    \end{itemize} }

    \While{$t  \leq T$}{
        $\rho_{SE} \gets U\rho_{SE}U^{\dagger} $\Comment*[r]{ Apply unitary on full system} 
        $\rho_{S}  \gets \operatorname{tr}_E(\rho_{SE})$\Comment*[r]{ Projective measurement of $\rho_E$} 
        $ o \gets e_i                                  $\Comment*[r]{ Measurement outcome $o$} 
        $ \mathbf{a} \gets \mathbf{a} + \mathcal{M}(E_o), E_o \in \mathcal{P}_m^{E},  o\in E_o$\Comment*[r]{ Emit symbol for outcome $o$} 
        $\rho_{SE}  \gets \rho_{S} \otimes \ket{e_0}\bra{e_0} $\Comment*[r]{ Set $\rho_E$ to initial state} 
        $t \gets t+1$
        }
    \Return{$\mathbf{a}$}
\end{algorithm}

\begin{example}
    \normalfont

    In~\cite{monras2011hidden}, Monras et al. discuss an example of QHMM in Hilbert space with dimension 2, which defines a stochastic process language with Hankel matrix of rank 3. This is an example where the minimal classic HMM (of order 3) is more complex than the equivalent QHMM.
    The quantum model $\textbf{Q}$ is defined  as follows:
    \begin{equation*}
        \textbf{Q}= \bigl\{ \Sigma,\H_{S},\mathcal{T} = \{  \mathcal{T}=\{ T_a \}_{a \in \Sigma},\}, \rho_0 \bigr\},
    \end{equation*}

    where 
    \begin{itemize}
    \item $\Sigma=\{0,1,2,3\}$ is an alphabet of observable symbols.
    \item $\H_{S}$ is a Hilbert space with dimension 2 (i.e. $\textbf{Q}$ is a 1 qubit quantum system).
    \item The operator-sum representation of $\mathcal{T}$ is $\mathcal{T}_a \cdot = K_a \cdot K_a^{\dagger}$, where the Kraus operators are: 
    $K_0=\frac{1}{\sqrt{2}}\ket{\uparrow}\bra{\uparrow}$, $K_1=\frac{1}{\sqrt{2}}\ket{\downarrow}\bra{\downarrow}$,
    $K_2=\frac{1}{\sqrt{2}}\ket{+}\bra{+}$, and $K_3=\frac{1}{\sqrt{2}}\ket{-}\bra{-}$.
    \item $\rho_0$ is the maximally entangled initial state. 
    \end{itemize}
\noindent    
    We provide the following unitary definition of the same QHMM:
    \begin{equation*}
        \textbf{Q}=\{ \Sigma,\H_{S},\H_{E},U,\mathcal{M} , \rho_0 \otimes \ket{e_0}\bra{e_0} \},
    \end{equation*}
    where 
    \begin{itemize}
    \item $\Sigma=\{0,1,2,3\}$.
    \item $\H_{S}$ is the state system Hilbert space with dimension 2.
    \item $\H_{E}$
    is the emission Hilbert space with dimension 4 (i.e. a 2 qubit quantum system with a measurement basis $\{ e | e \in [0,1,2,3] \}$), corresponding to the observable symbols.
    \item $U$ is the unitary operation implemented by the ``State transition'' circuit in \fig{MonrasCircuit}.
    \item $\rho_0$ is the maximally
    mixed state implemented by the ``Max mixed state'' circuit in \fig{MonrasCircuit}. 
    \item $\mathcal{M}$ is the identity $\{0,1,2,3\} \rightarrow \{0,1,2,3\}$.
    \end{itemize}

The circuit generating sequences with length 2 and their distribution are shown in \fig{MonrasCircuit} and \fig{monras_example}.
The ``X'' gates controlled on the classical ``observable'' registers $o_0$ and $o_1$ are used to reset the emission system to the initial state $\ket{e_0}$.  

\begin{figure*}[!ht]
    \centering
    \includegraphics[scale=1]{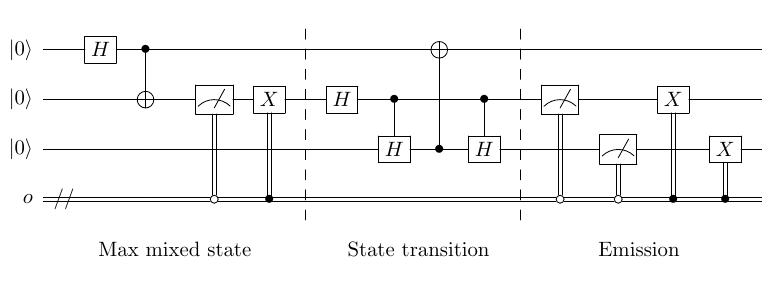}
    \caption{QHMM generating 1-symbol sequences, Monras et al. \cite{monras2011hidden}}
    \label{fig:MonrasCircuit}
\end{figure*}

\begin{figure*}[!ht]
        \centering
        \includegraphics[scale=1.0] {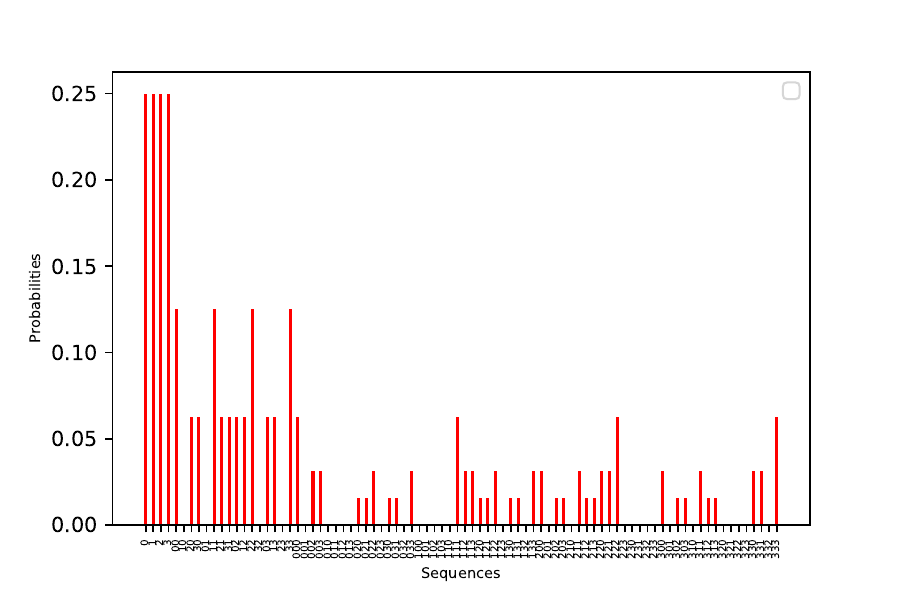}
        \caption{ Observable Sequences Distribution, Monras et al. \cite{monras2011hidden}}
        \label{fig:monras_example}
\end{figure*}

\end{example}

\subsection{\label{subsec:primmeassystemdesign}Primary and Measurement System Design. Test Case for the Ansatz Circuit Design}

To get ideas for system design we refer to the theory of open quantum systems, which studies the system dynamics under the system-environment interaction \cite{breuer2007}. Researchers consider Markovian and non-Markovian time evolutions \cite{Rivas2012,smart2022relaxation}. The precise definition of (non-) Markovianity is a subject of ongoing research \cite{rivas2014quantum,Pollock2018}. The major difference between these two evolution regimes are the memory effects in the non-Markovian case. Non-Markovian regime complicates the dynamics of the system, makes it irreversible, e.g. see \cite{filippov2020quantum}. The information in non-Markovian evolution flows from the system to the environment and from the environment to the system \cite{morris2019non}. In real physical systems the description of the environment or even its dimension may not be known to us, however, in the present research we design both the system and the environment to model stochastic processes, so we have full control.

Building on the inspiration from the open quantum systems we consider two circuit designs to implement QHMM (\fig{hqmm_circuit_design}).

\begin{figure*}[htbp]
    \centering
        \begin{minipage}{0.95\textwidth}
            \raggedright{a)}
            \begin{quantikz}
                & \lstick{$\ket{0}$} & \gate[wires=2]{\text{Initial State}} & \qw & \gate[wires=2]{U_{AE}} & \qw & \meter{} & \qw & \gate[wires=2]{U_{AE}} & \qw & \meter{} & \qw & \gate[wires=2]{U_{AE}} & \qw & \meter{} & \qw\\
                & \lstick{$\ket{0}$} &\qw & \qw & & \qw & \gate{\ket{0}} & \qw & & \qw & \gate{\ket{0}} & \qw & & \qw & \gate{\ket{0}} & \qw\\
            \end{quantikz}
        \end{minipage}
        \begin{minipage}{0.95\textwidth}
            \raggedright{b)}
            \begin{quantikz}
                & \lstick{$\ket{0}$} & \gate[wires=2]{\text{Initial State}} & \qw & \gate[wires=2]{U_{AE}} & \qw & \meter{} & \qw & \gate[wires=2]{U_{AE}} & \qw & \meter{} & \qw & \gate[wires=2]{U_{AE}} & \qw & \meter{} & \qw\\
                & \lstick{$\ket{0}$} &\qw & \qw & & \qw & \qw & \qw & & \qw & \qw & \qw & & \qw & \qw & \qw\\
            \end{quantikz}
        \end{minipage}
    \caption{QHMM Ansatz Circuit Design}
    \label{fig:hqmm_circuit_design}
\end{figure*}

The circuit on \fig{hqmm_circuit_design} a) collapses the state of the environment to the ground state after the first iteration of the unitary evolution with $U_{AE}$. The state of the system collapses to the result of the measurement so it is the only information transmitted to the next iteration. The circuit on \fig{hqmm_circuit_design} b) may exhibit longer term memory effects.

For circuit training purposes we introduce a separate transformation of the measurement system $U_{Meas}$ (\fig{hqmm_circuit_design_meas}):
$$U_{AE} = (U_{\text{Meas}}\otimes I) U_{AE}^1$$

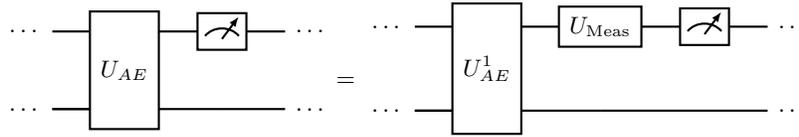
\begin{figure*}[htbp]
    \begin{center}
        \begin{minipage}{1\textwidth}
            \begin{quantikz}
                & \lstick{$\dots$} & \gate[wires=2]{U_{AE}} &  \meter{} & \qw \rstick{$\dots$}\\
                & \lstick{$\dots$} & \qw & \qw & \qw \rstick{$\dots$}\\
            \end{quantikz} = 
            \begin{quantikz}
                & \lstick{$\dots$} & \gate[wires=2]{U_{AE}^1} & \gate{U_{\text{Meas}}} & \meter{} & \qw \rstick{$\dots$}\\
                & \lstick{$\dots$} & \qw & \qw & \qw &  \qw \rstick{$\dots$}\\
            \end{quantikz}
        \end{minipage}
    \end{center}
\caption{Separated Measurement System Transformation}
\label{fig:hqmm_circuit_design_meas}
\end{figure*}

The choice of the initial state depends on the use case. In the examples shown on \fig{monras_qka_linear_hw_circuit} and \fig{market_realamplitudes_circuit} below we use maximally mixed state as the initial state.

Based on the preliminaries in Section \ref{subsubsection:qstates_qoperations} we can write the probability of a sequence for the circuit \fig{hqmm_circuit_design} a) as

    \begin{multline}
        \label{eqn:evolution_measurement_multi}
            \text{Pr}\{j_1j_2\dots j_k\} = \\ \text{Tr} \Biggl \{M_{j_k} {\mathlarger{\sum\limits_{i}{}}}{}K_i \dots M_{j_2} {\mathlarger{\sum\limits_{i}{}}}{}K_i M_{j_1} \biggr . \\ \Biggl . {\mathlarger{\sum\limits_{i}{}}} {K_i\rho_AK_i^\dagger}  M_{j_1}^\dagger K_i^\dagger  M_{j_2}^\dagger\dots K_i^\dagger  M_{j_k}^\dagger \biggr \}
    \end{multline}

\noindent
where $M$ are projective measurements, and $K$ are Kraus operators $K_i = \left(I\otimes\langle{e_i}\vert\right)U\left(I\otimes\vert{e_0}\rangle\right)$. Please, note that the same expression for the circuit \fig{hqmm_circuit_design} b) would be much more complicated, since the start state $e_0$ of the environment may be changing between iterations creating a new set of Kraus operators each time.

How do we know that the circuit on \fig{hqmm_circuit_design} a) is able to produce probabilities according to the equation
\eqn{evolution_measurement_multi}?
Here we suggest to use a well-known quantum operation - amplitude damping noise model
\cite{nielsen_chuang_2010, wilde2016quantuminformationtheory} as a test case and build a sequence generator.
It is described by the following unitary transformation:
\begin{equation*}
    \mathbf{U}_{AD}=
    \begin{pmatrix}
        1 & 0 & 0 & 0\\
        0 & 0 & 0 & 1\\
        0 & -\sqrt{\gamma} & \sqrt{1-\gamma} & 0\\
        0 & \sqrt{1-\gamma} & \sqrt{\gamma} & 0
    \end{pmatrix},
\end{equation*}
which, alternatively, can be represented in a circuit form:

\begin{multline*}
    \mathbf{U}_{AD}= \\ (I\otimes\vert0\rangle\langle0\vert+ X\otimes\vert1\rangle\langle1\vert)\times \\ \mathbf(\vert0\rangle\langle0\vert\otimes I+ \vert1\rangle\langle1\vert\otimes RY(\theta))
\end{multline*}

or

    \begin{center}
            \begin{quantikz}
                & \lstick{} & \ctrl{1} & \gate{X} & \qw \\
                & \lstick{} & \gate{RY(\theta)} & \ctrl{-1} & \qw\\
            \end{quantikz}
    \end{center}
\noindent
where $\gamma = \sin^2(\frac{\theta}{2})$. The full circuit is shown on \fig{ad_test_case}.

\begin{figure}[htbp]
    \begin{center}
            \begin{quantikz}
                & \lstick{$\ket{0}$} & \gate{H} & \ctrl{1} & \gate{X} & \meter{} & \qw\\
                & \lstick{$\ket{0}$} & \qw & \gate{RY(\theta)} & \ctrl{-1}& \gate{\ket{0}}  &  \qw\\
            \end{quantikz}
    \end{center}
    \caption{Amplitude damping test case circuit for 1 step sequence}
    \label{fig:ad_test_case}
\end{figure}
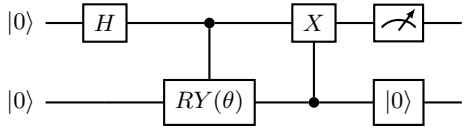

From this unitary we derive Kraus operators describing evolution of the principal subsystem given that the start state of the environment is $e_0=\vert{0}\rangle$:
\begin{equation*}
\begin{split}
    \mathbf{K_0}= \\
    & =
    \left(\begin{pmatrix}
        1 \\
        0
    \end{pmatrix}\otimes
    \begin{pmatrix}
        1 & 0 \\
        0 & 1
    \end{pmatrix}\right)^{\dagger}U
    \left(\begin{pmatrix}
        1 \\
        0
    \end{pmatrix}\otimes
    \begin{pmatrix}
        1 & 0 \\
        0 & 1
    \end{pmatrix}\right) \\
    & =
    \begin{pmatrix}
        1 & 0 \\
        0 & 0
    \end{pmatrix},
\end{split}
\end{equation*}

\begin{equation*}
\begin{split}
\mathbf{K_1}= \\
& =
\left(\begin{pmatrix}
        0 \\
        1
    \end{pmatrix}\otimes
    \begin{pmatrix}
        1 & 0 \\
        0 & 1
    \end{pmatrix}\right)^{\dagger}U
    \left(\begin{pmatrix}
        1 \\
        0
    \end{pmatrix}\otimes
    \begin{pmatrix}
        1 & 0 \\
        0 & 1
    \end{pmatrix}\right) \\
    & =
    \begin{pmatrix}
        0 & -\sqrt{\gamma} \\
        0 & \sqrt{1-\gamma}
    \end{pmatrix}.
\end{split}
\end{equation*}

In this test case we want to compare the sequence probabilities calculated using \eqn{evolution_measurement_multi} with circuit (\fig{ad_test_case}) results, when it is run on the simulator and on \emph{ibmq\_montreal} hardware device.
The probabilities calculated on the simulator will include shot, or sampling, noise. In order to reduce this error we use $100,000$ shots.
The hardware result will include both: hardware and shot noise.
The circuit that was run on the simulator and \emph{ibmq\_montreal} hardware device is shown on \fig{ampl_damping_qiskit_circ}.
\begin{figure*}[!ht]
    \centering
    \includegraphics[scale=0.55]{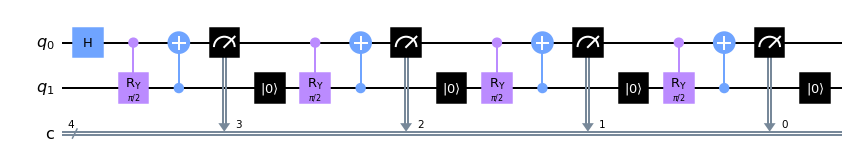}
    \caption{Amplitude damping noise circuit in Qiskit. $\theta=\frac{\pi}{2}$}
    \label{fig:ampl_damping_qiskit_circ}
\end{figure*}
Please note that the measurements are recorded on the classical register in reverse order following Qiskit's
convention of putting the least significant bit to the right.
Mid-circuit measurement and reset instructions have become available fairly recently in early 2021
\cite{midcircuitmeas}.
They are now part of a broader dynamic circuit capability introduced in 2022 \cite{dynamiccircuits}.

The comparison of the sequence probabilities is shown on \fig{ampl damping seq prob}.

\begin{figure*}[!ht]
    \centering
    \includegraphics[scale=0.55]{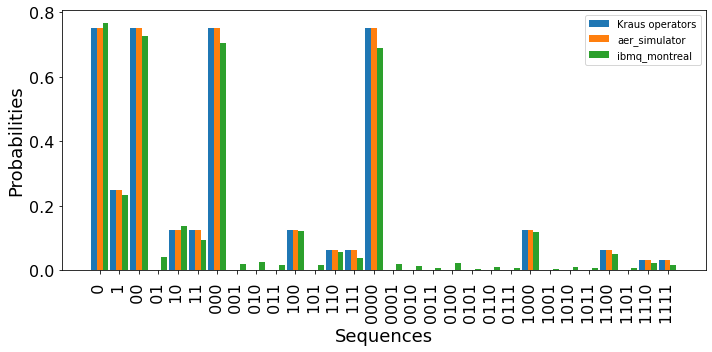}
    \caption{Sequence probabilities for amplitude damping noise model. $\theta=\frac{\pi}{2}$}
    \label{fig:ampl damping seq prob}
\end{figure*}

All calculations produce identical results up to sampling and hardware errors, so we confirm that our approach
works as expected.
It is worth noting that hardware results are very close to simulations.
Here we have used \emph{Sampler} primitive in \emph{Qiskit Runtime} to run this experiment
\cite{ampldampingruntimeexperiment}.
Behind the scenes this primitive implements error mitigation \cite{samplererrormitigation}.
We have repeated the experiments for a range of $\gamma$ parameters including the edges $0$ and $1$.
In all cases, we observed the same level of agreement between the results.

The Hankel matrix for the current process is shown in Table~\ref{tab:ampl_dump_Hankel_matr}. The rank of the Hankel matrix changes with the length of the sequence (see ~\fig{ampl_dump_Hankel_rank}). It is interesting to note that if we choose a different start state on the quantum circuit - use excited state $\vert1\rangle$ instead of $\vert+\rangle$ state, then the Hankel matrix will have rank $2$ for all sequences of length $2$ and above.

\begin{table*}

    \centering

        \begin{tabular}{|l|p{40pt}|p{40pt}|p{40pt}|p{40pt}|p{40pt}|p{40pt}|p{40pt}|}

            \hline

            \hfill{~} & \hfil empty &\hfil 0 &\hfil 1 &\hfil 00 &\hfil 01 &\hfil 10 &\hfil 11 \\

            \hline

            \hfil empty & \hfil 1 & \hfil 0.75 & \hfil 0.25 & \hfil 0.75 & \hfil 0 & \hfil 0.125 & \hfil 0.125 \\

            \hfil 0     & \hfil 0.75 & \hfil 0.75 & \hfil 0 & \hfil 0.75 & \hfil 0 & \hfil 0 & \hfil 0 \\

            \hfil 1     & \hfil 0.25 & \hfil 0.125 & \hfil 0.125 & \hfil 0.125 & \hfil 0 & \hfil 0.0625 & \hfil 0.0625 \\

            \hfil 00    & \hfil 0.75 & \hfil 0.75 & \hfil 0 & \hfil 0.75 & \hfil 0 & \hfil 0 & \hfil 0 \\

            \hfil 01    & \hfil 0 & \hfil 0 & \hfil 0 & \hfil 0 & \hfil 0 & \hfil 0 & \hfil 0 \\

            \hfil 10    & \hfil 0.125 & \hfil 0.125 & \hfil 0 & \hfil 0.125 & \hfil 0 & \hfil 0 & \hfil 0 \\

            \hfil 11    & \hfil 0.125 & \hfil 0.0625 & \hfil 0.0625 & \hfil 0.0625 & \hfil 0 & \hfil 0.03125 & \hfil 0.03125 \\

        \hline

        \end{tabular}

    \caption{Hankel matrix of the amplitude damping process.}
    \label{tab:ampl_dump_Hankel_matr} 
\end{table*}
 
    \begin{figure}[!ht]
    \centering
    \includegraphics[scale=0.55]{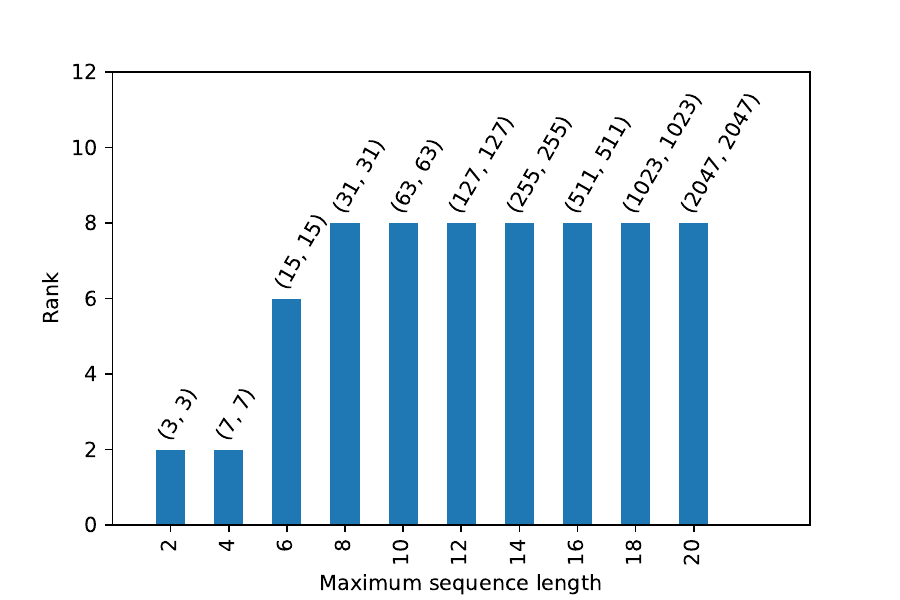}
    \caption{Hankel matrix rank for different length sequences. Hankel matrix dimensions are shown in parentheses.}
    \label{fig:ampl_dump_Hankel_rank}
    \end{figure}

\section{\label{section:learning_quantum_hidden_markov_models}Learning Quantum Hidden Markov Models}

In this section, we formalize and study the problem of learning QHMMs. The learning problem is defined as follows:  given an empirical specification of a stochastic process language $L$, which is referred to as the learning target, the objective is to find a QHMM $Q$ that is equivalent or approximates the sequence function $f^{L}$ of the target.
We discuss and formalize the components of the learning problem which include: 
\begin{itemize}
    \item Specification of the target stochastic language by estimating of its characteristics from the data. These characteristics may include 
     order, Hankel matrix, finite distributions.
     \item Definition of a set of potential solutions or a hypotheses space of unitary QHMMs.
     \item Design of hypotheses' quality criteria or a fitness function.
     \item Design of learning algorithms.
\end{itemize}

\subsection{Empirical Specification of Stochastic Languages}

For every stochastic process language $L$ we  defined a sequence function $f^L$ and corresponding generalized Hankel matrix $H^L$ (\ref{eqn:stochastic language seq function}). Sequence function $f^Q$ and Hankel matrix $H^Q$ are defined for the language of every QHMM $Q$ as well (\ref{eqn:qhmm_sequence_function}), (\ref{eqn:qhmm_hankel_matrix}). We assume that the target stochastic language is specified by a finite sample of the functional relation $f^L$, respectively a finite sub-matrix of  $H^L$. If the sample is derived analytically or through exhaustive simulation of a known HMM, then its probabilities will be exactly the same as the probabilities of the target distribution. It is possible though, that the sample is derived through observation of a "natural" random source which is assumed to be Markovian. In this case the probabilities of the sample need to be estimated and this will have impact on the precision of the learned model. First, we discuss the size of a "representative" sample for a language or the \textit{sample complexity} of the learning problem. It is known~\cite{carlyle_paz} that if a language $L$ is defined by a classic HMM of order $n$ then it is uniquely specified by the probabilities of all sequences up to length $2n-1$. If the alphabet $\Sigma$ contains $m$ symbols this estimate defines a sample size of $\sum_{i=1}^{2n-1}m^i$ probabilities. If it is acceptable the specification to be valid for \textit{almost all} HMMs-defined languages excluding a set with measure zero~\cite{huang}, then the representative sample would contain the probabilities of all sequences of length exactly $2l+1$ where $l > 8\lceil log_m{n} \rceil$ or $m^{2l+1}$ probabilities.
In many learning scenarios, a classical HMM for the learning target is not provided, and in some cases, it is not known whether such an HMM exists. Instead, the only available training data consists of a set of observed sequences, each of length $T$, sampled from an unknown process language $L$:
\begin{equation*}
    \mathcal{Y}^L=\bigl\{ \textbf{y}_1, \cdots \textbf{y}_l : \textbf{y}_i \in\Sigma^{T} \bigr\}.
\end{equation*}

Using these empirical samples we have to estimate the sequence function (\ref{eqn:stochastic language seq function}) and then the order and the finite  distributions of the unknown language. 
To estimate the sequence function for sequences with length $k$ we construct the following set:
\begin{equation}
\label{eqn:samle k} 
    \textbf{S}_k= \bigl \{ \textbf{a} : |\textbf{a}|=k, \exists  \textbf{p}, \textbf{s} \in \Sigma^{*}, \exists y \in  \mathcal{Y}_L : \textbf{p}\textbf{a}\textbf{s} = \textbf{y} \bigr  \}.
\end{equation}
\noindent
If the sample $\textbf{S}_k$ contains $m$ independently drawn sequences from the unknown
distribution $D^L_{k}$, then the estimate $\hat{f}(\textbf{a})$ is the empirical frequency of $\textbf{a}$ in $\textbf{S}_k$:
\begin{equation*}
    \hat{f}(\textbf{a})= \frac{1}{m} \sum_{\textbf{b} \in \textbf{S}_k} \textbf{1}_{\textbf{a}}(\textbf{b}).
\end{equation*}
\noindent
The estimate of the finite distributions $D^L_{k}$ is
\begin{equation}
\label{eqn:stochstic language distribution_1}
    \hat{D}^{L}_{k} = \bigl \{ \hat{f}(\textbf{a}) : \textbf{a} \in \textbf{S}_k  \bigr \}
\end{equation}
\noindent
The approximation of the Hankel matrix is:
\begin{equation*}
    \hat{H}[\textbf{u},\textbf{v}]=\hat{f}(\textbf{u}\textbf{v}), |\textbf{u}\textbf{v}|=k .
\end{equation*}
\noindent
The error of this approximation depends on the sample  size $m$~\cite{hsu} and is restricted as follows:
\begin{equation*}
\norm{H - \hat{H}} \leq \mathcal{O}\left( \frac{1}{\sqrt{m}}\right).
\end{equation*}
If we are given a sample $\mathcal{Y}^L$ with no additional information we have to estimate a plausible order of the
classical model by the maximum rank of the Hankel matrix.
We start with low sequences length and estimate the rank increasing the length.
If at particular sequence length the rank does not increase anymore we can use this rank as an estimate.
This procedure will provide only local approximation since the problem to infer the rank of the Hankel matrix from
data is undecidable~\cite{sontag}.

\subsection{QHMM Hypotheses Space}
\label{subsection:hypothesis space}
The hypotheses space $\Q$ of the learning problem contains the unitary QHMMs ($\mathbf{Definition}~\ref{def_unitary_qhmm}$):
\begin{equation}
\label{eqn:hyp_space}
    \Q = \{ \mathbf{Q} : \mathbf{Q}=\left( 
    \Sigma, \H_{S}, \H_{E}, U, \mathcal{M} , \rho_0 \right)\},
\end{equation}
\noindent
where the parameters are specified and restricted by the given samples of the target stochastic process language:

\begin{itemize}
    \item The alphabet $\Sigma$ contains the observed $m$ symbols.
    \item $\H_{S}$ is a Hilbert space with dimension $\hat{N}= \left\lceil \sqrt{\hat{r}} \right\rceil$ where ${\hat{r}}$ is the
    estimated maximal rank of the Hankel matrix built for the data sample $\mathcal{Y}$.
    \item $\H_{E}$ is a Hilbert space with dimension $ 2^{\lceil \log_2 m\rceil}\leq M \leq  \hat{N^2} $.
    We select any orthonormal basis $\mathcal{E}=\{\ket{e_{0}} \cdots \ket{e_{M-1}}  \}$ of $\H_{E}$.
    \item $U$ is a unitary operation on the Hilbert space $\H_S \otimes \H_E$ implemented by a quantum circuit of
    $log_2 \hat{N}+ log_2 M$ qubits using gates from a base \textit{gates type} set $\mathcal{G}=\{g_0 \cdots g_k  \}$.

    \item $\mathcal{M}$ is a bijective map $\mathcal{P}_m^{E} \rightarrow \Sigma$, where $\mathcal{P}_m^{E}$ is an $m$-element partition of $E$.
    
    \item $\rho_0$ is initial state implemented as one of the following: the maximally entangled state, the
    maximally mixed state, or the ground state.
\end{itemize}

Every hypothesis $\mathbf{Q}$ defines a sequence function $f^\mathbf{Q}$ and corresponding Hankel matrix $H^\textbf{Q}$ which are used for evaluation of
its quality. Every unitary operator $U$ has representation as a quantum circuit $\mathcal{C}_U$ with linear structure
\begin{equation}
    \label{eqn:c_u}
    \mathcal{C}_U=(g_i)_{i\geq1}.  
\end{equation}
\noindent
The quantum gates $g_i$ are encoded as 3-tuples:
\begin{equation}
\label{eqn:qgate}
    g=\langle \textbf{t},([q_c,] q_d),([p_1[, p_2]]) \rangle
\end{equation}
where $\textbf{t} \in \mathcal{G}$ is the gate's type, $q_c$ and $q_d$ are the control and data qubits, and $p_1,p_2$ are gate's parameters. We will make the assumption that the set of gate types $\mathcal{G}$ includes single-qubit gates and two-qubit controlled gates.

Formally, the hypothesis space is the infinite set of all finite quantum circuits in the $NM$-dimensional Hilbert space $\H_{S} \otimes \H_{E}$, encoded as lists of gates ~(\ref{eqn:qgate}).

\subsection{Quality of a Hypothesis}
\subsubsection {Precision of a Hypothesis}
A standard measure of the quality of a generative hypothesis $Q$ defining language $L^Q$ is its variational distance to the target distributions $D^L$(\ref{eqn:stochstic language distribution_0}). Let's assume that the distributions $D^L$ are defined by an unknown target QHMM model. The divergence between the target and the  hypothesis for any sequence ${\mathbf{a} \in \Sigma^*}$ can be defined as the divergence in probabilities assigned to the sequence ${\mathbf{a}}$ by the target and the hypothesis(\ref{eqn:seq_probability}):
\begin{equation}
\label{eqn:trace_dist1}
\begin{split}
 \delta^{LQ}(\mathbf{a}) & = \bigl{|} P\bigl[\mathbf{a} \big \vert L \bigr] - P\bigl[\mathbf{a} \big \vert Q \bigr]\bigr{|} \\ 
        & =  \bigl{|} \operatorname{tr}(T_{\mathbf{a}}^{L}\rho_0) - \operatorname{tr}(T_{\mathbf{a}}^{Q}\rho_0\bigr{|}) \\
        & =  \bigl{|} \operatorname{tr}\rho_{\mathbf{a}}^{L} - \operatorname{tr}\rho_{\mathbf{a}}^{Q} \bigr{|}, 
\end{split}
\end{equation}
\noindent
where $\mathcal{T}^{L}=\{ T_{a}^L \}_{a \in \Sigma}$ and $\mathcal{T}^{Q}=\{ T_{a}^Q \}_{a \in \Sigma}$ are the quantum operations of the target model and the hypothesis and we assume a fixed initial state $\rho_0$. The quantum states $\rho_{\mathbf{a}}^{L} = T_{\mathbf{a}}^{L}\rho_0$ and $\rho_{\mathbf{a}}^{Q}=T_{\mathbf{a}}^{Q}\rho_0$ represent the sequence generation paths under the target and the hypothesis models. The usual measurable distance between two quantum states is the \textit{trace norm} ${\lVert \rho_{\mathbf{a}}^L - \rho_{\mathbf{a}}^Q \rVert_1 = \operatorname{tr}{\bigl{|} \rho_{\mathbf{a}}^L - \rho_{\mathbf{a}}^Q \bigr{|}}}$. In the context of learning a QHMM, it is not feasible to directly measure the trace norm because the target model is generally unknown or may not even exist. The introduced divergence $\delta^{LQ}$ can be estimated empirically and is dominated by the theoretical trace norm: 
\begin{equation}
\label{eqn:trace_dist2}
\begin{split}
 \delta^{LQ}(\mathbf{a}) & = \bigl{|} \operatorname{tr}\rho_{\mathbf{a}}^{L} - \operatorname{tr}\rho_{\mathbf{a}}^{Q} \bigr{|} \\
 & = \bigl{|}\sum_{i=1}^N (\rho_{\mathbf{a}}^{L})_{ii} - (\rho_{\mathbf{a}}^{Q})_{ii}\bigr{|} \\
 & \leq  \sum_{i=1}^N \bigl{|}(\rho_{\mathbf{a}}^{L})_{ii} - (\rho_{\mathbf{a}}^{Q})_{ii}\bigr{|} \\
 & = \operatorname{tr}{\bigl{|} \rho_{\mathbf{a}}^L - \rho_{\mathbf{a}}^Q \bigr{|}}.
\end{split}
\end{equation}
\noindent
This relation will be used later in our analysis of QHMMs learning difficulty.

The divergence between the target language and the hypothesis on finite distributions of sequences with lengths exactly $n$ is
\begin{equation}
\label{eqn:distributions divergence trace n}
\Delta_{n}({D^L}||{D^Q})=\max_{\textbf{a} \in \Sigma^{n}} \delta^{LQ}
(\mathbf{a}) 
\end{equation}
\noindent
The average divergence between the hypothesis $Q$ and the target $L$ when given finite sample of sequences with lengths up to $n$ is defined as
\begin{equation}
\label{eqn:average divergence trace n}
\Delta_{\leq n}({D^L}||{D^Q})=\frac{1}{n}\sum_{i=1}^{n} \Delta_{i}({D^L}||{D^Q})    
\end{equation}
\noindent
We will refer to the average divergence (\ref{eqn:average divergence trace n}) $\Delta_{\leq n}({D^L}||{D^Q})$ as the \textit{empirical divergence} between a hypothesis and the learning target.

The \textit{empirical divergence rate} of the hypothesis $Q$ with respect to the target $L$ is defined as
\begin{equation}
\label{eqn:divergency rate limit}
\hat{\Delta}({D^L}||{D^Q})= \lim_{n\to\infty} \Delta_{\leq n}({D^L}||{D^Q}),
\end{equation}
\noindent
if the limit exists and it is finite. 

We can prove that the empirical  divergence rate between any QHMM hypothesis and a finite order stochastic process language converges with the increase of the sequences length $n$ of the finite distributions in the learning sample.
\begin{proposition}
\label{empirical  divergence rate}
If the learning target $L$ is a finite order stochastic process language, then for any hypothesis $Q \in \mathcal{Q}$ the empirical divergence rate (\ref{eqn:divergency rate limit}) 
$$\hat{\Delta}({D^L}||{D^Q})= \lim_{n\to\infty} \Delta_{\leq n}({D^L}||{D^Q})$$ 
exists and is finite.
\end{proposition}
\begin{proof}
For the finite order process languages the estimate of the relative entropy is consistent \cite{Finesso2010}: 
\begin{equation}
\label{eqn:Hankel Dist convergence}
\hat{D}_{KL}({D^L}||{D^Q})= \lim_{n\to\infty} \frac{1}{n}D_{KL}^{n}({D^L}||{D^Q}) 
\end{equation}
\noindent
where
\begin{equation}
\label{eqn:Hankel convergence}
D_{KL}^{n}({D^L}||{D^Q})= \sum_{\mathbf{a} \in \Sigma^{n}} P\bigl[\mathbf{a} \big\vert L \bigr]\log \frac{P\bigl[\mathbf{a} \big\vert L \bigr]}{P\bigl[\mathbf{a} \big\vert Q \bigr]},
\end{equation}
\noindent
and the trace distance between the finite distributions is bounded by the relative entropy according to the Pinsker inequality:
$$ D_{KL}({D^L}||{D^Q}) \geq \frac{1}{2}D_{\operatorname{tr}}^2({D^L}||{D^Q}),$$
\noindent
where
$$D_{\operatorname{tr}}({D^L}||{D^Q})= \lim_{n\to\infty} D_{\operatorname{tr}}({D^{L}_{n}}||{D^{Q}}_{n}),$$
$$D_{\operatorname{tr}}({D^{L}_{n}}||{D^{Q}}_{n})=\frac{1}{n}\sum_{i=1}^{n} d_{\operatorname{tr}}({D^L_{i}}||{D^Q_{i}}) ,$$
$$d_{\operatorname{tr}}({D^{L}_{i}}||{D^{Q}_{i}})=\sum_{\textbf{a} \in \Sigma^{i}}\operatorname{tr}{\bigl{|} \rho_{\mathbf{a}}^L - \rho_{\mathbf{a}}^Q \bigr{|}}.$$
From the inequality in (\ref{eqn:trace_dist2}) follows
\begin{equation}
\begin{split}
D_{KL}({D_L}||{D_Q}) & \geq \frac{1}{2}D_{\operatorname{tr}}^2({D^L}||{D^Q})\\
& \geq \frac{1}{2}\hat{\Delta}^2({D^L}||{D^Q}),
\end{split}
\end{equation} 
and we can conclude that the limit (\ref{eqn:divergency rate limit}) exists in the cases when the target $L$ is a finite order stochastic process language.

\end{proof}

The empirical divergence rate (\ref{eqn:divergency rate limit}) can be estimated from a sample of the target language and is  considered to be the \textit{precision} of a hypothesis $Q$. This measure can be used in any approximate QHMM learning model even for languages which are not generated by any hypothesis in the space $\Q$.  

Theoretically the divergence between the target $L$ and a hypothesis $Q$ defined by quantum operations $\mathcal{T}^L$ and $\mathcal{T}^Q$ ($\mathbf{Definition}$ \ref{def_qhmm}) is the maximum probability divergence for sequences of length exactly $n , n>0$ over all initial states $\rho_0$ is the trace norm defined as:
\begin{equation}
\label{trace_divergence}
    {\lVert (\mathcal{T}^L)^n - (\mathcal{T}^Q)^n\rVert}_1 \hspace{0.05cm} = \hspace{0.05cm}
    \underset{\rho}{\mathrm{max}} \hspace{0.05cm}{\lVert (\mathcal{T}^L)^n \rho-(\mathcal{T}^Q)^n \rho}\rVert_1   
\end{equation}
To estimate the relation between the trace divergence and the empirical divergence we consider the quantum states which define distributions over the sequences with length $n$ :
\begin{equation}
\label{n-th power}
    \rho_{n}^{i} = (\mathcal{T}^i)^{n}\rho_0, \rho_0 \in D(\H_S), i = \{L, Q\},   
\end{equation}
where $\rho_0$ is an initial state. The trace norm ${\lVert \rho_{n}^{L} - \rho_{n}^{Q}\rVert}_1$ is bounded from below by the empirical divergence (\ref{eqn:distributions divergence trace n}) of $Q$  as follows:
\begin{equation}
\label{eqn: trc_nrm_bnd}
\begin{split}
{\lVert \rho_{n}^{L} - \rho_{n}^{Q}\rVert}_1 & = {\lVert  (\mathcal{T}^L)^{n}\rho_0 - (\mathcal{T}^Q)^{n}\rho_0\rVert}_1 \\
                                             & =  \underset{\mathbf{a} \in \Sigma^{n} }{\mathrm{max}} \hspace{0.1cm}{\lVert T_{\mathbf{a}}^L \rho_0-T_{\mathbf{a}}^Q \rho_0}\rVert_1 \\  
                                             & =  \underset{\mathbf{a} \in \Sigma^{n} }{\mathrm{max}} \hspace{0.1cm}  {\operatorname{tr}\bigl{|}\rho_{\mathbf{a}}^L - \rho_{\mathbf{a}}^Q\bigr{|}} \\
                                             & \geq \underset{\mathbf{a} \in \Sigma^{n} }{\mathrm{max}} \hspace{0.1cm} \delta^{LQ}(\mathbf{a)} \\
                                             & = \Delta_{n}({D^L}||{D^Q})
\end{split}
\end{equation}
\noindent
From the inequality (\ref{eqn: trc_nrm_bnd}) follows that, given finite sample of sequences with length $n$,  the empirical divergence between the hypothesis $Q$ and the target $L$  (\ref{eqn:distributions divergence trace n})   is a lower bound of the trace divergence (\ref{trace_divergence}):
\begin{equation}
\label{eqn:avg dive trace vs trc divergence }
\Delta_{n}({D^L}||{D^Q}) \leq {\lVert (\mathcal{T}^L)^n - (\mathcal{T}^Q)^n\rVert}_1    
\end{equation}
\noindent
This inequality will be used to demonstrate a critical feature of the  heuristic search approach to the learning of QHMMs: small changes in the norms of the unitary hypotheses lead to small changes of their precision.  

\subsubsection {Complexity of a Hypothesis}
Another measure of quality is the hypothesis' complexity. If we assume that the state Hilbert space $\H_S$ has fixed dimension $N$, where $N=\operatorname{rank}(\hat{H_L})$, then the free parameters of a hypothesis $Q$ are the unitary $U$ and the dimension of the emission component $M$. If for a unitary $U$ we denote by $c_2(U)$ the number of two qubit gates then as a simple measure of its gate complexity we use the function:

$$ C_g(Q) = \frac{c_2(U)}{\binom{NM}{2}}$$

The complexity of the unitary implementation is related also to the dimension of the emission system $\H_E, |\Sigma| \leq M \leq N^2$. The following measure reflects the dimension-related complexity:  
$$ C_e(Q) = \frac{M}{N^2}$$
\noindent
The full complexity of a hypothesis is estimated as follows:
\begin{equation}
\label{eqn: u_complexity}
C(Q) = c_{q}C_{q}(Q) + c_{e}C_{e}(Q), 
\end{equation}
\noindent
where $c_q$ and $c_e$ are real hyper parameters reflecting the trade-off between complex unitary or wider circuit.

We will integrate the hypothesis' precision (\ref{eqn:average divergence trace n}) and complexity (\ref{eqn: u_complexity}) measures into a single quality function called $\textit{fitness}$ which is defined as follows: 
\begin{equation}
\label{eqn:fitness}
    F(Q) = -\left(D^{n}({D_L}||{D_Q})+C(Q)\right),
\end{equation}
\noindent
   
\subsection{\label{subsection:qhmm_learning_problem} Formalization of the QHMM Learning Problem}
We will discuss a specific formalization of the QHMM learning problem defined as follows:
\begin{definition}[Unitary QHMM learning problem]
    \label{def_unitary_qhmm_lp}
    A Unitary Quantum HMM learning problem is defined as follows. Given:
    \begin{itemize}
    \item A sample $\bigl \{D^{L}_t : 0 < t \leq n \bigr \}$ of the finite distributions of a unknown stochastic process language $L$ (learning target),
    \item A unitary hypothesis space $$\Q = \bigl \{ Q : Q=\left( 
    \Sigma, \H_{S}, \H_{E}, U, \mathcal{M} , \rho_0 \right)\bigr \}$$  \normalfont{defined in Section \ref{subsection:hypothesis space}},
    \item A fitness function $F(Q)$ \normalfont{(\ref{eqn:fitness})} defined over the hypothesis space,
    \end{itemize}
    \noindent  
    find a hypothesis $Q^{*} \in \Q$ that maximizes the fitness function:   
    \begin{equation}
    \label{eqn:argmax}
        \mathbf{Q}^{*} = \underset{Q \in \Q}{\mathrm{argmax}} \hspace{0.1cm} F(Q).
    \end{equation}
\end{definition}

The corresponding classical HMM learning problem is expected to be computationally hard \cite{Kearns1990}, therefore we propose the optimization task (\ref{eqn:argmax}) to be approached by heuristic search algorithms. 

The heuristic search algorithms use the fitness function as a heuristic search function, providing information about the distance and direction  towards the optimal solution. Therefore the properties of the fitness function are strongly related to the complexity of the learning task. 
To examine this relationship, we utilize the concept of a \textit{fitness landscape}, which represents a multidimensional surface defined as follows:
\begin{equation}
\label{F-Landscape}
\bigl \{ F(Q):Q \in \Q \bigr \}.
\end{equation}
\noindent
Intuitively, an optimization task is expected to be efficient, if the points on the fitness landscape are correlated with the optimal point or at least with the points in their neighbourhoods. This property requires the landscape to be \textit{smooth} and any small change of a  hypothesis $Q$ results in restricted change of its fitness function. To formalize these intuitive notions, we introduce appropriate metrics in the hypothesis spaces $\Q$ and in the fitness space. 

The distance in the hypothesis space is quantified by the standard operator norm of the unitary operators. If $Q_1$ and $Q_2$ are two hypotheses with unitary operators correspondingly $U_1$ and $U_2$ ($\mathbf{Definition}$~\ref{def_unitary_qhmm_lp}), then the distance $\Delta Q^{12} = \Delta U^{12}$ is defined as follows:
\begin{equation}
    \label{eqn:delta_op_nrm}
    \Delta Q^{12}={\lVert U_1-U_2 \rVert}.
\end{equation}
According to the \textit{genetic algorithms'} terminology, the distance in the hypothesis representation space is called \textit{genotypes} distance.

To introduce distance in the fitness space, we note that the unitary HQMMs representations are  defined in a larger Hilbert space (by tensoring an emission system) and use a \textit{stabilized} version of the trace norm known as \textit{diamond norm} \cite{AKN1998}: 
\begin{equation}
    \label{eqn:dm_nprm}
    {\lVert\mathcal{T}\rVert}_\diamond = \underset{\rho \in D(\H_S\otimes \H_S)}{\mathrm{max}} \hspace{0.1cm}{\lVert (\mathcal{T}\otimes I_N)\rho \rVert}_1.
\end{equation}
\noindent
In the diamond norm definition the dimension of the emission system $M$ is equal to the dimension of the state system $N$, since the increase of $M > N$ cannot increase the trace norm.
The diamond norm induces distance between the probabilities of $1$-symbol sequences generated by the hypotheses $Q_1$ and $Q_2$ :
\begin{align*}
\begin{split}
    \Delta \mathcal{T}^{12} & =  {\lVert \mathcal{T}^1 - \mathcal{T}^2\rVert}_\diamond \\
    & = \underset{\rho \in D(\H_S\otimes \H_S)}{\mathrm{max}} \hspace{0.1cm}{\lVert (\mathcal{T}^1-\mathcal{T}^2)\otimes I_N\rho \rVert}_1, 
\end{split}
\end{align*}
which we call a \textit{phenotypes} distance. Since the diamond norm is a stabilized version of the trace norm (Lemma 12, \cite{AKN1998}) we have the inequality:
\begin{equation}
\label{Diamond dominates Trace}
{\lVert (\mathcal{T}^1-\mathcal{T}^2) \rVert}_1 \leq {\lVert \mathcal{T}^1 - \mathcal{T}^2\rVert}_\diamond,   
\end{equation}
\noindent
and from inequality (\ref{eqn:avg dive trace vs trc divergence }) follows, that the average divergence between the hypothesis $Q_1$ and $Q_2$ on finite sample of sequences with length $1$ is a lower bound of the diamond norm:
\begin{equation}
\label{eqn:avg dive trace vs diamond}
\Delta_{1}({D^1}||{D^2})  \leq  {\lVert \mathcal{T}^1 - \mathcal{T}^2\rVert}_\diamond   
\end{equation}
\noindent

The introduced distances in hypothesis and fitness spaces imply smooth fitness landscape for the single-symbol sequences,  due to the Continuity of Stinespring’s representation Theorem \cite{KRETSCHMANN20081889}. According to the Theorem, two quantum operations $\mathcal{T}^1,\mathcal{T}^2$, respectively two QHMMs, are close in diamond norm iff they have unitary representations which are close in operator norm:   
\begin{equation}
    \label{eqn:cb_nprm}
    \frac{1}{2}\Delta \mathcal{T}^{12} \leq \Delta U^{12} \leq\sqrt{\Delta \mathcal{T}^{12}}
\end{equation}
This inequality suggests that the operator norm difference between the hypotheses dominates the trace (\ref{Diamond dominates Trace}) and average (\ref{eqn:avg dive trace vs diamond}) empirical divergences of their single-symbol distributions: 
\begin{equation}
    \label{eqn:op_norm_trace_distance}
    \frac{1}{2} \Delta_{1}({D^1}||{D^2}) \leq \frac{1}{2} {\lVert \mathcal{T}^1 - \mathcal{T}^2\rVert}_1  \leq \Delta U^{12}.  
\end{equation}
The distributions over sequences with any finite length $n$ are generated by the $n-th$ powers of the quantum operations (\ref{eqn:seq_probability}) and their unitary dilations: 
\begin{equation}
\label{eqn: n-th power of U}
    (\mathcal{T}^{i})^{n}\rho_0=U_{i}^n\rho_0 U_{i}^{{\dagger}n}, \rho_0 \in D(\H_S), i \in {1,2}.    
\end{equation}
The operator-norm distance between the unitary operators defining $n$-symbol sequences is related to the genotypes' distance $\Delta Q^{12}={\lVert U_1-U_2 \rVert}$ as follows:

\begin{equation}
    \label{eqn:op_norm_trace_distance nth rewrite}
    \begin{split}
        {\lVert U_1^{n}-U_2^{n} \rVert} & = {\Big\| (U_1-U_2)\sum_{k=1}^n U_1^{n-k}U_2^{k-1} \Big\|}\\
                                        & \leq {\Big\| U_1-U_2\Big\|} {\Big\| \sum_{k=1}^n U_1^{n-k}U_2^{k-1}\Big\|}\\
                                        & \leq {\big\| U_1-U_2\big\|} (\sum_{k=1}^n \big\| U_1^{n-k}U_2^{k-1}\big\|)\\
                                        & \leq n{\big\| U_1-U_2\big\|}.
    \end{split}
\end{equation}

\noindent
For the case of quantum operations defining $n$-symbol sequences, and accounting that the diamond norm is stabilized trace norm, from the Continuity Theorem (\ref{eqn:op_norm_trace_distance}) follows:
\begin{equation}
    \label{eqn:op_norm_trace_distance nth iteration}
    \frac{1}{2} {\big\| (\mathcal{T}^1)^{n} - (\mathcal{T}^2)^{n}\big\|}_1 \leq {\big\| U_1^{n}-U_2^{n} \big\|}.
\end{equation}
By applying (\ref{eqn:op_norm_trace_distance nth rewrite}) and (\ref{eqn:trace_dist2}),(\ref{eqn:distributions divergence trace n})  to (\ref{eqn:op_norm_trace_distance nth iteration}) we can derive an upper bound estimate for the $n-$symbol distributions divergence in terms of genotypes' distance in the hypothesis space :
\begin{equation}
\label{eqn: stepn}
    \frac{1}{2n} \Delta_n({D^1}||{D^2}) \leq \frac{1}{2n} {\big\| \mathcal{T}_1^{n} - \mathcal{T}_2^{n}\big\|}_1 \leq \Delta U^{12}
\end{equation}
The empirical divergence of a learning sample of sequences with lengths up to $n$ is bounded as follows: 
\begin{equation}
\label{eqn: sum stepn1}
    \frac{1}{2n}\sum_{i=1}^n\frac{1}{i} \Delta_i({D^1}||{D^2}) \leq \Delta U^{12}
\end{equation}
\noindent
This inequality allows the precision component in the fitness distance between of two hypotheses to be estimated by the operator norm distance of their unitary representations(\ref{eqn:average divergence trace n}) :
\begin{equation}
\label{eqn:restricted distributions divergence trace}
\Delta_{\leq n}({D^1}||{D^2})  = \frac{1}{n}\sum_{i=1}^n \Delta_{i}({D^1}||{D^2})
\end{equation}
\noindent
 By replacing the harmonic weight $\frac{1}{i}$ of the divergence in (\ref{eqn: sum stepn1}) by $\frac{1}{n}$ and applying (\ref{eqn:restricted distributions divergence trace}): 
 $$\Delta_{\leq n}({D^1}||{D^2})  \leq \sum_{i=1}^n\frac{1}{i} \Delta_{i}({D^1}||{D^2}) $$
 we derive the estimate:
\begin{equation}
\label{eqn: final estimate}
    \frac{1}{2n} \Delta_{\leq n}({D^1}||{D^2}) \leq \Delta U^{12}
\end{equation}

Inequality (\ref{eqn: final estimate}) demonstrates that the fitness landscape of the QHMMs learning problem is smooth: for every two unitary hypothesis close in genotype distance (operator norms) the corresponding phenotypes' distance (empirical distributions divergence) for any sequences length $n$ is restricted.

To experimentally investigate the difficulty of the heuristic approach to the QHMM learning task we estimate the landscape properties of the the QHMM learning problem discussed in Example \ref{exm: market}. The dependency between the change in the operator norm and the corresponding divergence of the sequence distributions is estimated by a random walk starting at the optimal hypothesis $Q^*$ with unitary $U^* = U_0$. A new unitary is generated at every step $t = 1,2,\dots$ by random single parameter mutation with standard deviation $10\%$. The expected fitness divergences (phenotypes distance) $\Delta_{n}(D^*||D^t)$ (\ref{eqn:distributions divergence trace n}) for sequences with lengths $n\in[2\dots5]$, and the expected total fitness divergence $\Delta_{\leq 5}({D^*}||{D^t})$ (\ref{eqn:restricted distributions divergence trace}) conditioned on the genotypes distance (\ref{eqn:delta_op_nrm})  $ \Delta U={\lVert U^{*}-U_t \rVert}, \Delta U \in [0,1]$ are shown on \fig{ Landscape Continuity}.

\begin{figure}[ht]
    \centering
    \begin{minipage}{.5\textwidth}
        \centering
        \includegraphics[width=1.0\linewidth]{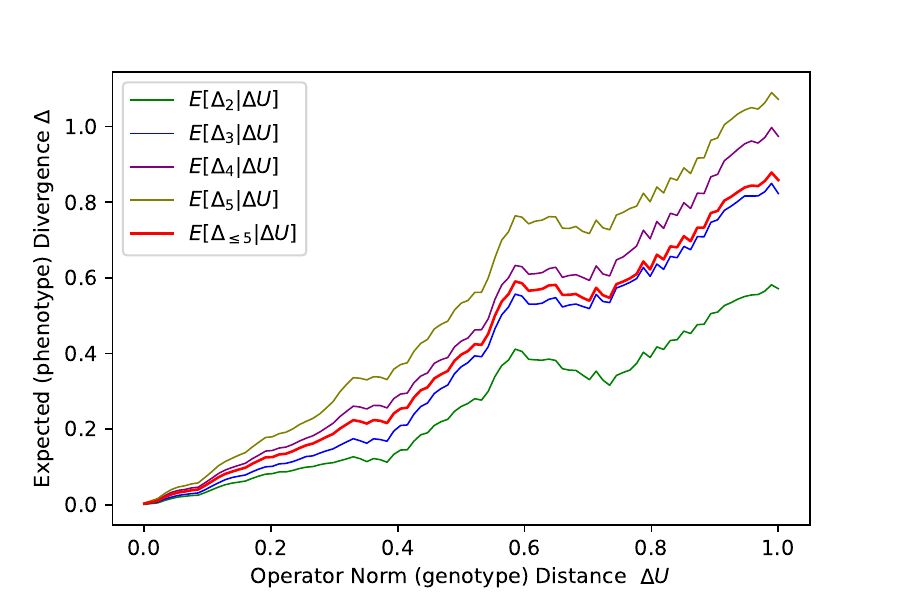}
        \caption{Smoothness of Fitness Landscape}
        \label{fig: Landscape Continuity}
    \end{minipage}
\end{figure}

\begin{figure}[ht]
    \centering
    \begin{minipage}{.5\textwidth}
        \centering
        \includegraphics[width=1.1\linewidth, height = 0.7\linewidth]{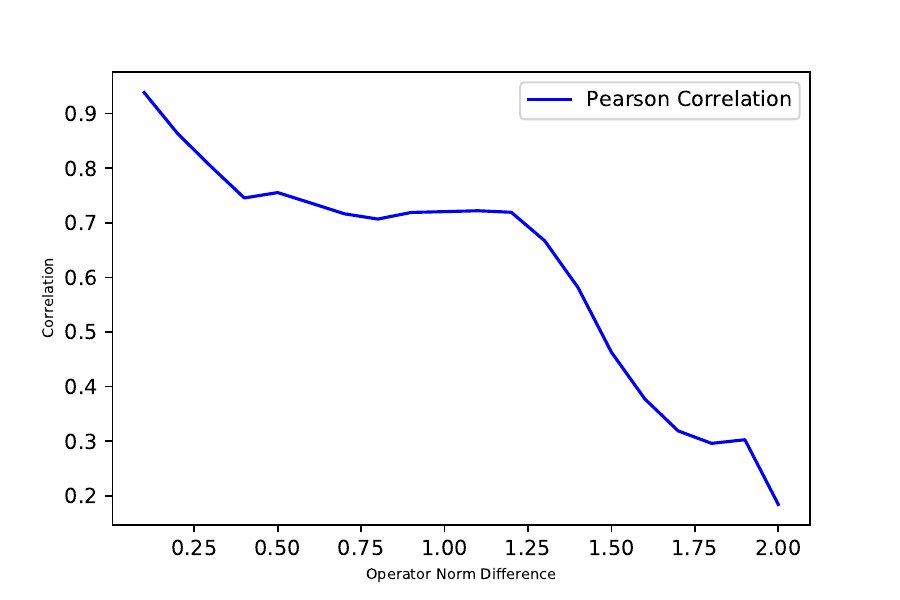}
        \caption{Correlation of Fitness Landscape}
        \label{fig: Landscape Correlations}
    \end{minipage}
\end{figure}

On \fig{ Landscape Correlations} we present estimations of the  Pearson correlation coefficient of the fitness of a hypothesis and its distance to the optimal solution for sequences with length up to 5 symbols: 
\begin{equation}
    \label{eqn:Pearson}
    r = \frac{\operatorname{cov(\Delta_{\leq 5}({D^*}||{D^t})), \Delta U)}}{\operatorname{std(\Delta_{\leq 5}({D^*}||{D^t}))std( \Delta U)}}
\end{equation}

The coefficient is estimated for three mutation rates and can be used as measure of the expected effectiveness of the evolutionary algorithms applied to the QHMM learning problem.

\subsection{\label{subsection:qhmm_evolutionary_learning}QHMM Learning Algorithm}
The learning algorithm takes as input an estimation of the Hankel matrix associated with the unknown process language $\hat{H}^L \in R^{\Sigma^{\leq n} \times \Sigma^{\leq n}}$ and the specification of a hypothesis space $\Q$. Two of the components of the hypothesis space are fixed by the  Hankel matrix $\hat{H}^L$: the alphabet $\Sigma$ and the dimension of the state Hilbert space $\H_S$ which is $\hat{N}=\operatorname{rank}(\hat{H}^L)$. The free parameters, subject of optimization, are the dimension $M$ of the emission Hilbert space $\H_E$, the unitary $U$, the mapping $\M$, and the type of the initial density $\rho_0$. In order to simplify the discussion and focus on the essential part of the algorithm, we will assume that only the unitary transformation $U$ will be learned.   

The hypothesis space $\Q$ comprises both discrete and continuous components, represented by structural and parametric subspaces. The structural subspace contains the non-parametric parts $g = \langle (t),(q),(.)\rangle$ of the linear quantum circuits' $\{\mathcal{C}_U \}$ as defined in (\ref{eqn:c_u}).
The parametric subspace is defined by vectors $\{\mathcal{P}_U \}$ with circuit gates' parameters:
\begin{equation*}
    \label{eqn:p_u}
    \mathcal{P}_U=(p : \exists g \in\mathcal{C}_U, g = \langle (.),(.),(p) \rangle.  
\end{equation*}
\noindent
The algorithm performs evolutionary global search in the structural subspace as the fitness of every circuit structure is estimated by parametric optimization of (\ref{eqn:fitness}):
\begin{equation}
    \label{eqn:fit_c}
    Fit(\mathcal{C}_U) = \underset{\mathcal{P}_U}{\mathrm{max}}F(Q_U),
\end{equation}
where $\mathcal{C}_U$ is the unitary encoding of the hypothesis $Q_U$.
The optimal parameters ${P}_U^* = \underset{\mathcal{P}_U}{\mathrm{argmax}}F(Q_U))$ of every hypothesis $Q_U$ are calculated by local multivariate nonlinear optimization procedure. Derivative-free solvers as $Powell, TNC, Cobyla$ as well as gradient- based $BFSG, GC, SLSQP$ have been used. The type of solver for each structure optimization is selected by an adaptive local search optimization procedure. After the optimization the optimal parameters ${P}_U^*$ become parameters of the hypothesis and the optimal fitness value becomes fitness of the hypothesis. The combination of evolutionary global search in the space of quantum circuit structures and  local parametric optimization classifies our learning method as \textit{Lamarckian learning} \cite{Lamarckian}.  

\begin{algorithm}[ht]
 \LinesNotNumbered
    \caption{Random Hypothesis}
    \label{alg:random}
        \begin{algorithmic} 
        \Procedure{RandomGate}{$\Q$}
            
            \State {$\textbf{t} \sim \operatorname{gatesDistribution}( \mathcal{G} )$} \qquad \qquad \Comment{Select gate's type}
            \State {$q_c, q_d \sim \operatorname{qubitsDistribution} (N, M)$} \qquad  \qquad \Comment{Select control and data qubits} 
            \State {$p_1,p_2 \sim \operatorname{parametersDistribution} ([0,\,8\pi])$}\qquad\Comment{Select gate parameters} 
            \State {$gate \gets \langle \textbf{t}, (q_c, q_d), (p_1,p_2) \rangle$} \qquad \quad \qquad \quad \qquad \quad\Comment{Create a gate}
            \State \Return{} $gate$ 
        \EndProcedure       
        \Procedure{RandomHypothesis}{$\Q$}
            \State {$minGts, maxGts$ $\gets$ $const_1, const_2$} \qquad \qquad \qquad \quad\Comment{Min-Max number of gates}
            \State {$numGts$ $\sim$ $\operatorname{Uniform}$($minGts, maxGts$)} \qquad \Comment{Random number of gates}
            \State {$\mathcal{C} \gets$ [\,\,\.]}\qquad \qquad\qquad \qquad \qquad\qquad \qquad \qquad  \Comment{Initial circuit for the hypothesis}
            \For{$g = 1$ \KwTo $numGts$} {
              \State{\hspace{\algorithmicindent} $\mathcal{C}$ $\gets$ $\mathcal{C}$ + [\Call{RandomGate}{$\Q$}]}
            }
            \State $\mathcal{P}_{\mathcal{C}}^{*}$ $\gets$ $\underset{\mathcal{P}_{\mathcal{C}}}{\mathrm{argmax}}\hspace{0.1cm} F(\mathcal{C})$  \Comment{Find parameters $\mathcal{P}_{\mathcal{C}}^*$ of $\mathcal{C}$ which maximize the fitness}
            \State $F_{\mathcal{C}}$ $\gets$ $ F(\mathcal{C}^*)$  \Comment{Fitness value of $\mathcal{C}^*$ } 
            \State \Return{} $\mathcal{C}^*$ 
        \EndProcedure         
    \end{algorithmic}
\end{algorithm}

\begin{algorithm}[ht]
    \LinesNotNumbered
    \SetKwInOut{Input}{Input}\SetKwInOut{Output}{Output}
    \DontPrintSemicolon
    \caption{Evolutionary Learning}\label{alg:evo_qhmm}
    
    \Input{$\mu, \lambda, \Q$ (\ref{eqn:hyp_space}), $F(Q)$ (\ref{eqn:fitness}), $F^*$, $gMax$, $\mathcal{G}$ }
    \Output{Best Hypothesis $Q^{*} = \underset{Q \in \Q}{\mathrm{argmax}} \hspace{0.1cm} F(Q)$}
    
    \textbf{Initialize} 
    \begin{itemize} 
        \setlength{\itemsep}{0pt}%
        \setlength{\parskip}{0pt}%
        \item {$g \gets 0$}   \Comment*[f]{ Evolutionary generation} 
        \item {$t \gets 0$}   \Comment*[f]{  Search progress steps  } 
        \item {$P_g \gets \Call{Random}{\mu, \Q}$} \Comment*[f]{Initial population: $\mu$ random hypotheses} 
         
    \end{itemize}
    
    \While{ $\underset{Q \in P_g}{\mathrm{max}}\hspace{0.1cm} F(Q) < F^* \operatorname{and} g < gMax$}{
        $\tau       \gets \mathbf{Temperature}(t)$ \\ 
        $Parents \gets \mathbf{Select} (P_g,\lambda, selectionDistribution)$ \\
        $Children \gets \mathbf{Modify}(Parents,\tau, modificationDistributions)$ \\
        $P_{g+1}   \gets \mathbf{Select} (P_g\cup Children,\mu, survivalDistribution)$ \\
        $g \gets g + 1$ \\
        \Comment*[f]{Check for search progress during last $progWin$ steps } \\
        \eIf{ $\mathbf{noProgress}(progWin)$} 
            {
                $t\gets 0$
            }{
                    $t\gets t+1$
            }
        
        $Distributions \gets \mathbf{Adapt}(Distributions)$
    }
    
    \Return{$\underset{Q \in P_g}{\mathrm{argmax}}\hspace{0.1cm} F(Q)$ }
\end{algorithm}

\begin{algorithm}[ht]
 \LinesNotNumbered
    \caption{Mutate a Hypothesis}
    \label{alg:mutation}
        \begin{algorithmic} 
        \Procedure{MutateHypothesis}{$\mathcal{C}$, $pos$, $mType$}
        
          \Switch{$mType$}{
            \hspace*{0.25cm}\textbf{when} $"gte"$ $:$ $\mathcal{C}[pos].gate \sim \operatorname{gatesDistribution}( \mathcal{G} )$ \Comment{ Mutate gate's type }\par
            \hspace*{0.25cm}\textbf{when} $"qbt"$ $:$ $\mathcal{C}[pos].qubits \sim \operatorname{qubitsDistribution}( N, M )$ \Comment{ Mutate gate's qubits }\par
            \hspace*{0.25cm}\textbf{when} $"rpl"$ $:$ $\mathcal{C}[pos] \gets  \Call{RandomGate}{\Q}$ \Comment{ Replace gate }\par
            \hspace*{0.25cm}\textbf{when} $"dlt"$ $:$ $\mathcal{C}[pos:] \gets \mathcal{C}[pos+1:]$  \Comment{ Delete gate }\par
            \hspace*{0.25cm}\textbf{when} $"ins"$ $:$ $\mathcal{C} \gets \mathcal{C}[:pos] + [\Call{RandomGate}{\Q}] + \mathcal{C}[pos:]$  \Comment{ Insert gate }\par
         }
        \Return{$\mathcal{C}$}
        \EndProcedure         
    \end{algorithmic}
\end{algorithm}

\begin{algorithm}[ht]
 \LinesNotNumbered
    \caption{Hypothesis Modification}
    \label{alg:Hypothesis Modification}
        \begin{algorithmic}      
        \Procedure{ModifyHypothesis}{$\mathcal{C}$}
            \Indp            
            \State {$\mathcal{C}_{best}$ $\gets$ $\mathcal{C}$ } \Comment{Current best hypothesis}

            \State {$\mathcal{C}_{current}$ $\gets$ $\mathcal{C}$ } \Comment{Current hypothesis}

            \State {$searchSteps \sim \operatorname{localSearchLen}(1..max)$}\qquad \Comment{Random local search steps}

            \For{$step = 1$ \KwTo $searchSteps$} {
                \Indp
                \State {$sType \sim \operatorname{localSearchType}(\mathbf{"depth"}, \mathbf{"breadth"})$}  \Comment{local search type}
                \State {$mProb \sim \operatorname{mutationRate}(0.1, ..., 0.5)$}  \Comment{Mutation rate }
                \State {$\mathcal{C}$ $\gets$ $\mathcal{C}_{current}$ } \Comment{Current hypothesis to search around}
            
                \For{ $pos = 1$ \KwTo $\operatorname{Size}(\mathcal{C})$}
                 { 
                    \If{$\Call{Random}{}> mProb$} {
                        \State \hspace*{0.19cm} {$mType \sim \operatorname{mutationType}("gte", ..., "ins")$}  \Comment{Mutation type }
                         \Call{MutateHypothesis}{$\mathcal{C}$, $pos$, $mType$}  \Comment{Mutatate hypothesis at pos}
                    }
                }
                \State $\mathcal{P}_{\mathcal{C}}^{*}$ $\gets$ $\underset{\mathcal{P}_{\mathcal{C}}}{\mathrm{argmax}}\hspace{0.1cm} F(\mathcal{C})$  \Comment{ Find parameters $\mathcal{P}_{\mathcal{C}}^*$ which maximize the fitness}
                
                \State $F_{\mathcal{C}}$ $\gets$ $ F(\mathcal{C})$  \Comment{Fitness value of $\mathcal{C}$ with optimal parameters $\mathcal{P}_{\mathcal{C}}^{*}$ }

                \If{$\Call{F}{\mathcal{C}}> \Call{F}{\mathcal{C}_{best}}$} {
                    \State {$\mathcal{C}_{best}$ $\gets$ $\mathcal{C}$ }  \Comment{ Save the new best hypothesis}}
                \If{$\Call{Accept}{\operatorname{F}(\mathcal{C}),\operatorname{F}(\mathcal{C}_{current}), Temperature} $} {
                    \State {$\mathcal{C}_{current}$ $\gets$ $\mathcal{C}$ }   \Comment{ New current  hypothesis}}            
            }

            \If{$\Call{Accept}{\operatorname{F}(\mathcal{C}),\operatorname{F}(\mathcal{C}_{best}), Temperature} $}
            {
                \State { \hspace{0.5cm} $\mathcal{C}_{best}$ $\gets$ $\mathcal{C}$ }   \Comment{ new preferred hypothesis}} 
            \State \Return{} $\mathcal{C}_{best}$
        \EndProcedure         
    \end{algorithmic}
\end{algorithm}

The algorithm starts with random generation of a finite sample of $\mu$ hypotheses forming the initial \textit{population} (Algorithm \ref{alg:random}). The process of iterative improvement of the population consists of three base randomized steps, including selection of set of existing hypotheses to be improved referred to as parents, modification of the parents to generate offspring or children, selection from the parents and children a new generation of the population (Algorithm \ref{alg:evo_qhmm}). The modification operations include range of hypothesis mutations (Algorithm \ref{alg:mutation}) utilized by a stochastic local search (Algorithm \ref{alg:Hypothesis Modification}). 
Since the random operations selecting offspring and survivors are biased towards hypotheses with higher fitness (\textit{survival of the fittest}), the entire population evolves towards regions of the space $\Q$ which contain better potential solutions and eventually an optimal solution. The trade-off between exploration and exploitation within a hypothesis space region is controlled by a global variable referred to as  $\textit{Temperature}$. The temperature defines the probability of accepting a new hypothesis, even if it is not superior to its parent. The temperature gradually decreases with the search progression, thereby increasing the likelihood of selecting only better solutions. If a certain threshold of steps is reached without any fitness improvement, the temperature is reset to its highest value to allow exploration of new regions. The temperature($\tau$) is calculated as a function of the number of steps $t$ (FIG (\ref{fig:temp_control})) as follows:

\begin{equation*}
\tau={(t^{\frac{3}{2}}+1)}^{-\frac{1}{4}}.    
\end{equation*}

The probability a new hypothesis with fitness $F_{new}$ to be accepted instead of its superior parent with fitness $F_{old} \geq F_{new}$ at temperature $\tau$ is defined by:
$$P_{new}= \exp(\frac{0.6}{\tau}\frac{F_{old}-F_{new}}{F_{old}})$$
\begin{figure}[ht]
    \centering
    \begin{minipage}{.5\textwidth}
        \centering
        \includegraphics[width=1.0\linewidth]{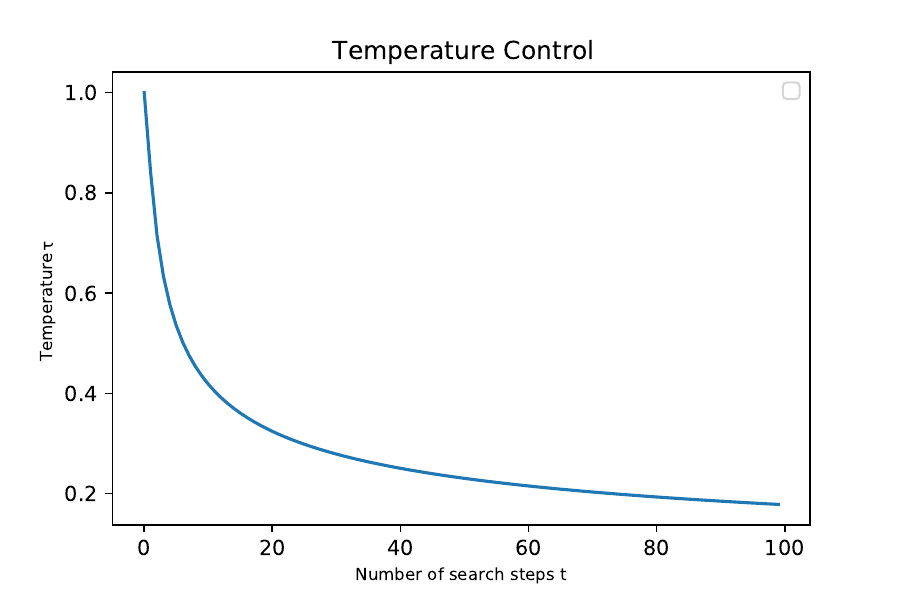}
        \caption{Temperature control}
        \label{fig:temp_control}
    \end{minipage}
\end{figure}

\begin{figure}[ht]
    \centering
    \begin{minipage}{.5\textwidth}
        \centering
        \includegraphics[width=1.0\linewidth]{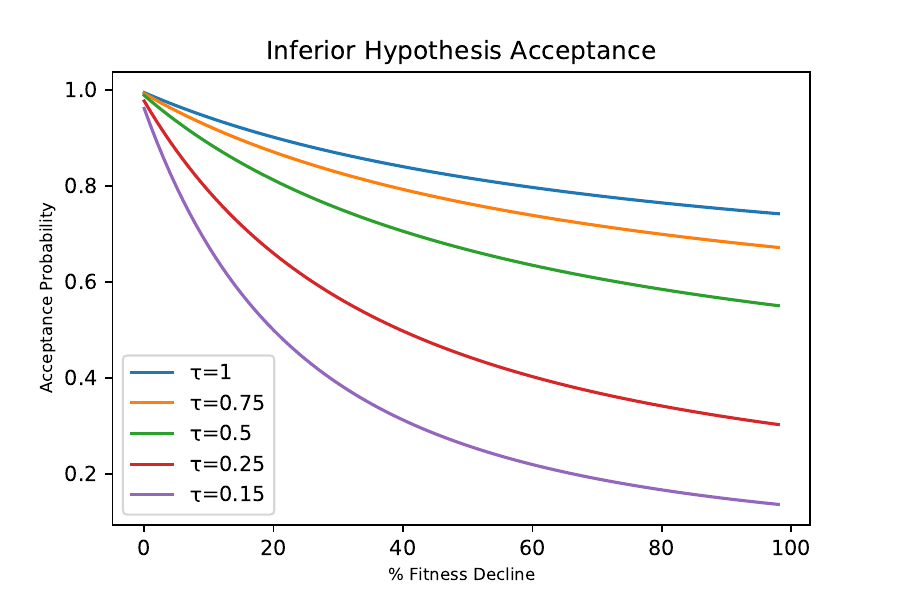}
        \caption{Inferior Hypotheses Acceptance}
        \label{fig:Hypothesis_Acceptance_Probability}
    \end{minipage}
\end{figure}

\begin{figure}[ht]
    \centering
    \begin{minipage}{.5\textwidth}
        \centering
        \includegraphics[width=1.0\linewidth]{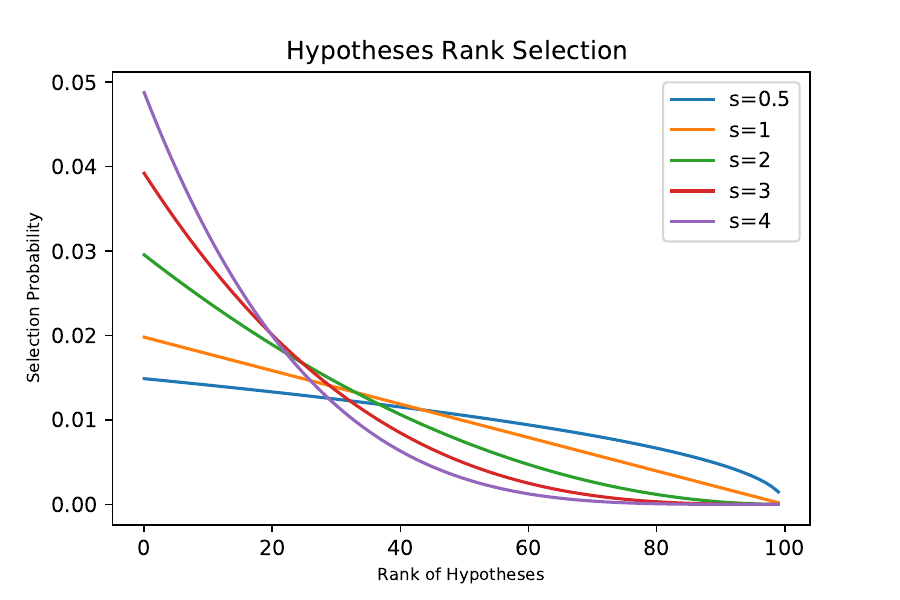}
        \caption{Rank Selection Distributions}
        \label{fig:Hypothesis selection distribution}
    \end{minipage}
\end{figure}

\begin{figure}[ht]
    \centering
    \begin{minipage}{.5\textwidth}
        \centering
        \includegraphics[width=1.0\linewidth]{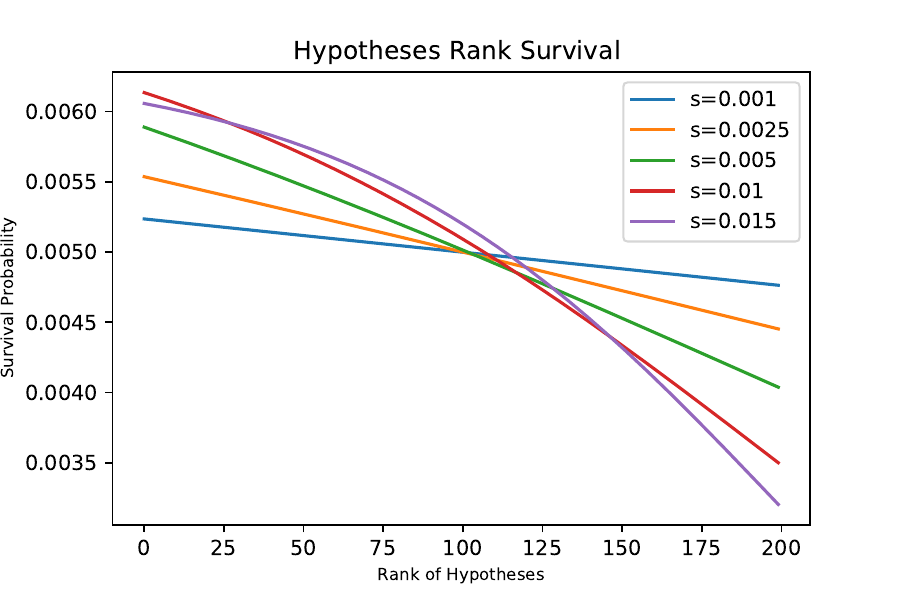}
        \caption{Rank Survival Distributions}
        \label{fig: Hypothesis survival distribution}
    \end{minipage}
\end{figure}

\begin{figure*}[ht]
    \centering
        \includegraphics[scale=1]{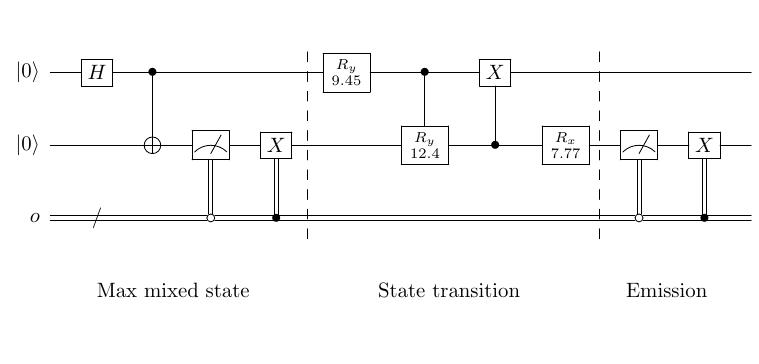}
         \caption{QHMM Defining Market Distribution}
    \label{fig:example31fig}
\end{figure*}

\begin{figure*}[ht]
        \centering
        \includegraphics[scale=1]{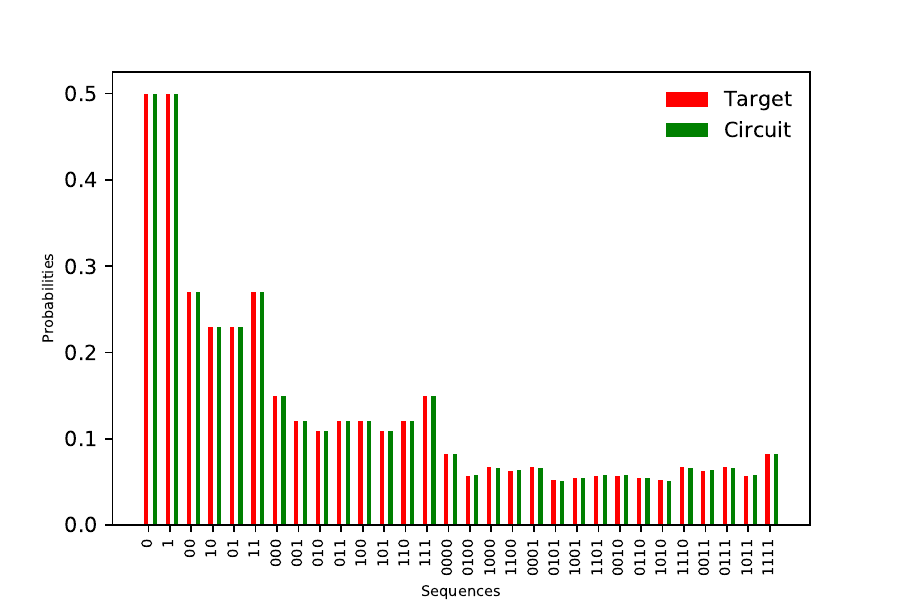}
    \caption{Targeted and Learned Market Distributions}
    \label{fig:example32fig}
\end{figure*}

The control of hypothesis acceptance probability by the global temperature $\tau$ is illustrated \fig{Hypothesis_Acceptance_Probability}.

The selection of a set of hypotheses, referred to as parents, to be modified and evaluated as potential new members of the population relies on the fitness of these hypotheses. The candidates with better fitness are more likely to be selected for further consideration. This approach can be successful in case of unimodal landscapes. However, if the fitness landscape is rugged, characterized by numerous local extrema, employing a uniform, strong selection pressure strategy can result in premature convergence towards a local solution. Therefore, in our algorithm we use selection distributions based on rank, fitness, and tournaments.
Every selection distribution type is used with up to 5 levels of selection pressure. For example, let's consider rank-based selection where the hypothesis are ranked according their fitness values with the fittest hypothesis receiving the lowest rank. Then the probability a hypothesis with rank  $i$  to be selected for offspring generation is:
$$
 P_{selection}[Q_i]= \frac{(\mu - i)^s}{\sum_{r=1}^{\mu}(\mu - r)^s},
$$
where $\mu$ is the size of the population and $s$ is a parameter called \textit{selection strength}. The shape of the rank selection distribution for different values of the parameter is shown on \fig{Hypothesis selection distribution}. 

The survival distributions define the chance a hypothesis - old or newly generated - to participate in the new generation. These distributions are also based on rank and fitness and use up to 5 levels of selection pressure. The probability a hypothesis with rank  $i$  to be selected to ``survive'' in the next generation is:
\begin{equation}
 P_{survival}[Q_i]= \frac{(d_i+1)^{-1}}{\sum_{r=1}^{\mu}(d_r+1)^{-1}}, 
\end{equation}
where
\begin{equation}
 d_r = \exp{s(r-\mu - \lambda)},
\end{equation}
$\mu$ is the size of the population, $\lambda$ is the size of the offspring,  and $s$ is a parameter called \textit{survival strength}. The shape of the rank survival distribution for different values of the parameter is shown on \fig{ Hypothesis survival distribution}.

The evolutionary algorithm uses multiple distributions to implement the processes of selection, modification, and survival (TABLE~\ref{tab:distributions descriptions}). The parameters of these distributions along with the population size $\mu$ and the offspring size $\lambda$ are referred to as \textit{hyperparameters} of the algorithm. 
The hyperparameters are critical for the performance of the algorithm. Since each learning task is unique and the fitness landscape can change during the evolutionary process, it is essential to dynamically adjust the hyperparameters as the learning evolves. To achieve adaptive control over the hyperparameters, we employ a reinforcement learning algorithm based on the "multi-armed bandit model". This algorithm effectively addresses the exploration-exploitation trade-off, and maximizes the cumulative reward obtained through the parameters adaptation.   
Let's assume that the evolutionary algorithm utilizes a procedure or distribution $A(h)$ which depends on a parameter $h \in \{h_1, \cdots, h_k\}$. For a generation $t$ we define a distribution $D_{t}^h = \{ p_i: p_i = P[h_i]\}$. Every time when $A$ is used during the generation $t$ we first sample a value $h \sim D_{t}^h$ and then use it for this particular activation $A(h)$. We register if the particular usage of $A(h)$ has been successful or not. Success is considered when the fitness of a hypothesis or a population has increased and it this case the value of $h$ receives reward $1$. in the end of the generation $t$ we will have a reward vector  $\mathbf{r}_t = [r_t^{1}, \cdots, r_t^{k}]$, where $r_t^{j}$ is the number of times the usage of the value $h=h_j$ during the generation $t$ has resulted in fitness improvement. Let the initial distribution be the uniform distribution  $D_{0}^h = \{\frac{1}{k}, \cdots, \frac{1}{k} \}$. The distribution at generation $t+1$ is defined as follows:
$$D_{t+1}^h = \Biggl\{ p_i: p_i = \gamma\frac{1}{k}+(1-\gamma)\frac{r_t^{i}}{\sum_{j=1}^{k}r_t^{j}}, i=[1,k]\Biggr\},$$
\noindent
where $\gamma \in [0, 1]$ quantifies the trade-off between exploration and exploitation. The pseudo code of the algorithm is presented as Algorithm \ref{alg:evo_qhmm}.

\begin{table*}[ht]
        \begin{tabular}{|l|l|l|}
            \hline
            \hspace*{1.5mm}{$Distribution$}\hspace*{1.5mm} &  \hspace*{1.5mm}Description &  \hspace*{1.5mm}Domain  \\
            \hline
            selectionDistributionTypes &   \hspace*{1.5mm}Probabilities of  selection types. & $[Fitness, Rank, Tournament]$\\
            selectionDistributionStrength &   \hspace*{1.5mm}Probabilities of selective pressure levels &  $[0.1, 0.2, 0.5, 0.7, 1]$\\
            survivalDistributionTypes&   \hspace*{1.5mm}Probabilities of survival types & $[Fitness, Rank]$\\
            survivalDistributionStrength&   \hspace*{1.5mm}Probabilities of survival pressure levels. &  $[0.1, 0.2, 0.5, 0.7, 1]$\\
            gatesDistribution &   \hspace*{1.5mm}Probabilities of gate types &  $\mathcal{G}=\{g_0 \cdots g_k  \}$\\
            qubitsDistribution &   \hspace*{1.5mm}Probabilities  of qubits &  $[1,n+m], P(n+m,2) $\\
            localSearchLen Distribution &   \hspace*{1.5mm}Probabilities  of local search steps &  $[1,10]$\\
            localSearchType Distribution &   \hspace*{1.5mm}Probabilities  of local search types &  $[depth,breadth]$\\ 
            mutationRate Distribution &   \hspace*{1.5mm}Probabilities  of hypothesis' mutation rates &  $[0.1, ... 0.5]$\\ 
            mutationType Distribution &   \hspace*{1.5mm}Probabilities  of hypothesis' mutation types &  $['gte','qbt','rpl','dlt','ins']$\\
            optimizationAlgs Distribution &   \hspace*{1.5mm}Probabilities of nonlinear solvers &  $['tnc','cbla','bfsg','gc','slsqp']$\\
            \hline
        \end{tabular}
        \caption{Adaptive distributions used by the evolutionary algorithm}
        \label{tab:distributions descriptions}
    \end{table*}

\begin{example}
\label {exm: market}
    \normalfont
    In this example we apply the QHMM learning algorithm to the Market classical HMM discussed in \textbf{Example}~\ref{example1}.
    The minimal order of the classical model estimated using the Hankel matrix is: $n=4$.
    The learning sample as distributions of sequences with lengths up to $t=2n-1=7$ was generated using the string
    function of the classic model (\eqn{hmm_sequence_function}).
    The QHMM sequences distribution is generated using the sequence function (EQ~\ref{eqn:hmm_sequence_function}).
    The preliminary model specification based on the learning sample analysis is the following:
    \begin{equation*}
        \textbf{Q}=\{ \Sigma, \H_{S}, \H_{E}, U, \mathcal{M} , \rho_0\}
    \end{equation*}
    where $\Sigma=\{0,1\}$, $\H_{S}$ is a Hilbert space with dimension $N=\sqrt{n}=2$ (i.e. 1 qubit quantum system),
    $\H_{E}$ is a Hilbert space with dimension $2$ (i.e. 1 qubit quantum system and the measurement basis is $\{ [0],\, [1] \}$), $U$ is a unitary operation on the Hilbert space $\H_E \otimes \H_S$ implemented by quantum
    gates in $\mathcal{G}=\{X,Y,RX,RY\}$, and $\rho_0$ is the tensor product of uniformly distributed state and
    grounded emission systems.

    Exact match (i.e zero-divergence between classic and quantum distributions) was reached in average of 150
    iterations by population of 100 solutions.
    The best solution is shown in \fig{example31fig}.
    Part of the target and learned distributions at the end of the search are shown in \fig{example32fig}.
    
\end{example}

    \begin{table*}[ht]
        \begin{tabular}{c|l}
            \hline
            \hspace*{1.5mm}{$S$}\hspace*{1.5mm} &  \hspace*{1.5mm}Description  \\
            \hline
            0 &   \hspace*{1.5mm}Low Variance \\
            1 &   \hspace*{1.5mm}Low-Medium Variance\\
            2 &   \hspace*{1.5mm}Medium-High Variance\\
            3 &   \hspace*{1.5mm}High Variance \\
            \hline
        \end{tabular}
        \hfill
        \begin{tabular}{|c|c|c|c|c|}
            \hline
            \hspace*{1.5mm}{$S$}\hspace*{1.5mm} & \hspace*{3.5mm}{0}\hspace*{3.5mm} & \hspace*{3.5mm}{1}\hspace*{3.5mm}
            & \hspace*{3.5mm}{2}\hspace*{3.5mm} & \hspace*{3.5mm}{3}\hspace*{3.5mm}  \\
            \hline
            0 &  0.60&	0.25&	0.05&	0.10\\
            1 &  0.05&  0.15&   0.05&   0.75\\
            2 &  0.75&  0.05&   0.15&   0.05\\
            3 &  0.10&  0.05&   0.65&   0.20 \\
            \hline
        \end{tabular}
        \hfill
        \begin{tabular}{|c|c|c|c|c|}
            \hline
            \hspace*{1.5mm}{$S$}\hspace*{3.5mm}& \hspace*{3.5mm}{0}\hspace*{3.5mm}& \hspace*{3.5mm}{1}\hspace*{3.5mm}& \hspace*{3.5mm}{2}\hspace*{3.5mm} & \hspace*{3.5mm}{3}\hspace*{3.5mm}   \\
            \hline
            0 &  0.00&	0.50& 0.50& 0.00\\
            1 &  0.01&	0.49& 0.49& 0.01\\
            2 &  0.13&	0.37& 0.37& 0.13\\
            3 &  0.22&	0.28& 0.28& 0.22\\
            \hline
        \end{tabular}
        \caption{Gaussian Mixture HMM. Left: Hidden States Descriptions. Center: States Transition Probabilities. Right: Emission Probabilities.}
        \label{tab:hmm_gauss_mixture_example}
    \end{table*}

\begin{figure*}[ht]
    \centering
    \begin{minipage*}{.49\textwidth}
        \centering
        \includegraphics[width=0.8\linewidth]{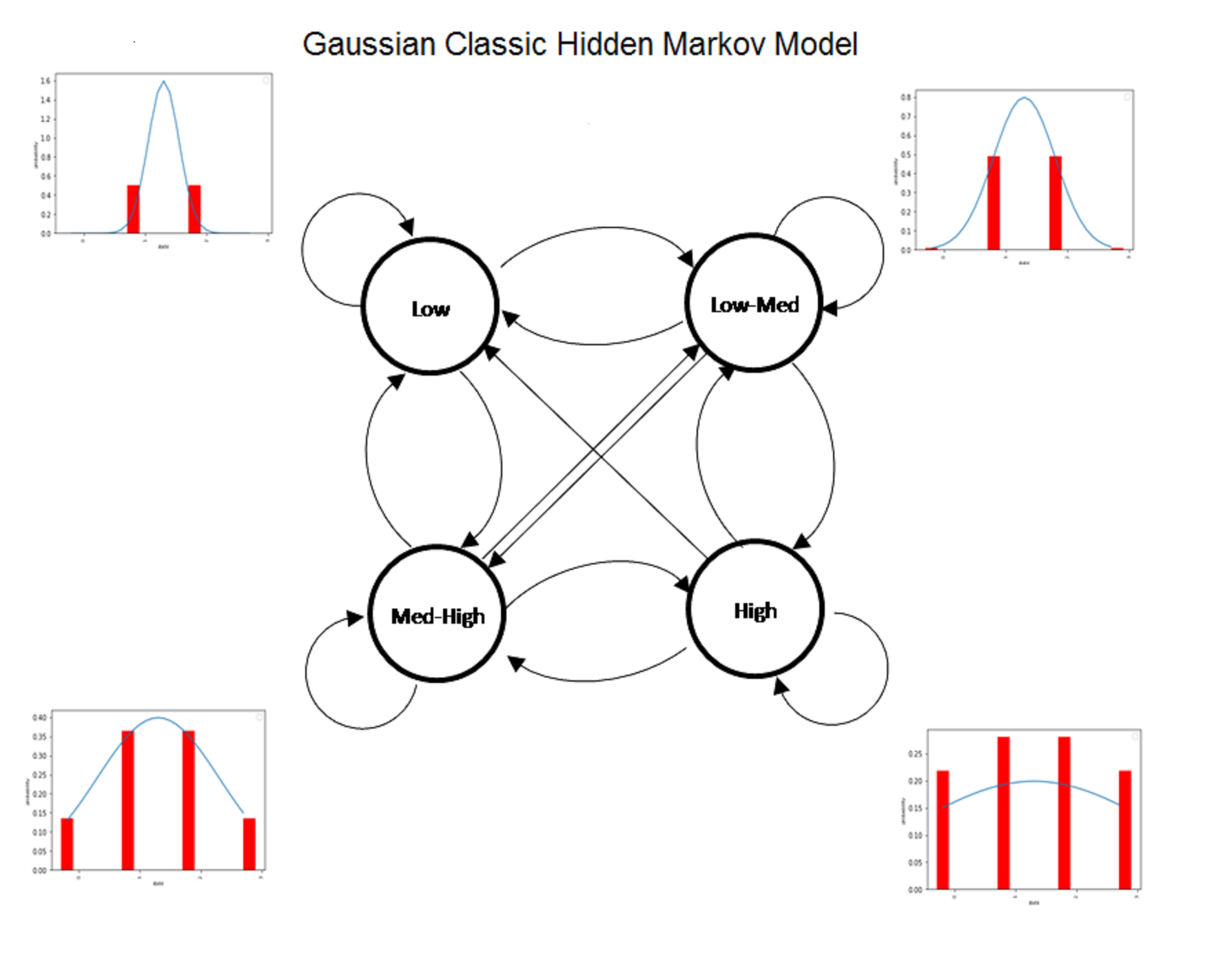}
    \end{minipage*}
    \begin{minipage}{.49\textwidth}
        \centering
        \includegraphics[width=1.0\linewidth]{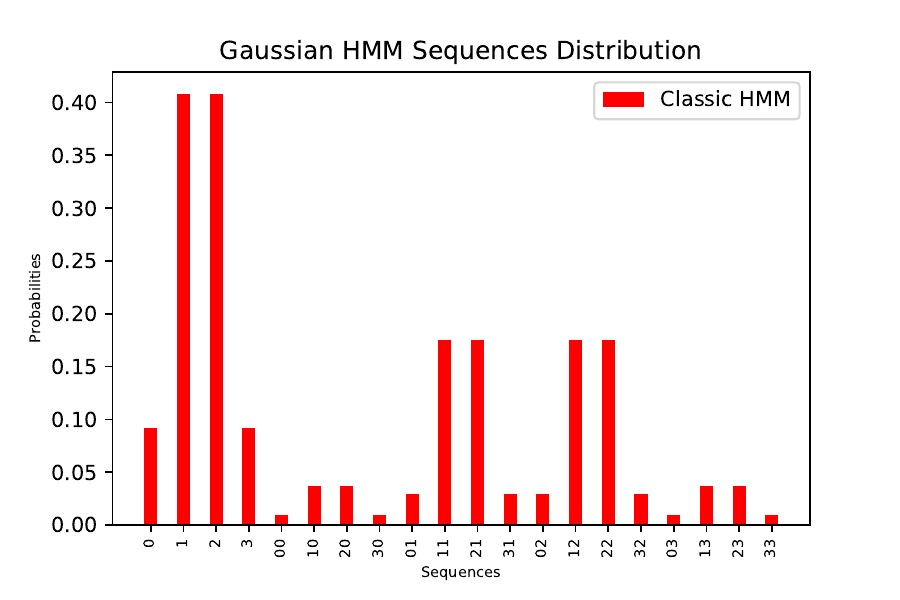}
    \end{minipage}

    \caption{Gaussian Mixture HMM}
    \label{fig:hmm_gauss_mixture}
\end{figure*}

    \begin{figure*}[ht]
        \centering
            \includegraphics[scale=1.0]{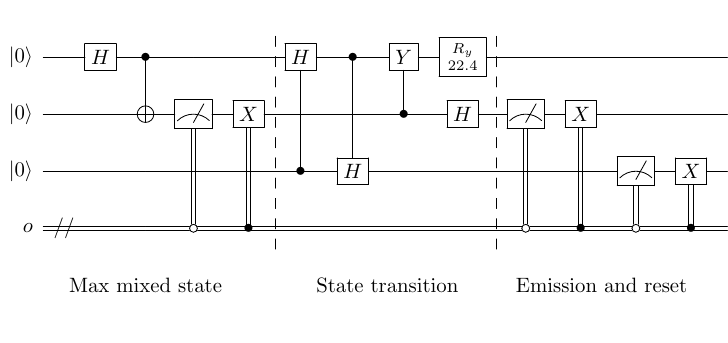}
        \caption{Gaussian QHMM}
        \label{fig:exampleGaussMix1}
    \end{figure*}
    
    \begin{figure*}[!tbp]
            \centering
            \includegraphics[scale=0.8]{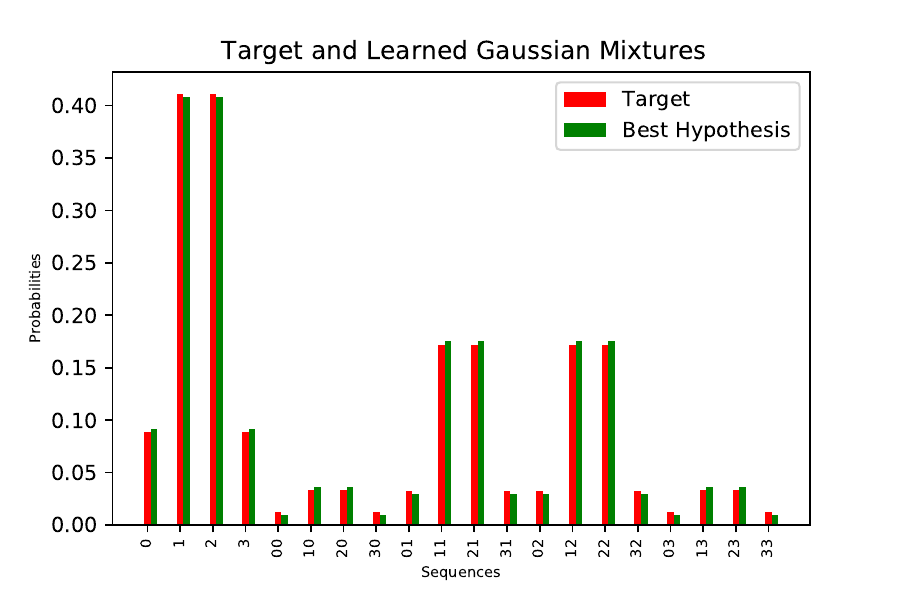}
        \caption{Learned Gaussian Mixture}
        \label{fig:exampleGaussMix2}
    \end{figure*}

In the next example we demonstrate that the QHMMs are efficient generators of stochastic mixtures of distributions. 
\begin{example}
    \label{example Gaussian Mixture}
    \normalfont
    Let's consider a process each point of which is drawn from a Normal distribution belonging to a finite set of Normal distributions:
            $$\{ Y_t: Y_t \sim \mathcal{N}(\mu, \,\sigma^2(s), \, s \in S) \}, $$ 
    \noindent
  where $S =\{s_1, \ldots, s_N\}$ is a finite set of unobservable states. 
  The classical HMM of such a process with n=4 hidden states, defining  standard deviations correspondingly $[0.5,\, 1,\, 2,\, 4 ]$ is specified in Table~\ref{tab:hmm_gauss_mixture_example}. The model is minimal and in each state $s \in S$ an integer observable in the interval $[0,3]$ is emitted by sampling from a normal distribution $\mathcal{N}(\mu=1.5, \, \sigma(s))$ as shown in \fig{hmm_gauss_mixture}. The learning sample is a set of distributions of sequences with lengths up to $t=2N-1=7$. It was generated using the string function of the classic model (\eqn{hmm_sequence_function}).
 The preliminary model specification based on the learning sample analysis is the following:
  \begin{equation*}
        \textbf{Q}=\{ \Sigma, \H_{S}, \H_{E}, U, \mathcal{M} , \rho_0\}
  \end{equation*}
  where $\Sigma=\{0,1\}$, $\H_{S}$ is a Hilbert space with dimension $N=\sqrt{n}=2$ (i.e. 1 qubit quantum system), $\H_{E}$ is a Hilbert space with dimension $4$ (i.e. 2-qubit quantum system and the measurement basis is $\{ [0,\,0],\, [0,\,1],\, [1,\,0],\, [1,\,1] \}$), $U$ is a unitary operation on the Hilbert space $\H_E \otimes \H_S$ implemented by quantum
  gates in $\mathcal{G}=\{X,Y,RX,RY,P\}$, and $\rho_0$ is the tensor product of maximally mixed state system and grounded emission systems.
  Divergence error less than $0.01\%$  was reached in average of 250
  iterations by population of 150 hypotheses. The best solution is shown in \fig{exampleGaussMix1}.  Part of the target and learned distributions at the end of the search are shown in \fig{exampleGaussMix2}.

\end{example}
\subsection{\label{subsection:ansatz_learning_qhmm}Learning QHMM with Ansatz Circuits}

In this paper, we propose the following algorithm for learning \emph{quantum hidden Markov model} using ansatz circuits.

\begin{algorithm}[H]
    \SetKwInOut{Input}{Input}\SetKwInOut{Output}{Output}
    \DontPrintSemicolon
    \caption{Quantum Hidden Markov Model Learning with Ansatz Circuits}
    \label{alg:train_hqmm}

    \Input{table of observed sequences and corresponding probabilities, min-max sequence length (optional)}
    \Output{trained circuit that can be extended to the required sequence length}

    Initialize start state, variational quantum ansatz circuit, start values of the parameters. Multiple
    options for start states and ansatz circuits are discussed in Section~\ref{subsec:primmeassystemdesign}. \;

    $ Cost = \sum_{i} l_i \times \left(p_i^{target} - p_i^{current}\right)^2 $ where $l_i$ is the sequence length\;

    Find optimal parameters for quantum circuit with classical technique, e.g. Nelder–Mead method. \;

\end{algorithm}

This approach essentially utilizes the training of parameterized quantum circuits to model sequence data. 

\begin{figure*}[h!]
        \begin{minipage}{.75\textwidth}
            \centering
            \begin{quantikz}[slice style=blue]
                & \lstick{} & \qw & \ctrl{1} \gategroup[wires=3, steps=3, style={dotted, cap=round, inner sep=7pt}, label style={label position=above, yshift=0.4cm}]{Entanglement block} \gategroup[wires=3, steps=5, style={dotted, cap=round, inner sep=14pt}, label style={label position=below, yshift=-0.75cm, xshift=3.0cm}]{N repetitions} & \ctrl{2} & \qw & \qw & \gate{R_y(\theta_{0})} & \qw & \qw & \ctrl{1}\gategroup[wires=3, steps=5, style={dotted, cap=round, inner sep=14pt}, label style={label position=below, yshift=-0.5cm}]{} & \ctrl{2} & \qw & \qw & \gate{R_y(\theta_{0})} & \qw & \qw \\
                & \lstick{} & \qw & \targ{} & \qw & \ctrl{1} & \qw & \gate{R_y(\theta_{1})} & \qw & \qw & \targ{} & \qw & \ctrl{1} & \qw & \gate{R_y(\theta_{1})} & \qw & \qw \\
                & \lstick{} & \qw & \qw & \targ{} & \targ{} & \qw & \gate{R_y(\theta_{2})} & \qw & \qw & \qw & \targ{} & \targ{} & \qw & \gate{R_y(\theta_{2})} & \qw & \qw
            \end{quantikz}
        \end{minipage}
    \caption{Real Amplitudes Ansatz Circuit}
    \label{fig:RealAmplitudes}
\end{figure*}

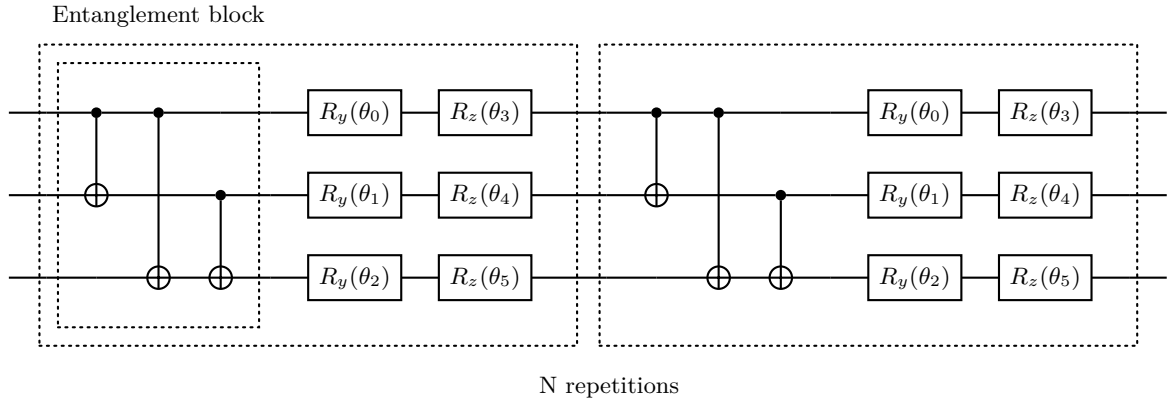
\begin{figure*}[!ht]
        \begin{minipage}{1.\textwidth}
            \centering
            \begin{quantikz}[slice style=blue]
                & \lstick{} & \qw & \ctrl{1} \gategroup[wires=3, steps=3, style={dotted, cap=round, inner sep=7pt}, label style={label position=above, yshift=0.4cm}]{Entanglement block} \gategroup[wires=3, steps=6, style={dotted, cap=round, inner sep=14pt}, label style={label position=below, yshift=-0.75cm, xshift=4.0cm}]{N repetitions} & \ctrl{2} & \qw & \qw & \gate{R_y(\theta_0)} & \gate{R_z(\theta_3)} & \qw & \qw & \ctrl{1}\gategroup[wires=3, steps=6, style={dotted, cap=round, inner sep=14pt}, label style={label position=below, yshift=-0.5cm}]{} & \ctrl{2} & \qw & \qw & \gate{R_y(\theta_0)} & \gate{R_z(\theta_3)} & \qw & \qw \\
                & \lstick{} & \qw & \targ{} & \qw & \ctrl{1} & \qw & \gate{R_y(\theta_1)} & \gate{R_z(\theta_4)} & \qw & \qw & \targ{} & \qw & \ctrl{1} & \qw & \gate{R_y(\theta_1)} & \gate{R_z(\theta_4)} & \qw & \qw \\
                & \lstick{} & \qw & \qw & \targ{} & \targ{} & \qw & \gate{R_y(\theta_2)} & \gate{R_z(\theta_5)} & \qw & \qw & \qw & \targ{} & \targ{} & \qw & \gate{R_y(\theta_2)} & \gate{R_z(\theta_5)} & \qw & \qw
            \end{quantikz}
        \end{minipage}
    \caption{EfficientSU2 Ansatz Circuit with $R_y$ and $R_z$ Rotations}
    \label{fig:EfficientSU2}
\end{figure*}

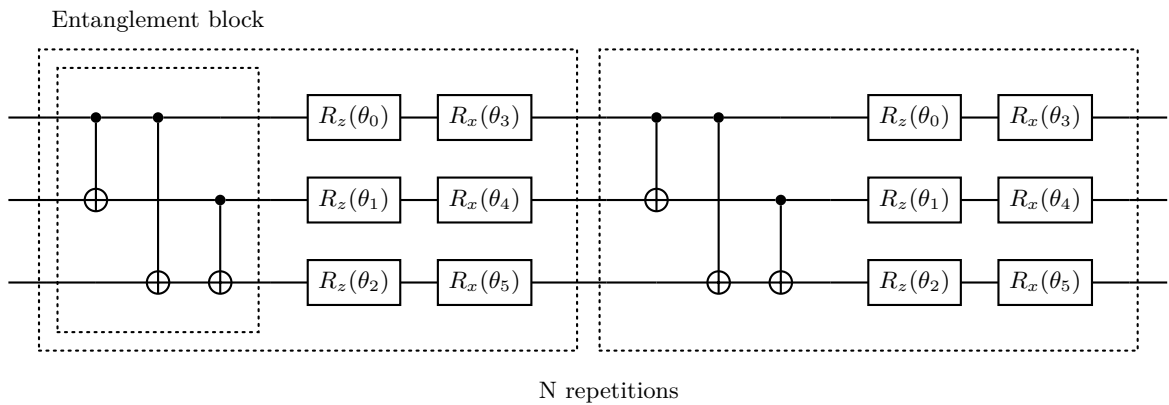
\begin{figure*}[!ht]
        \begin{minipage}{1.\textwidth}
            \centering
            \begin{quantikz}[slice style=blue]
                & \lstick{} & \qw & \ctrl{1} \gategroup[wires=3, steps=3, style={dotted, cap=round, inner sep=7pt}, label style={label position=above, yshift=0.4cm}]{Entanglement block} \gategroup[wires=3, steps=6, style={dotted, cap=round, inner sep=14pt}, label style={label position=below, yshift=-0.75cm, xshift=4.0cm}]{N repetitions} & \ctrl{2} & \qw & \qw & \gate{R_z(\theta_0)} & \gate{R_x(\theta_3)} & \qw & \qw & \ctrl{1}\gategroup[wires=3, steps=6, style={dotted, cap=round, inner sep=14pt}, label style={label position=below, yshift=-0.5cm}]{} & \ctrl{2} & \qw & \qw & \gate{R_z(\theta_0)} & \gate{R_x(\theta_3)} & \qw & \qw \\
                & \lstick{} & \qw & \targ{} & \qw & \ctrl{1} & \qw & \gate{R_z(\theta_1)} & \gate{R_x(\theta_4)} & \qw & \qw & \targ{} & \qw & \ctrl{1} & \qw & \gate{R_z(\theta_1)} & \gate{R_x(\theta_4)} & \qw & \qw \\
                & \lstick{} & \qw & \qw & \targ{} & \targ{} & \qw & \gate{R_z(\theta_2)} & \gate{R_x(\theta_5)} & \qw & \qw & \qw & \targ{} & \targ{} & \qw & \gate{R_z(\theta_2)} & \gate{R_x(\theta_5)} & \qw & \qw
            \end{quantikz}
        \end{minipage}
    \caption{EfficientSU2 Ansatz Circuit with $R_z$ and $R_x$ Rotations}
    \label{fig:qka}
\end{figure*}

\begin{figure*}[!ht]
    \centering
    \includegraphics[scale=0.55]{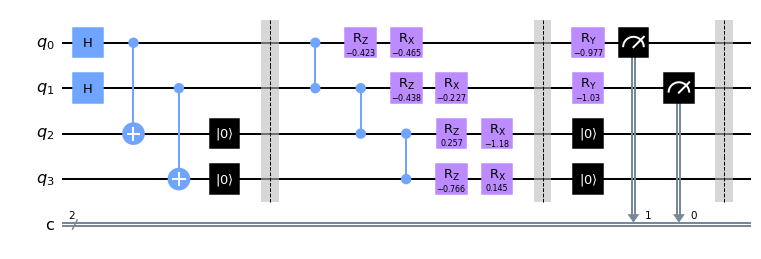}
    \caption{Circuit for 1 step with optimized parameters}
    \label{fig:monras_qka_linear_hw_circuit}
\end{figure*}

\begin{figure*}[!ht]
    \centering
    \includegraphics[scale=0.4]{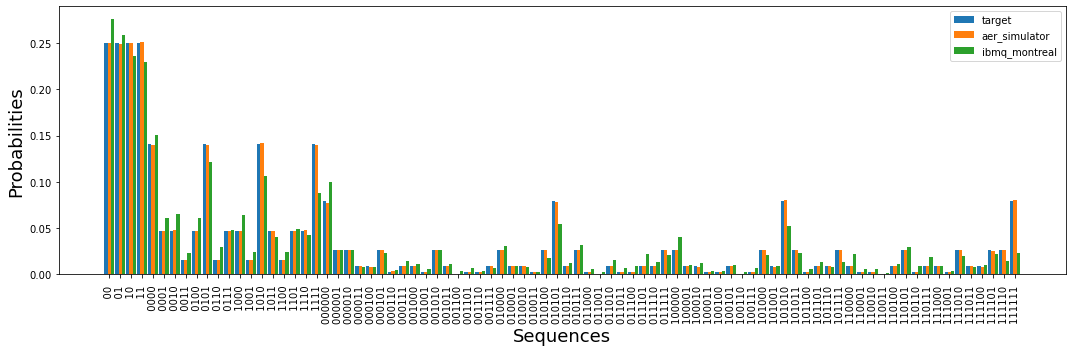}
    \caption{Predicted and observed probabilities}
    \label{fig:monras_qka_linear_hw_probabilities}
\end{figure*}

\subsection{Ansatz Circuit Template}

The topic of efficient ansatz selection is an area of active research \cite{https://doi.org/10.1002/qute.201900070,
    aoki1983, du2022quantum, haug2021capacity, holmes2022connecting}.

Due to the presence of noise in quantum devices, one must select an ansatz which is sufficiently
expressive (i.e. able to access the solution space in the Hilbert space) while maintaining a low parameter count and
lower depth of the quantum circuits in order to suppress the noise.
 
Moreover, the proposed ansatz should contain sufficient entanglement to be non-simulable on a classical computer.
Qiskit provides a variety of built-in ansatz templates with various combinations of the number of parametric gates
and the circuit depth \cite{Qiskit}.
The two ansatz circuits we test in this work are \texttt{RealAmplitudes} (\fig{RealAmplitudes}) and
\texttt{EfficientSU2} (\fig{EfficientSU2} and \fig{qka}). We consider two variations for EfficientSU2 ansatz circuit: 1) with $R_y$, $R_z$ (\fig{EfficientSU2}) and 2) with $R_z$, $R_x$ (\fig{qka}) rotation gates. Furthermore, in the entanglement block we test full and linear entanglement schemes. The latter scheme is more hardware efficient as it doesn't assume fully connected mapping between qubits.

\section{\label{section:examples_of_qhmms}Examples of QHMMs}

\subsection{Simple Four Symbol Stochastic Process}

Let's consider a process defined by the following Kraus operators \cite{monras2011hidden}:
$$\frac{1}{\sqrt{2}}\vert{0}\rangle{}\langle{0}\vert{},\frac{1}{\sqrt{2}}\vert{1}\rangle{}\langle{1}\vert{},\frac{1}{\sqrt{2}}\vert{+}\rangle{}\langle{+}\vert{},\frac{1}{\sqrt{2}}\vert{-}\rangle{}\langle{-}\vert{}$$ or
$$\left(\begin{matrix}\frac{1}{\sqrt{2}}&0\\0&0\end{matrix} \right), \left(\begin{matrix}0&0\\0&\frac{1}{\sqrt{2}}\end{matrix} \right), \frac{1}{2}\left(\begin{matrix}\frac{1}{\sqrt{2}}&\frac{1}{\sqrt{2}}\\\frac{1}{\sqrt{2}}&\frac{1}{\sqrt{2}}\end{matrix} \right),\frac{1}{2}\left(\begin{matrix}\frac{1}{\sqrt{2}}&-\frac{1}{\sqrt{2}}\\-\frac{1}{\sqrt{2}}&\frac{1}{\sqrt{2}}\end{matrix} \right)$$

\begin{table}[!ht]
    \centering
    \begin{tabular}{|l|l|r|}
        \hline
        Ansatz &  Entanglement &  Cost \\
        \hline
        EfficientSU2 ($R_z,R_x$)    & full      & 3.90e-5 \\
        EfficientSU2 ($R_z,R_x$)    & linear    & 4.23e-5 \\
        RealAmplitudes & full      & 0.155  \\
        RealAmplitudes & linear    & 0.689 \\
        EfficientSU2 ($R_y,R_z$)    & full      & 0.156  \\
        EfficientSU2 ($R_y,R_z$)    & linear    & 0.155  \\
    \hline
    \end{tabular}
    \caption{Optimized cost for different ansatz circuits}
    \label{tab:monras_results_comparison}
\end{table}

In this case, we will use the maximally mixed state as the start state.
Next, we model this sequence using the set of ansatz circuits described in Section~\ref{subsec:primmeassystemdesign} and
compare different entanglement strategies that include full entanglement and linear entanglement.
The comparison of the results is presented in Table~\ref{tab:monras_results_comparison}.

Clearly, the EfficientSU2 ($R_z,R_x$) style ansatz performs much better than others.
However, there is only marginal difference between full and linear entanglement strategies. Due to the heavy hex lattice used in IBM processors, it is expensive to run the circuit with full entanglement on four qubits.
However, we have seen that the use of linear entanglement produces equally good results, so we will proceed with the linear entanglement design \fig{monras_qka_linear_hw_circuit}. The result is presented in
\fig{monras_qka_linear_hw_probabilities}.
Hardware results~\cite{monraslinearexperiment} clearly capture the pattern.
However, the current level of hardware noise still seems too high for longer sequences.

The hardware run was executed using \emph{Sampler} primitive on \emph{Qiskit Runtime} with \emph{optimization\_level}=$2$
\cite{qiskitConfigureerrorsuppression,qiskitconfigureerrormitigation}.

\subsection{Classic Market Model with Four Hidden States}

The process is defined by transition and emission matrices in \tab{hmm_example}.

We will simulate this process using 2 qubits: one for the principal system and one for the environment.
We will use the maximally mixed state as the initial state.
For the main transition unitary, we will use RealAmplitudes, EfficientSU2 ($R_y,R_z$) and EfficientSU2 ($R_z,R_x$) ansatz circuits.
Clearly, for the circuit of only 2 qubits linear and full entanglement produce the same result.

Once we train the model following Algorithm~\ref{alg:train_hqmm}, we get the following costs associated with the ansatz circuits
(Table~\ref{tab:market_results_comparison}):

\begin{table}[!ht]
    \centering
    \caption{Optimized cost for different ansatz circuits}
    \label{tab:market_results_comparison}
    \begin{tabular}{|l|r|}
        \hline
        Ansatz &  Cost \\
        \hline
        EfficientSU2 ($R_z,R_x$)    & 0.00030 \\
        RealAmplitudes & 0.00030  \\
        EfficientSU2 ($R_y,R_z$)    & 0.00036  \\
        \hline
    \end{tabular}
\end{table}

The results are very close, so we can pick the simplest circuit to continue.
In this case, we choose RealAmplitudes since it has only 2 parameters.
The circuit for 1 step with optimized parameters is shown in \fig{market_realamplitudes_circuit}.

\begin{figure}[!ht]
    \hspace*{-1.75cm}\includegraphics[scale=0.4]{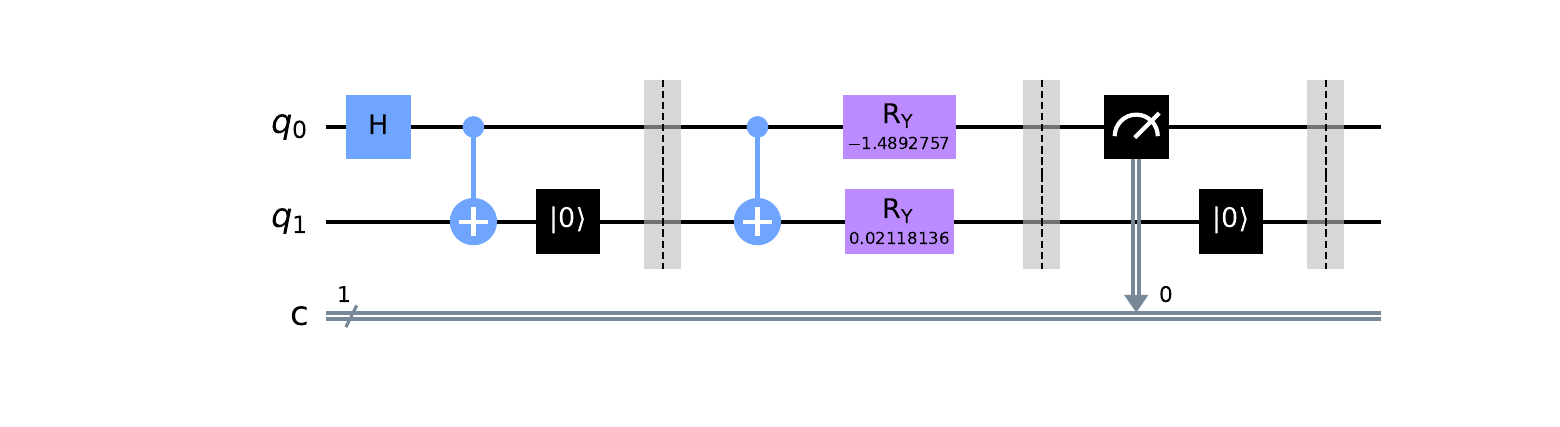}
    \caption{Circuit for 1 step with optimized parameters}
    \label{fig:market_realamplitudes_circuit}
\end{figure}
\begin{figure}[!ht]
    \centering
    \includegraphics[scale=0.3]{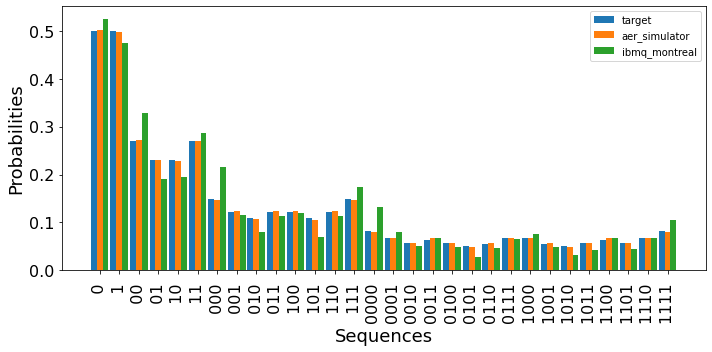}
    \caption{Predicted and observed probabilities}
    \label{fig:market_real_amplitudes_probabilities}
\end{figure}

Comparison with the hardware run on ibmq\_montreal using Qiskit Runtime and Sampler primitive
~\cite{marketrealamplexperiment} is shown in \fig{market_real_amplitudes_probabilities}.
The hardware run was executed using \emph{Sampler} primitive on \emph{Qiskit Runtime} with $optimization\_level=2$
~\cite{qiskitConfigureerrorsuppression,qiskitconfigureerrormitigation}.

The training has converged to the final solution very quickly. In fact, it is easy to see that even the RealAmplitudes ansatz is too expressive for this stochastic process and we can further simplify it by removing $R_Y$ rotation gate on $q_1$.

\newpage

\section{\label{section:conclusions_outlook}Conclusions and Outlook}

In this paper, we utilized the theory of open quantum systems to study complexity, unitary implementation, and learning from examples of a class of quantum stochastic generators known as Quantum Hidden Markov Models. We have demonstrated that the integration of unitary dynamics of a principal (or ``state'') quantum system, entangled with an emission system (``environment'') where observations are generated through orthogonal measurement, presents a state-efficient model of finite-order stochastic languages. The Markovian behavior of the model's state evolution is guaranteed by the direction of the information flow from the principal system to the environment. We note that there are other quantum architectures that generate joint evolution of a Markovian state process and a dependent observable process. It is especially important to characterize the processes generated by sequential measurements of cluster state quantum systems \cite{Briegel2001}, \cite{monras2011hidden}. Since these systems have fewer free parameters, the general entanglement architecture and the parameters of the measurement basis, we would expect efficient learning algorithms. 

Another important research opportunity is to apply the approach developed in the article to devices and processes with non-Markovian behavior. In this case, we need to establish an information flow, or input, from the environment towards the state system \cite{Breuer2010}. This will allow to model higher-order Hidden Markov Models, finite-state transducers, and generators with attention mechanisms.

We proposed three variational ansatz circuits to be used as a starting point to model classical sequence data and
tested on two use cases with 4 hidden states and 2 and 4 observed outcomes.

\section{Acknowledgements}

We thank Vaibhaw Kumar, Jae-Eun Park, Laura Schleeper, and Rukhsan UI Haq for helpful discussions at the start of this work. We would like to thank John Watrous for his comments regarding quantum channels.

The views expressed in this article are those of the authors and do not represent the views of Wells Fargo. This
article is for informational purposes only. Nothing contained in this article should be construed as investment advice.
Wells Fargo makes no express or implied warranties and expressly disclaims all legal, tax, and accounting implications
related to this article.\\

\newpage
\newpage
\bibliographystyle{ieeetr}
\bibliography{main}


\end{document}